\documentclass[11pt]{article}

\usepackage{epsf}
\usepackage{amsmath}
\usepackage{amssymb}
\usepackage{amsfonts}
\usepackage{graphicx}
\usepackage{rotating}
\usepackage{fullpage}
\usepackage{bm}
\usepackage{enumerate}
\usepackage[shortlabels]{enumitem}
\usepackage{fullpage}
\usepackage[usenames,dvipsnames]{xcolor}
\usepackage{framed}
\definecolor{shadecolor}{rgb}{0.9,0.9,0.9}

\usepackage{accents}
\newcommand{\dbtilde}[1]{\accentset{\approx}{#1}}

\newcommand{\n}{m}

\newtheorem{theorem}{Theorem}[section]
\newtheorem{proposition}[theorem]{Proposition}
\newtheorem{corollary}[theorem]{Corollary}
\newtheorem{lemma}[theorem]{Lemma}

\newtheorem{definition}{Definition}[section]

\newenvironment{proof}{\noindent{\bf Proof.}}{\bigskip} 
\newcommand{\qed}{\hfill$\rule{1.2ex}{1.2ex}$\bigskip}

\newcommand{\remove}[1]{}
\newcommand{\cg}[1]{\textcolor{black}{#1}}
\newcommand{\cgr}[1]{\textcolor{black}{#1}}

\begin{document}

\title{(In)Existence of Equilibria for 2-Players, 2-Values Games with Concave Valuations\thanks{This work is partially supported by the German Research Foundation (DFG)
		within the Collaborative Research Centre
		``On-the-Fly-Computing'' (SFB 901),
		and by funds for the promotion of research
		at the University of Cyprus.
}} 


\author{Chryssis Georgiou\thanks{Dept. of Computer Science, University of Cyprus, Nicosia, Cyprus; chryssis@cs.ucy.ac.cy}\\
	\and Marios Mavronicolas\thanks{Dept. of Computer Science, University of Cyprus, Nicosia, Cyprus; mavronic@cs.ucy.ac.cy}  
	\and Burkhard Monien\thanks{University of Paderborn, Paderborn, Germany; bm@upb.de}
}




\maketitle



\begin{abstract}
We consider {\it 2-players}, {\it 2-values} minimization games where
the players' {\it costs}
take on two values, $a,b$, $a<b$. The players play mixed strategies and their costs are evaluated by {\em
unimodal valuations}. \cg{This broad class of valuations includes all concave, one-parameter 
functions $\mathsf{F}: [0,1]\rightarrow \mathbb{R}$
with a unique maximum point.} 
Our  main result \cg{is an impossibility result stating} that:
\begin{itemize}
	\item {If the maximum is obtained in $(0,1)$} and $\mathsf{F}\left(\frac{1}{2}\right)\ne b$, then there exists a 2-players, 2-values game without {\sf F}-equilibrium.
\end{itemize}
The counterexample game used for the impossibility result belongs to a new class of very sparse 2-players, 2-values bimatrix games which we call {\em normal games}.

\noindent In an attempt to investigate the remaining case  $\mathsf{F}\left(\frac{1}{2}\right) = b$,
we show that: 
\begin{itemize}
\item {Every normal, $n$-strategies game has an ${\mathsf{F}}$-equilibrium
when
${\mathsf{F}}\left( \frac{1}{2}
\right)                  
=
b$.
We present a linear time algorithm
	for computing such an equilibrium.}
\item {For 2-players, 2-values games with 3 strategies we have that
if $\mathsf{F}\left(\frac{1}{2}\right) \le b$, then
every 2-players, 2-values, 3-strategies game has an {\sf F}-equilibrium;
if $\mathsf{F}\left(\frac{1}{2}\right) > b$,
then there exists a normal  2-players, 2-values, 3-strategies game
without {\sf F}-equilibrium.}
\end{itemize}
{To the best of our knowledge, this work is the first to provide an (almost complete)
	answer on whether there is, for a given concave function {\sf F}, a counterexample game without
	{\sf F}-equilibrium.}
\end{abstract}

\section{Introduction}
\label{introduction}

\subsection{Motivation and Framework}
\label{motivaton}
\cg{Concave valuations are concave functions from probability distributions to reals.
	In general, concave functions find many applications in Science and Engineering.
	For example,  
	in Microeconomic Theory, production functions are usually
	assumed to be concave, resulting in diminishing returns to input factors; see, e.g.,~\cite[pp. 363--364]{PR15}.
	Also, their concavity makes concave valuations especially suitable to model {\it risk} -- see, 
	e.g.,~\cite{HL69,M52,S63}.} Furthermore, concave functions are utilized in transportation networks (see, e.g.,~\cite{SZ07}) and 
in wireless networks (see, e.g.,~\cite{CSS16}).

In this work, {\it we study $2$-players, $2$-values minimization games with
concave valuations for their existence of equilibria.} 
{\it Equilibria} are special states of the system where no entity has an
incentive to switch. A (minimization) game may or may not have an equilibrium under a concave valuation (cf., \cite{C90,MM16,MM17}). 

For our investigation on the conditions for existence of equilibria we consider {\em $2$-values games} with {\em unimodal valuations}:
\begin{itemize}
	\item {\it $2$-values games} are games where the players costs take on 
	two values $a$ and $b$, with $a<b$. We consider a subclass of counterexample $2$-values
	games, termed {\it normal games}, where 
	$(i)$ there is no $(a,a)$ entry in the cost bimatrix, $(ii)$ player 1 (the
	row player) has exactly one $a$ entry per column, and $(iii)$ player 2
	(the column player) has exactly one $a$ entry per row. So, normal games are very sparse $2$-values games offering a large degree of simplicity, which we use as a tool
	for deriving our main impossibility result.
	\item {\it Unimodal valuations} are valuations that can
	be expressed as a one-parameter concave function {\sf F} with a unique
	maximum point. Many valuations, like 
	$\mathsf{E}+\gamma\cdot\mathsf{Var}$ and 
	${\mathsf{E}} + {\mathsf{\gamma}} \cdot {\mathsf{SD}}$,
	are demonstrated to be unimodal;
	${\mathsf{E}}$, ${\mathsf{Var}},$ and ${\mathsf{SD}}$
	denote {\it expectation,} {\it variance}, and {\it standard deviation,}
	respectively, and $\gamma>0$ is the {\it risk factor}. 
\end{itemize}

The question about the existence of equilibria was first put forward by
Crawford~\cite{C90}, who gave an explicit counterexample game without an
equilibrium for some concave valuation. 
Crawford's game was used as a gadget in ${\mathcal NP}$-hardness proofs
of the equilibrium existence problem for the {\it explicit} valuations
$\mathsf{E}+\gamma\cdot\mathsf{Var}$ and 
${\mathsf{E}} + {\mathsf{\gamma}} \cdot {\mathsf{SD}}$~\cite{MM16}
and $\mathsf{CVaR}_{\alpha}$~\cite{MM17};  $\mathsf{CVaR}_{\alpha}$ denotes 
{\it conditional
	valuation-at risk} and
$\alpha\in (0,1)$ is the {\it confidence level}. (It should be noted that
the given reductions in \cite{MM16,MM17} yield non-sparse games.)
In contrast, our work deals with an {\em abstract unimodal} valuation {\sf V};
this is the first investigation into games without an equilibrium under a
general class of concave functions. Previous work has addressed specific
concave valuations such as 
$\mathsf{E}+\gamma\cdot\mathsf{Var}$~\cite{MM15,MM16}
and  $\mathsf{CVaR}_{\alpha}$~\cite{MM17}. 


Our results provide an almost complete answer to the equilibrium existence
problem for $2$-players, $2$-values games with concave valuations. In short, 
we define a sequence of normal games, and show 
that for every unimodal
valuation {\sf F} with $\mathsf{F}\left(\frac{1}{2}\right)\ne b$, 
there exists a game from the sequence that does not have  an {\sf F}-equilibrium (Section~\ref{sec:non-existence of equilibria}). For the remaining case $\mathsf{F}\left(\frac{1}{2}\right)= b$, we provide matching existence results (Section~\ref{sec:existence of equilibria}).

\subsection{Contribution in More Detail}
\label{contribution}

\cg{We proceed with a detailed exposition of
our results.}
%
%
\cg{Given a unimodal valuation ${\mathsf V}$ expressed by the concave
function {\sf F}, we consider the value of $x$, denoted by $x_0(\mathsf{F})$,
where {\sf F} takes the maximum value in $[0,1]$. We are interested in the impact
of $x_0(\mathsf{F})$ on the existence of {\sf V}-equilibria or, as we call them, {\sf F}-equilibria.
We obtain the following impossibility result as our main result:}
\begin{itemize}
	\item \cg{If $x_0(\mathsf{F}) > 0$ and $\mathsf{F}\left(\frac{\textstyle 1}{\textstyle 2}\right)\ne b$, then there is a normal 2-players, 2-values game without {\sf F}-equilibrium (Theorem~\ref{panorama}).}
	\end{itemize}

\noindent The impossibility result is shown using a family $\{\mathsf{C}_{\n}~:~\n \in\mathbb{N}\}$ of $(\n+1)$-strategies games;
$\mathsf{C}_{\n}$ is defined as the 1-row, 1-column extension of a Toepliz combinatorial
bimatrix, $\mathsf{D}_\n$, which defines an $\n$-strategies normal game.  Interestingly, the bimatirx game defined by $\mathsf{D}_\n$ has a uniform  {\sf F}-equilibrium, 
while in Theorem~\ref{hawaii grande} we prove {\it necessary} and
{\it sufficient} conditions for the uniqueness of this uniform equilibrium.
The impossibility result follows by proving, in Theorem~\ref{grande resort}, {\it necessary} and
{\it sufficient} conditions guaranteeing that 
$\mathsf{C}_{\n}$ has no
{\sf F}-equilibrium.
For example,
\begin{eqnarray*}
	{\mathsf{C}}_{2}
	& = & \left( \begin{array}{lll}
		(a,b) & (b,a) & (a,b) \\
		(b,a) & (a,b) & (b,b) \\
		(b,b) & (b,b) & (b,a)\\ 
	\end{array}
	\right)\, \text{has no}~\mathsf{F}\text{-equilibrium if and only if } 
	\mathsf{F}\left(\frac{1}{2}\right)>b.            
\end{eqnarray*}
To the best of our knowledge these
are the first examples of $2$-values games without equilibrium for a 
class of concave valuations. The Crawford game has $3$ values.

\noindent \cg{We compliment our main result by observing that if $x_0(\mathsf{F}) = 0$, then an {\sf F}-equilibrium exists for all 2-players, 2-values games and is ${\mathcal{PPAD}}$-hard to compute (Theorem~\ref{2 players 2 values ppad hard}).}

\noindent\cg{Our main result leaves open the remaining case  $\mathsf{F}\left(\frac{1}{2}\right) = b$, for which
we obtain matching existence results: }
\begin{itemize}
	\item Every 2-players, 2-values, $n$-strategies normal game has an {\sf F}-equilibrium when
	the unimodal valuation {\sf F} satisfies $\mathsf{F}\left(\frac{1}{2}\right)=b$, and it
	can be computed in $O(n)$ time (Corollary~\ref{cor:final-result}). This existential result is derived 
    in Theorem~\ref{thm:winning-pair}, where we determine efficiently a {\it winning pair}, which defines an 
	{\sf F}-equilibrium. The algorithm runs in time $O(n)$ and 
	the supports in the winning pair we determine have size $2$. 
     
      \item  Every 2-players, 2-values, 3-strategies game has an {\sf F}-equilibrium when the 
     unimodal valuation {\sf F} satisfies $\mathsf{F}\left(\frac{1}{2}\right) \le b$ (Theorem~\ref{2 values 3 strategies existence}).
     The proof of this result is purely combinatorial.
     This result nicely illuminates Theorem~\ref{grande resort}, which states that for $\mathsf{F}\left(\frac{1}{2}\right) > b$,
     the game $\mathsf{C}_2$, having $3$ strategies, has no {\sf F}-equilibrium.   
\end{itemize}


\subsection{Other Related Work}
\label{related work}

\cg{
	Win-lose games are the special case of $2$-values games where
	$a=0$ and $b=1$. Win-lose games are as powerful as $2$-values games 
	when {\sf E}-equilibria are considered. Abbott et al~\cite{AKV05} show that computing an {\sf E}-equilibrium for a given win-lose bimatrix game is $\mathcal {PPAD}$-complete.}

\cg{Two different models of sparse win-lose bimatrix games have been
	considered in~\cite{CDT06,CLR06}. In the model of Chen {\it et al.}~\cite{CDT06}, 
	for the row player each row, and for the column player each column
	of the payoff matrix
	has at most two 1's entries. In the model of Codenotti {\it et al.}~\cite{CLR06},
	there are at most two 1's per row and per column for both players.
	So the model in~\cite{CLR06} is inferior to the
	model in~\cite{CDT06}. Recall that in our model, there is at most
	one $a$ per column for the row player, and at most one $a$ per row
	for the column player. So our model is not comparable with the models 
	in~\cite{CDT06,CLR06}.} In both papers \cite{CDT06} and \cite{CLR06},
polynomial time algorithms are given for win-lose games from the corresponding
model of sparseness.

Chen {\it et al.}~\cite{CDT06} also show that approximating a Nash equilibrium is 
$\mathcal{PPAD}$-hard in the model where both matricies contain at most $10$
non-zero entries in each row and in each column. Liu and Sheng~\cite{LS18}
extend this result to win-lose games when both matrices contain $O(\log n)$
$1$'s entries in each row and in each column.
\cg{Bil\`o and Mavronicolas~\cite{BM12} prove that is ${\mathcal NP}$-complete
	to decide if a simultaneously win-lose and symmetric $2$-players game has an {\sf E}-equilibrium with certain properties.}    

\subsection{Paper Organization}
\cg{Section~\ref{sec:framework} presents the game-theoretic framework, together with some
simple properties regarding equilibria for 2-players, 2-values games, and the definition
of normal games. In Section~\ref{sec:standrd valuations}  we define unimodal valuations, 
we demonstrate that common risk-averse valuations are unimodal valuations, prove some
simple properties of unimodal valuations, and consider the case of $x_0(\mathsf{F})=0$.
Our main impossibility result is presented in Section~\ref{sec:non-existence of equilibria}.
Section~\ref{sec:existence of equilibria} presents matching existence results. 
We conclude in Section~\ref{sec:conclusion}.}

\section{Framework}
\label{sec:framework}

\subsection{Definitions and Preliminaries}

\noindent
For an integer $\nu \geq 2$,
an {\it $\nu$-players game}
${\mathsf{G}}$,
or {\it game,}
consists of
{\it (i)}
$\nu$ finite sets
$\left\{ S_{k}
\right\}_{k \in [\nu]}$
of {\it strategies,}
and
{\it (ii)}
$\nu$ {\it cost functions}
$\left\{ {\mathsf{\mu}}_{k}
\right\}_{k \in [\nu]}$,
each mapping
${\mathcal{S}} = \times_{k \in [\nu]}
              S_{k}$
to the reals.
So,
${\mathsf{\mu}}_{k}({\bf s})$
is the {\it cost}
of player $k \in [\nu]$
on the {\it profile}
${\bf s}
=
\langle s_{1},
           \ldots,
           s_{\nu}
\rangle$
of strategies,
one per player.
All costs
are non-negative.
Assume,
without loss of generality,
that for each player $k \in [\nu]$,
$S_{k} = \{ 0, \ldots, n-1 \}$,
with $n  \geq 2$.

A {\it mixed strategy}
for player $k \in [\nu]$
is a probability distribution $p_{k}$
on $S_{k}$;
the {\it support}
of player $k$
in $p_{k}$
is the set
$\sigma (p_{k})
=
\left\{ s_{k} \in S_{k}\
    \mid\
    p_{k}(s_{k}) > 0
\right\}$.
Denote as $p_{k}^{s_{k}}$
the {\it pure strategy}
of player $k$
choosing the strategy $s_{k}$
with probability $1$.
A {\it mixed profile}
is a tuple
${\bf p} = \langle p_{1},
         \ldots,
         p_{\nu}
     \rangle$
of $\nu$ 
mixed strategies,
one per player. 
The mixed profile ${\bf p}$
induces
probabilities
${\bf p}({\bf s})$
for each profile
${\bf s} \in {\mathcal{S}}$
with
${\bf p}({\bf s})
=
\prod_{k \in [\nu]}
  p_{k} (s_{k})$.
For a player $k \in [\nu]$,
the {\it partial profile}
${\bf s}_{-k}$
(resp.,
{\it partial mixed profile}
${\bf p}_{-k}$)
results by eliminating
the strategy $s_{k}$
(resp.,
the mixed strategy $p_{k}$)
from ${\bf s}$
(resp.,
from ${\bf p}$).

For a player $k \in [\nu]$,
a {\it valuation function,}
or {\it valuation} for short,
${\mathsf V}_{k}$
is a real-valued function,
yielding a value
${\mathsf V}_{k}({\bf p})$
to each mixed profile
${\bf p}$,
so that
in the special case
when ${\bf p}$ 
is a profile ${\bf s}$,
${\mathsf V}_{k}({\bf s})
=
{\mathsf{\mu}}_{k}({\bf s})$.
The {\it expectation valuation}
is defined as
${\mathsf{E}}_{k}({\bf p})
 = 
 \sum_{{\bf s}
       \in
       S}
   {\bf p}({\bf s})\,
   {\mathsf{\mu}}_{k}({\bf s})$,
with $k \in [\nu]$.
A {\it valuation}
${\mathsf V}
=
\langle {\mathsf V}_{1},
    \ldots,
    {\mathsf V}_{\nu}
\rangle$
is a tuple of valuations,
one per player;
${\mathsf{G}}^{{\mathsf{V}}}$
denotes ${\mathsf{G}}$
together with ${\mathsf{V}}$.
We shall view
each valuation ${\mathsf{V}}_{k}$,
with $k \in [\nu]$,
as a function of the mixed strategy $p_{k}$,
for a fixed partial mixed profile ${\bf p}_{-k}$.
The valuation ${\mathsf V}$
is {\it concave}\footnote{Recall that a function  $\mathsf{f}:\mathsf{D} \rightarrow \mathbb{R}$  on a convex set   $\mathsf{D}\subseteq \mathbb{R}^{m}$ is {\it concave} if for every pair of points   $x,y\in \mathsf{D},$ for all   $\mathsf{\lambda}\in[0,1],$   $\mathsf{f}(\mathsf{\lambda}y+(1-\mathsf{\lambda})x) \ge \mathsf{\lambda}\mathsf{f}(y)+(1-\mathsf{\lambda})\mathsf{f}(x)$.} if the following condition
holds
for every game ${\mathsf{G}}$:
For each player $k \in [\nu]$,
the valuation\
${\mathsf{V}}_{k} \left( p_{k},
                         {\bf p}_{-k}
                  \right)$
is concave 
in $p_{k}$,
for a fixed 
${\bf p}_{-k}$.

\sloppy{A  ${\mathsf{V}}_{k}$-best response is a (pure or mixed) strategy $p_k$ that minimizes
	$\mathsf{V}_k(p_k,{\bf p}_{-k})$ for a fixed $p_k$. In this work, we shall consider {\em minimization} games, where each player $k\in [\nu]$ is interested in minimizing ${\mathsf{V}}_{k}({\bf p})$,
	that is, each player seeks to minimize her cost.} 
 
The mixed profile ${\bf p}$
is a {\it ${\mathsf{V}}$-equilibrium}
if for each player $k \in [\nu]$,
the mixed strategy $p_{k}$
is a ${\mathsf{V}}_{k}$-best response
to ${\bf p}_{-k}$;
so,
no player could unilaterally deviate
to another mixed strategy
to reduce her cost. {In this case, we will say that
	player $k$ is {\it ${\mathsf{V}}$-happy} with  ${\bf p}$.}
{We will be referring to an {\it ${\mathsf{E}}$-equilibrium}
when the valuation is the expectation valuation (and {\it ${\mathsf{E}}$-happy} is
defined analogously).}
A {\it pure equilibrium}
is an equilibrium where all players play pure strategies.

\vspace{.5em}
\noindent
{We shall consider two properties
defined in reference to best-responses:}

\begin{definition}[{The Weak-Equilibrium-for-Expectation Property~\cite[Section 2.6]{MM15}}]
\label{optimal value property definition}
The valuation ${\mathsf{V}}$
has the {\it Weak-Equilibrium-for-Expectation} property
if the following condition holds
for every game ${\mathsf{G}}$:
For each player $k \in [\nu]$
and a partial mixed profile ${\bf p}_{-k}$,
if $p_{k}$
is a ${\mathsf{V}}_{k}$-best response to ${\bf p}_{-k}$,
then,
for each pair of strategies
$\ell, \ell' \in \sigma(p_{k})$,
${\mathsf{E}}_{k}(p_{k}^{\ell}, {\bf p}_{-k})
  =
  {\mathsf{E}}_{k}(p_{k}^{\ell'}, {\bf p}_{-k})$.          
\end{definition}

\noindent
We shall often abbreviate the
{\it Weak-Equilibrium-for-Expectation} property
as {\sf WEEP}.
We shall omit reference to the game ${\mathsf{G}}$
when immaterial.

\begin{definition}[{The Optimal-Value Property~\cite[Section 3.2]{MM16}}]
\label{optimal value definition}
The valuation ${\mathsf{V}}$
has the {\it Optimal-Value} property
if the following condition holds
for every game ${\mathsf{G}}:$
For each player $k \in [\nu]$,
and a partial mixed profile ${\bf p}_{-k}$,
if ${\widehat{p}}_{k}$
is a ${\mathsf{V}}_{k}$-best response to ${\bf p}_{-k}$,
then,
for any mixed strategy $q_{k}$
with ${\sigma}(q_{k})
          \subseteq
          {\sigma}({\widehat{p}}_{k})$,
${\mathsf{V}}_{k}(q_{k}, {\bf p}_{-k})
  =
  {\mathsf{V}}_{k}({\widehat{p}}_{k}, {\bf p}_{-k})$.          
\end{definition}

\remove{
\noindent
Recall the definitions of concavity and quisiconcavity~\cite{MM17}.
\noindent
A function   $\mathsf{f}:\mathsf{D} \rightarrow \mathbb{R}$  on a convex set   $\mathsf{D}\subseteq \mathbb{R}^{n}$ is {\it concave} if for every pair of points   $x,y\in \mathsf{D},$ for all   $\mathsf{\lambda}\in[0,1],$   $\mathsf{f}(\mathsf{\lambda}y+(1-\mathsf{\lambda})x) \ge \mathsf{\lambda}\mathsf{f}(y)+(1-\mathsf{\lambda})\mathsf{f}(x)$;   
$\mathsf{f}$  is {\it strictly concave} if the inequality is strict for all   $\mathsf{\lambda}\in (0,1)$. 
\noindent
Quasiconcavity generalizes concavity: A function   $\mathsf{f}:\mathsf{D} \rightarrow \mathbb{R}$ {\it quasiconcave} if for every pair of points 
 $x,y\in \mathsf{D},$ for all  $\mathsf{\lambda}\in(0,1),$
  $\mathsf{f}(\mathsf{\lambda}y+(1-\mathsf{\lambda})x) \ge
  \min\{\mathsf{f}(y),\mathsf{f}(x)\}$.
  A quasiconcave function need not be concave. Thus, these functions miss many significant properties of concave functions. For example, a quasiconcave function is not necessarily continuous, while a concave function on an open interval is also continuous on it. 
  A quasiconcave function $\mathsf{f}:\mathsf{D} \rightarrow \mathbb{R}$  is {\it strictly quasiconcave} if for every pair of points  $x,y\in \mathsf{D},$ $\mathsf{f}(x)\ne \mathsf{f}(y)$  implies that for all $\mathsf{\lambda}\in(0,1),$ $\mathsf{f}(\mathsf{\lambda}y+(1-\mathsf{\lambda})x) >
  \min\{\mathsf{f}(y),\mathsf{f}(x)\}$.
  Hence, the notion of strictness for strictly quasiconcave functions is weaker than for strictly concave functions. A notable property of a strictly quasiconcave function is a property of concave functions: a local maximum is also a global maximum.  
}

\noindent
We recall from~\cite[Proposition 3.1]{MM16}:

\begin{proposition}[{Concavity $\Rightarrow$ Optimal-Value}]
\label{concavity implies optimal value property}
Assume that
the valuation
${\mathsf{V}}_{k}(p_{k}, {\bf p}_{-k})$
is concave in $p_{k}$.
Then,
${\mathsf{V}}_{k}$
has the {\it Optimal-Value} property.         
\end{proposition}

%

\noindent
In this work we focus
on $2$-players games
with $\nu=2$,
where for each player $k \in [2]$,
$\overline{k}$ denotes the other player.
A $2$-players game ${\mathsf{G}}$
will be denoted as
$\left( {\mathsf{\alpha}}_{ij},
{\mathsf{\beta}}_{ij}
\right)_{0 \leq i, j < n}$,
where 
${\mathsf{\alpha}}_{ij} = {\mathsf{\mu}}_{1}(i, j)$
and
${\mathsf{\beta}}_{ij} = {\mathsf{\mu}}_{2}(i, j)$.\smallskip
           
\noindent Furthermore, we focus on {\emph{{$2$-values games}}}, in which
	there are only two cost values $a$ and $b$,
	with $a < b$;
	that is,
	$\{ {\mathsf{\mu}}_{k}({\bf s})
	\mid
	\mbox{$k \in [2]$ and
		${\bf s} \in {\mathcal{S}}$}
	\}
	=
	\{ a, b \}$.       
To exclude the trivial existence of a pure equilibrium, we consider games
where no $(a,a)$ entries 
exist in the bimatrix $\left( {\mathsf{\alpha}}_{ij},
{\mathsf{\beta}}_{ij}
\right)_{0 \leq i, j < n}$.

\vspace{.5em}
\noindent Now recall the definition of domination
among strategies in the 2-players game ${\mathsf{G}}$.
For a given mixed profile ${\bf p}$,
{\it strategy $\ell$ dominates strategy $\ell'$
	with respect to $\sigma (p_{2})$}
if the following conditions hold:
\begin{enumerate}
	
	\item[{\sf (1)}]
	${\mathsf{\mu}}_{1}(\ell, j)
	\leq
	{\mathsf{\mu}}_{1}(\ell', j)$
	for all strategies $j \in \sigma (p_{2})$.

	\item[{\sf (2)}]
	There is a strategy ${\widehat{j}} \in \sigma (p_{2})$
	such that
	${\mathsf{\mu}}_{1}(\ell,\widehat{j})
	<
	{\mathsf{\mu}}_{1}(\ell',\widehat{j})$.

\end{enumerate}

\noindent
We observe:

\begin{lemma}
	\label{domination and weep}
	Assume that the valuation ${\mathsf{V}}$
	has the {\sf WEEP} property.
	Then,
	for each player $k \in [2]$
	and a partial mixed profile ${\bf p}_{-k}$,
	if $p_{k}$ is a ${\mathsf{V}}_{k}$-best-response to ${\bf p}_{-k}$,
	then no strategy in $\sigma (p_{k})$
	dominates some strategy in $\sigma (p_{k})$
	with respect to $\sigma (p_{\overline{k}})$.
\end{lemma}

\subsection{Simple Properties}
\label{definitions and simple properties}

{We now prove two simple properties used later in the paper.} 
We begin with a sufficient condition
for a pure equilibrium
for a $2$-players, $2$-values game:

%
%

\remove{
\begin{center}
\fbox{
\begin{minipage}{6.7in}
\begin{definition}[\textcolor{red}{No-Pure Game}]
\textcolor{red}{A 2-players, 2-values game ${\mathsf{G}}$ is
{\em {\emph{\textbf{no-pure}}}}
if it fulfills:}
\begin{enumerate}

\item[{\sf (1)}]
\textcolor{red}{There is no $(a,a)$ entry.}
 
\item[{\sf (2)}]
\textcolor{red}{Every column includes both cost values $a$ and $b$
for the row player:
For every column $j$,
there are distinct rows $i, i^{\prime}$
with 
${\mathsf{\mu}}_{1}(i, j)
  \neq
  {\mathsf{\mu}}_{1}(i^{\prime}, j)$.}

\item[{\sf (3)}]
\textcolor{red}{Every row includes both cost values
$a$ and $b$
for the column player:
For every row $i$,
there are distinct columns $j, j^{\prime}$
with
${\mathsf{\mu}}_{2}(i, j)
  \neq
  {\mathsf{\mu}}_{2}(i, j^{\prime})$.}

\end{enumerate}
\end{definition}
\end{minipage}
}
\end{center}
}


\begin{lemma}
\label{lemma A}
Consider
a $2$-players, $2$-values game.
If there is
an ${\mathsf{E}}$-equilibrium
$\langle p_{1}, p_{2}\rangle$
with
$|\sigma (p_{1})| = 1$
or
$|\sigma (p_{2})| = 1$,
then there is also a pure equilibrium.
\end{lemma}

\begin{proof}
Without loss of generality,
assume that
$|\sigma (p_{1})| = 1$
with
$\sigma_{1} = \{ 0 \}$.
By the {\sf WEEP},
there is
$\alpha \in \{ a, b \}$
such that
${\mathsf{\mu}}_{2}(0, j) = \alpha$
for all $j \in \sigma (p_{2})$.
Since ${\bf p}$ is an ${\mathsf{E}}$-equilibrium,
this implies that
$\alpha \leq {\mathsf{\mu}}_{2}(0, j)$
for all strategies $j \in \{0,\ldots,n-1\}$.
Take now $\widehat{j} \in \sigma (p_{2})$
so that
${\mathsf{\mu}}_{1}(0, \widehat{j})
  =
  \min_{j \in \sigma_{2}}
     {\mathsf{\mu}}_{1}(0, j)$.
We prove that 
the profile $(0, \widehat{j})$
is a pure equilibrium:
\begin{itemize}

\item
\underline{Player $2$ is happy
with $(0, \widehat{j})$}:
We have to prove that
${\mathsf{\mu}}_{2}(0, \widehat{j})
  \leq
  {\mathsf{\mu}}_{2}(0, j)$
for all strategies $j \in \{0,\ldots,n-1\}$. 
	Since $\widehat{j} \in \sigma_{2}$
and ${\bf p}$ is an ${\mathsf{E}}$-equilibrium,
${\mathsf{\mu}}_{2}(0, \widehat{j})
  =
  {\mathsf{E}}_{2}(p_{1}^{0}, p_{2}^{\widehat{j}})
  \leq 
  {\mathsf{E}}_{2}(p_{1}^{0}, p_{2}^{j})
  =
  {\mathsf{\mu}}_{2}(0, j)$.

\item
\underline{Player $1$ is happy
with $(1, \widehat{j})$}:
We have to prove that
${\mathsf{\mu}}_{1}(0, \widehat{j})
  \leq
  {\mathsf{\mu}}_{1}(i, \widehat{j})$
for all strategies $i \in \{0,\ldots,n-1\}$.             
This is vacuous if ${\mathsf{\mu}}_{1}(0, \widehat{j}) = a$.
So assume that
${\mathsf{\mu}}_{1}(0, \widehat{j}) = b$.
By the definition of $\widehat{j}$,
it follows that
${\mathsf{\mu}}_{1}(0, j) = b$
for all $j \in \sigma (p_{2})$.
Hence,
${\mathsf{E}}_{1}(p_{1}^{0}, p_{2}) = b$.
Since $p_{1}^{0}$
is an ${\mathsf{E}}$-best response to $p_{2}$,
it follows that
${\mathsf{E}}_{1}(p_{1}^{i}, p_{2})
  \geq
  b$
for all strategies $i \in \{0,\ldots,n-1\}$.
This implies that
${\mathsf{\mu}}_{1}(i, j) = b$
for all pairs of strategies
$i \in \{0,\ldots,n-1\}$
and
$j \in \sigma (p_{2})$.
Since
$\widehat{j} \in \sigma (p_{2})$,
this implies that
${\mathsf{\mu}}_{1}(i, \widehat{j}) = b$
for all strategies $i \in \{0,\ldots,n-1\}$.
Hence,
${\mathsf{\mu}}_{1}(0, \widehat{j})
  =
  {\mathsf{\mu}}_{1}(i, \widehat{j})$
for all strategies $i \in \{0,\ldots,n-1\}$.       

\end{itemize}
The claim follows.\qed
\end{proof}

\noindent
We continue to prove a simple property of ${\mathsf{E}}$-equilibria:

\begin{lemma}
\label{lemma B hat}
Consider a $2$-players, $2$-values game,
with an ${\mathsf{E}}$-equilibrium
${\bf p}
  =
  \langle p_{1}, p_{2}\rangle$
with
$|\sigma (p_{2})| = 2$.
Define the mixed strategy
$\widetilde{p}_{2}$
for player $2$
with
$\sigma (\widetilde{p}_{2})
 =
 \sigma (p_{2})$
and
$\widetilde{p}_{2}(j)
  =
  \frac{\textstyle 1}
         {\textstyle 2}$  
for all $j \in \sigma (\widetilde{p}_{2})$.
Then,
$\langle p_{1}, \widetilde{p}_{2}\rangle$
is also an ${\mathsf{E}}$-equilibrium.
\end{lemma}

\begin{proof}
Without loss of generality,
take
$\sigma (p_{2})
  =
  \{ 0, 1 \}$.
To prove that
$\langle p_{1}, \widetilde{p}_{2}\rangle$
is an ${\mathsf{E}}$-equilibrium,
We consider each player separately: 
\begin{itemize}

\item
\underline{Player $2$ is ${\mathsf{E}}$-happy}:
By the {\sf WEEP} for
$(p_{1}, p_{2})$,
${\mathsf{E}}_{2}(p_{1}, p_{2}^{0})
  =
  {\mathsf{E}}_{2}(p_{1}, p_{2}^{1})$;
since $\langle p_{1}, p_{2}\rangle$
is an ${\mathsf{E}}$-equilibrium,
${\mathsf{E}}_{2}(p_{1}, p_{2}^{0})
  \leq
  {\mathsf{E}}_{2}(p_{1}, p_{2}^{\ell})$
for all strategies
$\ell \not\in \sigma (p_{2})$. 
Since $\sigma (p_{2})
           =
           \sigma(\widetilde{p}_{2})$,
the claim follows.

\item
\underline{Player $1$ is {${\mathsf{E}}$-happy}:}  
By the {\sf WEEP},
there is a number
$\alpha \in {\mathbb{R}}$
with 
${\mathsf{E}}_{1}(p_{1}^{i}, p_{2})
  =
  \alpha$
for all strategies $i \in \sigma (p_{1})$.   
By the linearity of Expectation,
there are only three possible cases:
\begin{enumerate}

\item[{\sf (1)}]
${\mathsf{\mu}}_{1}(i, 0)
  =
  {\mathsf{\mu}}_{1}(i, 1)
  =
  a$
for all $i \in \sigma (p_{1})$.

\item[{\sf (2)}]
${\mathsf{\mu}}_{1}(i, 0)
  =
  {\mathsf{\mu}}_{1}(i, 1)
  =
  b$
for all $i \in \sigma (p_{1})$.

\item[{\sf (3)}]
${\mathsf{\mu}}_{1}(i, 0)
  \neq
  {\mathsf{\mu}}_{1}(i, 1)$
for all $i \in \sigma (p_{1})$.
\end{enumerate} 
\noindent It follows that
{
\begin{eqnarray*}
           {\mathsf{E}}_{1}\left( p_{1}^{i}, \widetilde{p}_{2}
                                       \right)
& = & \left\{ \begin{array}{ll}
                       a\,                              & \mbox{if
                                                                        ${\mathsf{\mu}}_{1}(i, 0)
                                                                          =
                                                                          {\mathsf{\mu}}_{1}(i, 1)
                                                                          =
                                                                          a$
                                                                          for all $i \in \sigma (p_{1})$}                                                                                             \\
                       b\,                              & \mbox{if
                                                                        ${\mathsf{\mu}}_{1}(i, 0)
                                                                           =
                                                                           {\mathsf{\mu}}_{1}(i, 1)
                                                                           =
                                                                           b$
                                                                           for all $i \in \sigma (p_{1})$}                                                                                             \\
                      \frac{\textstyle a+b}
                              {\textstyle 2}     &  \mbox{if ${\mathsf{\mu}}_{1}(i, 0)
                                                                              \neq
                                                                              {\mathsf{\mu}}_{1}(i, 1)$
                                                                         for all $i \in \sigma (p_{1})$}                                                                                              \\                      
                    \end{array}
           \right.\, .
\end{eqnarray*}
}

\noindent So in all Cases
{\sf (1)}, {\sf (2)} and {\sf (3)},
${\mathsf{E}}_{1}(p_{1}^{i}, \widetilde{p}_{2})$
is constant over all strategies
$i \in \sigma (p_{1})$,
and the {\sf WEEP} holds,
so that
${\mathsf{E}}_{1}(p_{1}, \widetilde{p}_{2})
  =
  {\mathsf{E}}_{1}(p_{1}^{i}, \widetilde{p}_{2})$
for any strategy
$i \in \sigma (p_{1})$.
Thus,
it remains to prove that
player $1$ cannot improve
by switching to some strategy
$\ell \not\in \sigma(p_1)$.
Clearly,
{
\begin{eqnarray*}
           {\mathsf{E}}_{1}\left( p_{1}^{\ell}, \widetilde{p}_{2}
                                       \right)
& = & \left\{ \begin{array}{ll}
                       a\,                              & \mbox{if
                                                                        ${\mathsf{\mu}}_{1}(\ell, 0)
                                                                          =
                                                                          {\mathsf{\mu}}_{1}(\ell, 1)
                                                                          =
                                                                          a$}                                                                                                                                      \\
                       b\,                              & \mbox{if
                                                                        ${\mathsf{\mu}}_{1}(\ell, 0)
                                                                           =
                                                                           {\mathsf{\mu}}_{1}(\ell, 1)
                                                                           =
                                                                           b$}                                                                                                                                     \\
                      \frac{\textstyle a+b}
                              {\textstyle 2}     &  \mbox{if ${\mathsf{\mu}}_{1}(\ell, 0)
                                                                              \neq
                                                                              {\mathsf{\mu}}_{1}(\ell, 1)$}                                                                                          \\              
                    \end{array}
           \right.\, .           
\end{eqnarray*}
}

\noindent This is evident in Case {\sf (1)}.
Note that in Case {\sf (2)},
it must also hold that
${\mathsf{\mu}}_{1}(\ell, 0)
  =
  {\mathsf{\mu}}_{1}(\ell, 1)
  =
  b$
for all strategies $\ell \not\in \sigma(p_1)$
since $\langle p_{1}, p_{2}\rangle$
is an ${\mathsf{E}}$-equilibrium.
So player $1$ cannot improve
in this case either.
In Case {\sf (3)}, 
since $\langle p_{1}, p_{2}\rangle$
is an ${\mathsf{E}}$-equilibrium,
it must hold that 
for all strategies $\ell \not\in \sigma(p_1)$,
either
${\mathsf{\mu}}_{1}(\ell, 0)
  =
  {\mathsf{\mu}}_{1}(\ell, 1)
  =
  b$
or
${\mathsf{\mu}}_{1}(\ell, 0)
  \neq
  {\mathsf{\mu}}_{1}(\ell, 1)$.
So player~$1$ cannot improve
in this case either.

\end{itemize}
\noindent The claim follows.\qed 
\end{proof}

\subsection{Normal Games}
\label{normal games}

\noindent
We now introduce a restriction of 2-players, 2-values games to very sparse games that we call {\em normal games}.


\begin{center}
\fbox{
\begin{minipage}{6in}
\begin{definition}[{Normal Game}]
\label{normal game}
A 2-players, 2-values game ${\mathsf{G}}$
is {\em {\emph{\textbf{normal}}}}
if it fulfills:
\begin{enumerate}
\item[{\sf (1)}]
There is no $(a, a)$ entry
in the bimatrix of ${\mathsf{G}}$.
\item[{\sf (2)}]
Player $1$ (the {\it row player})
has exactly one $a$ entry
per column.
\item[{\sf (3)}]
Player $2$ (the {\it column player})
has exactly one $a$ entry
per row.
\end{enumerate}
\end{definition}
\end{minipage}
}
\end{center}

\noindent Note that the definition of a normal game
is symmetric with respect to the two players. {Also note that excluding $(a,a)$ entries, excludes the trivial existence of pure equilibria.}
We prove:

\begin{lemma}
\label{no equilibrium for normal game}
Consider the normal game 
${\mathsf{G}}^{{\mathsf{V}}}$,
where ${\mathsf{V}}$
has the {\sf WEEP}. 
Then,
${\mathsf{G}}^{{\mathsf{V}}}$
has no ${\mathsf{V}}$-equilibrium $\langle p_{1}, p_{2}\rangle$
with $|\sigma (p_{1})| = 1$
or with
$|\sigma (p_{2})| = 1$. 

\end{lemma}

\begin{proof}
Assume, by way of contradiction,
that ${\mathsf{G}}$
has a ${\mathsf{V}}$-equilibrium
$\langle p_{1}, p_{2}\rangle$
with
$|\sigma (p_{1})| = 1$.
(By the symmetry in the definition of a normal game,
this assumption is with no loss of generality.)
Let $\sigma(p_1)
  =
  \{ i \}$.
We proceed by case analysis.
Assume first that $|\sigma (p_{2})|=1$.
Let $\sigma (p_{2})
  =
  \{ j \}$.  
By Condition {\sf (1)}
for a normal game,
either
${\mathsf{\mu}}_{1}(i,j) = b$
or
${\mathsf{\mu}}_{2}(i,j) = b$.
If
${\mathsf{\mu}}_{1}(i,j) = b$,
then,
by Condition {\sf (2)}
for a normal game,
there is a strategy ${\widehat{i}}$
with
${\mathsf{\mu}}_{1}({\widehat{i}}, j) = a$.
Hence,
player $1$ improves by switching to strategy ${\widehat{i}}$.
A contradiction.
The case where
${\mathsf{\mu}}_{2}(i, j) = b$
is handled identically,
using Condition {\sf (3)} for a normal game.

Assume now that $|\sigma (p_{2})| > 1$.
By the {\sf WEEP},
there is an  $\alpha \in \{ a, b \}$
with 
${\mathsf{\mu}}_{2}(i, j)
  =
  \alpha$
for all $j \in \sigma (p_{2})$.
Since $|\sigma (p_{2})| > 1$,
Condition {\sf (3)} for a normal game
implies that
$\alpha = b$.  
Again by Condition {\sf (3)}
for a normal game,
there is a strategy 
${\widehat{j}} \in \sigma (p_{2})$
with
${\mathsf{\mu}}_{2}(i, {\widehat{j}}) = a$.
Hence,
player $2$ improves by switching to strategy ${\widehat{j}}$.
A contradiction.  \qed
\end{proof}

\noindent Note that Lemma~\ref{no equilibrium for normal game} implies 
that a normal game has no pure equilibrium and that $|\sigma(p_k)|\ge 2,~k\in [2]$,
for any {\sf V}-equilibrium $\langle p_{1}, p_{2}\rangle$.

\section{{Unimodal Valuations}}
\label{sec:standrd valuations}

\subsection{{Definitions}}
\label{unimodal valuations definition}

\noindent
For each player $k \in [2]$, define
{
{
\begin{eqnarray*}
x_{k}(p_{1}, p_{2}) & :=  & \sum_{(s_{1}, s_{2}) 
                                                         \in
                                                         S_{1} \times S_{2}
                                                        \mid
                                                       {\mathsf{\mu}}_{k}(s_{1}, s_{2})
                                                        = a}
                                               p_{1}(s_{1})
                                               \cdot
                                               p_{2}(s_{2})\, .
\end{eqnarray*}              
}
}  
{Note  that $x_{k}$ is linear in each of its arguments:
for each $\lambda \in [0,1]$,}
{
{
\begin{eqnarray*}
         x_{k}(\lambda \cdot p_{k}^{\prime} + (1-\lambda) \cdot p_{k}^{\prime\prime}, 
                   p_{\bar{k}})
&   = \lambda \cdot x_{k}(p_{k}^{\prime}, p_{\bar{k}})
         +
         (1-\lambda) \cdot x_{k}(p_{k}^{\prime\prime}, p_{\bar{k}})\, .
\end{eqnarray*}
}
}
\noindent
{We proceed to define:}

\begin{center}
\fbox{
\begin{minipage}{6in}
\begin{definition}[One-Parameter Valuation]
\label{onedimensional valuation definition}
For a $2$-players, $2$-values game,
a valuation ${\mathsf{V}}$
is {\em {\emph{\textbf{one-parameter}}}}
if for each mixed profile
$\langle p_{1}, p_{2}\rangle$,
for each player $k \in [2]$,
${\mathsf{V}}_{k}(p_{1}, p_{2})$
can be written as
$${\mathsf{V}}_{k}(p_{1}, p_{2})
  = 
  {\mathsf{F}}(x_{k}(p_{1}, p_{2}))\, ,$$

\noindent for some function
${\mathsf{F}}: [0, 1] \rightarrow {\mathbb{R}}$,
with ${\mathsf{F}}(0) = b$
and
${\mathsf{F}}(1) = a$.
\end{definition}
\end{minipage}
}
\end{center}

\noindent {We observe:}

\begin{lemma}
\label{painfully trivial}
{For a $2$-players, $2$-values game,
a one-parameter valuation ${\mathsf{V}}$
is concave
if and only if
${\mathsf{F}}$ is concave.}
\end{lemma}

\begin{proof}
{Consider the mixed profile
$\langle p_{k}, p_{\overline{k}}\rangle
  :=
  (\lambda \cdot p_{k}^{\prime} + (1-\lambda) p_{k}^{\prime\prime}, p_{\overline{k}})$,
with $\lambda \in [0, 1]$.}
{Recall the linearity
of $x_{k}$ in $p_{k}$,
for each player $k \in [2]$.}
{Since ${\mathsf{V}}$
is a one-parameter valuation,
for $k=1$,
we get that}
{
{
	\begin{eqnarray*}
		{\mathsf{F}}(x_{1}(p_{1}, p_{2}))
		& \geq & \lambda \cdot {\mathsf{F}}(x_{1}(p_{1}^{\prime}, p_{2}))
		+
		(1-\lambda) \cdot {\mathsf{F}}(x_{1}(p_{1}^{\prime\prime}, p_{2}))
	\end{eqnarray*}
}
}
{if and only if}
{
{
	\begin{eqnarray*}
		{\mathsf{V}}_{1}(p_{1}, p_{2})
		& \geq & \lambda \cdot {\mathsf{V}}_{1}(p_{1}^{\prime}, p_{2})
		+
		(1-\lambda) \cdot {\mathsf{V}}_{1}(p_{1}^{\prime\prime}, p_{2})\, .
	\end{eqnarray*}
}
}
\noindent
{Hence, 
${\mathsf{V}}_{1}$ is concave
if and only if
${\mathsf{F}}$ is concave. (The case for $k=2$ is symmetric.)}\qed
\end{proof}

\noindent
{We now introduce a restriction of one-parameter valuations.}

\begin{center}
\fbox{
\begin{minipage}{6in}
\begin{definition}[{Unimodal Valuation}]
\label{unimodal valuation definition}
{In the 2-values case,
a one-parameter valuation ${\mathsf{V}}$ is {\em {\emph{\textbf{unimodal}}}}
if ${\mathsf{F}}: [0, 1] \rightarrow {\mathbb{R}}$
is a concave function
{with a unique maximum point.}
}  
\end{definition}
\end{minipage}
}
\end{center}

\noindent
{By Lemma~\ref{painfully trivial},
a unimodal valuation ${\mathsf{V}}$
is concave.}
{\it We shall often identify
the unimodal valuation ${\mathsf{V}}$
with the function ${\mathsf{F}}$
and refer to a ${\mathsf{V}}$-equilibrium 
for a unimodal valuation ${\mathsf{V}}$
as an ${\mathsf{F}}$-equilibrium.}\vspace{1em}

\noindent {Finally, we define $x_0 = x_0(\mathsf{F})$ to be
the value of $x$ where $\mathsf{F}(x)$ takes
its maximal value in $[0,1]$. 
}

\subsection{{Examples}}

\noindent {Recall the {\it Expectation} ${\mathsf{E}}$,
{\it Variance} ${\mathsf{Var}}$
and
{\it Standard Deviation} ${\mathsf{SD}}$
valuations.  
We shall consider the valuations
$
  \mbox{{\sf EVar}}^{{\mathsf{\gamma}}}
  =
  {\mathsf{E}} + {\mathsf{\gamma}} \cdot {\mathsf{Var}}$
and
$
  \mbox{{\sf ESD}}^{{\mathsf{\gamma}}}
  =
  {\mathsf{E}} + {\mathsf{\gamma}} \cdot {\mathsf{SD}}$,
with ${\mathsf{\gamma}} > 0$.}
{Both ${\mathsf{EVar}}^{{\mathsf{\gamma}}}$ and ${\mathsf{ESD}}^{{\mathsf{\gamma}}}$
are concave functions
as the sums of two concave functions;
hence, 
they have the {\sf WEEP}
and the {\it Optimal-Value} property (Proposition~\ref{concavity implies optimal value property}).
We first derive formulas
for 
${\mathsf{E}}$, ${\mathsf{Var}}$ and ${\mathsf{SD}}$
in terms of $x$.
Fix a mixed profile
$\langle p_{1}, p_{2}\rangle$.
Then,} 
{
{
\begin{eqnarray*}
          {\mathsf{E}}_{1}(p_{1}, p_{2})
 & = & \sum_{(s_{1}, s_{2})
                      \in
                      S_{1} \times S_{2}}
             p_{1}(s_{1})
             \cdot
             p_{2}(s_{2})
             \cdot         
             {\mathsf{\mu}}_{1}(s_{1}, s_{2})                                                                                          \\
& = & a
          \cdot
          \sum_{(s_{1}, s_{2})
                     \in 
                     S_{1} \times S_{2}
                     \mid
                     {\mathsf{\mu}}_{1}(s_{1}, s_{2})
                     = a}
             p_{1}(s_{1})
             \cdot
             p_{2}(s_{2})
           \\
          & + & b
          \cdot
          \sum_{(s_{1}, s_{2})
                     \in 
                     S_{1} \times S_{2}
                     \mid
                     {\mathsf{\mu}}_{1}(s_{1}, s_{2})
                     = b}
             p_{1}(s_{1})
             \cdot
             p_{2}(s_{2})                                                                                                                            \\
& = & (a - b) \cdot x_1(p_1,p_2)
          +
          b\, .                     
\end{eqnarray*}
}
}
\noindent
{So
${\mathsf{E}}$
is a one-parameter valuation
with
${\mathsf{F}}(x) = (a-b) \cdot x + b$.}
{Note that
${\mathsf{F}}(x)$
is strictly monotone decreasing in $x$
for $x \in [0, 1]$.}
{For Variance,
we apply the ``mean of square minus square of mean'' formula
${\mathsf{Var}}(x) =
   {\mathsf{E}}(x^{2}) - ({\mathsf{E}}(x))^{2}$
for a random variable $x$,
to derive}
{
{
\begin{eqnarray*}
          {\mathsf{Var}}_{1}(p_{1}, p_{2})
& = & (a^{2} - b^{2}) \cdot x_1(p_1,p_2) + b^{2}
          -
          ((a-b)x_1(p_1,p_2) + b)^{2}                                                                                         \\
& = & (a^{2}-b^{2})
          \cdot
          x_1(p_1,p_2)
          +
          b^{2}
          -
          (a-b)^{2} x_1^{2}(p_1,p_2) \\
         & -&
          2 (a-b) \cdot b \cdot x_1(p_1,p_2)
          -
          b^{2}                                                                                                        \\
& = & x_1(p_1,p_2) (1-x_1(p_1,p_2)) (b-a)^{2}\, ,
\end{eqnarray*}
}
}
{which implies}
{
{
\begin{eqnarray*}
          {\mathsf{SD}}_{1}(p_{1}, p_{2})\ \
          :=\ \
          \sqrt{{\mathsf{Var}}_{1}(p_{1}, p_{2})}
& = & (b-a) \sqrt{x_1(p_1,p_2) (1-x_1(p_1,p_2))}\, .
\end{eqnarray*}
}
}

{\color{black}
\noindent Recently, {\it Conditional Value-at-Risk}~\cite{RU02} 
became very popular. $\mathsf{CVaR}_{\alpha}$, $\alpha\in (0,1)$ being
the confidence level, is recognized as a model of risk in volatile economic circumstances. 
For a discrete random variable ${\bf q}$ taking on values $0\le v_1 < v_2 < \ldots < v_{\ell}$
with probabilities $q_j$, $1\le j \le \ell$, Value-at-Risk $\mathsf{VaR}_{\alpha}$ and 
Conditional Value-at-Risk $\mathsf{CVaR}_{\alpha}$ are defined by 
$$\mathsf{VaR}_{\alpha}({\bf q}) = \min\left\{v_{\kappa}~:~\sum_{j=1}^{\kappa} q_{j}\ge \alpha\right\}$$ 
and 
$$\mathsf{CVaR}_{\alpha}({\bf q}) = \frac{1}{1-\alpha}\left[\left(\sum_{\kappa:v_{\kappa}\le \mathsf{VaR}_{\alpha}({\bf q})} q_{\kappa}-\alpha\right)\cdot \mathsf{VaR}_{\alpha}({\bf q})
+ \sum_{\kappa:v_{\kappa}> \mathsf{VaR}_{\alpha}({\bf q})} q_{\kappa}\cdot v_{\kappa}\right].$$
}

\subsubsection{{The Valuation ${\mathsf{EVar}}^{{\mathsf{\gamma}}}$}}

{
${\mathsf{EVar}}^{{\mathsf{\gamma}}}$ 
is a one-parameter valuation
represented by the concave function}
{
{
\begin{eqnarray*}
           {\mathsf{F}}_{1}^{{\mathsf{\gamma}}}(x)
& := & a \cdot x
          +
           b \cdot (1-x)
          +
          {\mathsf{\gamma}}
          \cdot
          (b-a)^{2}
          \cdot
          x \cdot (1-x)\, ,
\end{eqnarray*}
}
}
{with ${\mathsf{F}}_{1}^{{\mathsf{\gamma}}}(0) = b$
and ${\mathsf{F}}_{1}^{{\mathsf{\gamma}}}(1) = a$.
}
{Note that}
{
{
\begin{eqnarray*}
      ({\mathsf{F}}_{1}^{{\mathsf{\gamma}}}(x))^{\prime}
& = & (a-b)
      +
      {\mathsf{\gamma}}
      \cdot
      (1-2x)
      \cdot
      (b-a)^{2}\, ,
\end{eqnarray*}
}
}
{with} 
{
{
\begin{eqnarray*}
      ({\mathsf{F}}_{1}^{{\mathsf{\gamma}}}(x))^{\prime}_{x = x_{0}}\ \
      =\ \
      0
& \Longleftrightarrow & x_{0}\ \
      =\ \
      \frac{\textstyle 1}
           {\textstyle 2}
      \left( 1 - \frac{\textstyle 1}
                            {\textstyle {\mathsf{\gamma}} \cdot (b-a)}
      \right)\, .
\end{eqnarray*}
}
}
\noindent
{Since
$({\mathsf{F}}_{1}^{{\mathsf{\gamma}}}(x))^{\prime\prime}
 =
 -2 {\mathsf{\gamma}} \cdot (b-a)^{2}
 <
 0$
for all $x \in [0, 1]$, 
it follows that
${\mathsf{F}}_{1}^{{\mathsf{\gamma}}}(x)$
has a local maximum at~$x = x_{0}$.} 

\remove{
{
\begin{eqnarray*}
      {\mathsf{F}}_{1}^{{\mathsf{\gamma}}}(x_{0})
& = & b +
          (a-b)
          \cdot
          \frac{\textstyle 1}
                  {\textstyle 2}
          \left( 1 -
                   \frac{\textstyle 1}
                          {\textstyle {\mathsf{\gamma}} (b-a)}
     \right)
      +
      {\mathsf{\gamma}}
      \cdot
      (b-a)^{2}
      \cdot
      \underbrace{\frac{\textstyle 1}
                                  {\textstyle 2}
   \left( 1 - \frac{\textstyle 1}
                         {\textstyle {\mathsf{\gamma}} \cdot (b-a)}
  \right)
 \left( \frac{\textstyle 1}
                      {\textstyle 2}
             +
             \frac{\textstyle 1}
                     {\textstyle 2 }
             \cdot
             \frac{\textstyle 1}
                     {\textstyle {\mathsf{\gamma}} \cdot (b-a)}
      \right)}_{= \frac{\textstyle 1}
                                 {\textstyle 4}
                     \left( {\textstyle 1} -
                              \frac{\textstyle 1}
                                      {\textstyle {\mathsf{\gamma}} \cdot (b-a)^{2}}
                     \right)}                                                                                \\
& = & \frac{\textstyle 1}
                  {\textstyle 2}
          (a+b)
          +
          \frac{\textstyle 1}
                  {\textstyle 2 {\mathsf{\gamma}}}
      +
      \frac{\textstyle 1}
           {\textstyle 4}
      {\mathsf{\gamma}}
      \cdot     
      (b-a)^{2}
      -
      \frac{\textstyle 1}
              {\textstyle 4 {\mathsf{\gamma}}}                                          \\
& = &  \frac{\textstyle 1}
                  {\textstyle 2}
          (a+b)
          +
           \frac{\textstyle 1}
                   {\textstyle 4 {\mathsf{\gamma}}}               
          +
          \frac{\textstyle 1}
                  {\textstyle 4}
          {\mathsf{\gamma}}
          \cdot
          (b-a)^{2}\, .
\end{eqnarray*}
}
}
\noindent
{To determine the monotonicity properties
of ${\mathsf{F}}_{1}^{{\mathsf{\gamma}}}(x)$,
we distinguish two cases:}
\begin{itemize}

\item
{
For ${\mathsf{\gamma}} \cdot (b-a) \leq 1$,
${\mathsf{F}}_{1}^{{\mathsf{\gamma}}}(x)$
decreases strictly monotone
for $x \in [0, 1]$.}

\item
{
For ${\mathsf{\gamma}} \cdot (b-a) > 1$,
$0 < x_{0} < \frac{\textstyle 1}
                               {\textstyle 2}$,
and ${\mathsf{F}}_{1}^{{\mathsf{\gamma}}}(x)$
increases strictly monotone
for $x \in [0, x_{0}]$
and decreases strictly monotone for
$x \in [x_{0}, 1]$.
}

\end{itemize}

\noindent
So,
${\mathsf{F}}_{1}^{{\mathsf{\gamma}}}$ is a unimodal valuation. Further, note that 
{
${\mathsf{F}}_{1}^{{\mathsf{\gamma}}}\left( \frac{\textstyle 1}
{\textstyle m}
\right)
=
\frac{\textstyle a}
{\textstyle m}
+
\frac{\textstyle m-1}
{\textstyle m}\cdot b
+
{\mathsf{\gamma}}
\cdot
(b-a)^{2}
\cdot
\frac{\textstyle m-1}
{\textstyle m}$.          
Thus,
${\mathsf{F}}_{1}^{{\mathsf{\gamma}}}\left( \frac{\textstyle 1}
{\textstyle m}
\right)
=
b$
if and only if
${\mathsf{\gamma}}
\cdot
(b-a)
=
\frac{\textstyle m-1}
{\textstyle m}$, for $m\in\mathbb{N}$.
}

	\remove{
${\mathsf{F}}_{1}^{{\mathsf{\gamma}}}\left( \frac{\textstyle 1}
                                                                                     {\textstyle 2}
                                                                   \right)
  =
  \frac{\textstyle a+b}
          {\textstyle 2}
  +
  {\mathsf{\gamma}}
  \cdot
  (b-a)^{2}
  \cdot
  \frac{\textstyle 1}
         {\textstyle 4}$.          
Thus,
${\mathsf{F}}_{1}^{{\mathsf{\gamma}}}\left( \frac{\textstyle 1}
                                                                                     {\textstyle 2}
                                                                   \right)
  >
  b$
if and only if
${\mathsf{\gamma}}
  \cdot
  (b-a)
  >
  2$.}

\subsubsection{{The Valuation ${\mathsf{ESD}}^{\mathsf{\gamma}}$}}

{${\mathsf{ESD}}^{{\mathsf{\gamma}}}$
is a one-parameter valuation
represented by the concave function}
{
{
\begin{eqnarray*}
      {\mathsf{F}}_{2}^{{\mathsf{\gamma}}}(x)
& := & a \cdot x
      +
      b \cdot (1-x)
      +
      {\mathsf{\gamma}}
      \cdot
      (b-a)
      \cdot
      \sqrt{x \cdot (1-x)}\, ,
\end{eqnarray*}
}
}
{with ${\mathsf{F}}_{2}^{{\mathsf{\gamma}}}(0) = b$
and ${\mathsf{F}}_{2}^{{\mathsf{\gamma}}}(1) = a$.}
{We have}
{
{
\begin{eqnarray*}
          ({\mathsf{F}}_{2}^{{\mathsf{\gamma}}}(x))^{\prime}
& = & (a-b)
          +
          {\mathsf{\gamma}}
          \cdot
          \frac{1-2x}
                  {2 \sqrt{x \cdot (1-x)}}
          \cdot
          (b-a)\, .
\end{eqnarray*}
}
}
{Note that
$({\mathsf{F}}_{2}^{{\mathsf{\gamma}}}(x))^{\prime}
  <
  0$
for $x > \frac{\textstyle 1}
                      {\textstyle 2}$.  
Hence,
we seek $x_{0} < \frac{\textstyle 1}
                                      {\textstyle 2}$
with
$({\mathsf{F}}_{2}^{{\mathsf{\gamma}}}(x))^{\prime}_{x = x_{0}} =~0$.
So,}                                      
{
{
\begin{eqnarray*}
      ({\mathsf{F}}_{2}^{{\mathsf{\gamma}}}(x))_{x = x_{0}}\ \
      =\ \
      0
& \Longleftrightarrow &  {\mathsf{\gamma}}
                                        \cdot
                                        \frac{1-2x_{0}}
                                                {2 \sqrt{x_{0} \cdot (1-x_{0})}}\ \
                                        =\ \
                                        1                                                                                              \\
& \Longleftrightarrow & {\mathsf{\gamma}}
                                       \cdot
                                       (1-2x_{0})\ \
                                       =\ \
                                      2 \sqrt{x_{0} \cdot (1-x_{0})}\, .
\end{eqnarray*}
}
}
\noindent
{It follows that                                      
${\mathsf{\gamma}}^{2}
  \cdot
  (1 - 4x_{0} + 4x_{0}^{2})
  =
  4 (x_{0} - x_{0}^{2})$,}                               
\noindent
{yielding}
{
{
\begin{eqnarray*}
x_{0}
& = &
 \frac{\textstyle 4 ({\mathsf{\gamma}}^{2} + 1) - \sqrt{ 16 ({\mathsf{\gamma}}^{2} + 1)^{2} - 16 {\mathsf{\gamma}}^{2} ({\mathsf{\gamma}}^{2} + 1)}}
         {\textstyle 2 \cdot 4 ({\mathsf{\gamma}}^{2} + 1)}\ \ 
 =\ \
 \frac{\textstyle 1}
         {\textstyle 2}
 -
 \frac{\textstyle 1}
        {\textstyle 2 \sqrt{{\mathsf{\gamma}}^{2} + 1}}\, .
\end{eqnarray*}
}
}                
{Since ${\mathsf{\gamma}} > 0$,
it follows that
$0 < x_{0} < \frac{\textstyle 1}
                               {\textstyle 2}$.}
{Since                                           
$({\mathsf{F}}_{2}^{{\mathsf{\gamma}}}(x))^{\prime}_{x=0} = + \infty$
and
$({\mathsf{F}}_{2}^{{\mathsf{\gamma}}})^{\prime}_{x=1} < 0$,
it follows that
$x_{0}
  =
  \frac{\textstyle 1}
         {\textstyle 2}
  -
  \frac{\textstyle 1}
         {\textstyle 2 \sqrt{{\mathsf{\gamma}}^{2}}}$
is a local maximum of ${\mathsf{F}}_{2}^{\mathsf{\gamma}}(x)$,
which is unique.}  
{Hence, ${\mathsf{ESD}}^{\mathsf{\gamma}}$
is a unimodal valuation
for all values of ${\mathsf{\gamma}} > 0$
and $a, b$ with $a < b$.}                  
Further,
note that
{
	${\mathsf{F}}_{1}^{{\mathsf{\gamma}}}\left( \frac{\textstyle 1}
	{\textstyle m}
	\right)
	=
	\frac{\textstyle a}
	{\textstyle m}
	+
	\frac{\textstyle m-1}
	{\textstyle m}\cdot b
	+
	{\mathsf{\gamma}}
	\cdot
	(b-a)
	\cdot
	\frac{\textstyle \sqrt{m-1}}
	{\textstyle m}$.          
	Thus,
	${\mathsf{F}}_{1}^{{\mathsf{\gamma}}}\left( \frac{\textstyle 1}
	{\textstyle m}
	\right)
	=
	b$
	if and only if
	${\mathsf{\gamma}}
	=
	\frac{\textstyle 1}
	{\textstyle \sqrt{m-1}}$, for $m\in\mathbb{N}$.
}

\remove{
${\mathsf{F}}_{2}^{{\mathsf{\gamma}}}\left( \frac{\textstyle 1}
                                                                                     {\textstyle 2}
                                                                    \right)                                                                    
   =
   \frac{\textstyle a+b}
           {\textstyle 2}
  +
  {\mathsf{\gamma}}
  \cdot
  (b-a)
  \cdot
  \frac{\textstyle 1}
         {\textstyle 2}$.          
Thus,
${\mathsf{F}}_{1}^{{\mathsf{\gamma}}}\left( \frac{\textstyle 1}
                                                                                     {\textstyle 2}
                                                                   \right)
  >
  b$
if and only if
${\mathsf{\gamma}} > 1$.}                                                        

{\color{black}         
\subsubsection{The Valuation $\mathsf{CVaR}_{\alpha}$} 

$\mathsf{CVaR}_{\alpha}$ is a one-parameter valuation represented by the concave 
function $\mathsf{F}_3^{\alpha}(x)$: 

\noindent If $x<\alpha$, then $\displaystyle \mathsf{F}_3^{\alpha}(x)=\frac{1}{1-\alpha}\cdot (1-\alpha)\cdot b = b$,
and if $x\ge\alpha$, then $\displaystyle \mathsf{F}_3^{\alpha}(x)=\frac{1}{1-\alpha}\cdot ((x-\alpha)\cdot a + (1-x)\cdot b)$.\\
Note that
\[   
\mathsf{VaR}_{\alpha}(x)= 
\begin{cases}
a, &\quad\text{if } x\ge \alpha\\
b, &\quad\text{if } x< \alpha\\
\end{cases}
\]

\noindent So, $\mathsf{F}_3^{\alpha}$ is a continuous function with $\mathsf{F}_3^{\alpha}(x)=b$ for $0\le x \le \alpha$,
$\mathsf{F}_3^{\alpha}(1)=a$ and linear for $\alpha\le x \le 1$. {\it As it does not have a unique maximum, it is not a unimodal
	valuation.} Furthermore, since $\mathsf{F}_3^{\alpha}(x)$ is monotone
decreasing in $x$, an {\sf E}-equilibrium is also an $\mathsf{F}_3^{\alpha}$-equilibrium, but since $\mathsf{F}_3^{\alpha}(x)$
is constant for $0\le x \le \alpha$, this does not hold vice versa. So, an $\mathsf{F}_3^{\alpha}$-equilibrium
always exists, but computing an  $\mathsf{F}_3^{\alpha}$-equilibrium is not necessarily $\mathcal{PPAD}$-hard.  

\noindent Observe that the {\sf WEEP} does not hold for $\mathsf{F}_3^{\alpha}$. To see this, consider the game\smallskip

$\left( 
\begin{array}{ll}
(a,b) & (b,a)  \\
(b,a) & (a,b) \\  
\end{array}
\right)$ with $\alpha=\frac{3}{4}$, $p_2(1)=\frac{1}{4}$, $p_2(2)=\frac{3}{4}$.\smallskip

\noindent Then $x_1(p_1^{1},p_2)=\frac{1}{4}$, $x_1(p_1^{2},p_2)=\frac{3}{4}$, and
$\mathsf{F}_3^{\alpha}(p_1^{1},p_2)=\mathsf{F}_3^{\alpha}(p_1^{2},p_2)=b$ but
$\mathsf{E}(p_1^{1},p_2)=\frac{1}{4}\cdot a + \frac{3}{4}\cdot b \ne 
\mathsf{E}(p_1^{2},p_2)=\frac{3}{4}\cdot a + \frac{1}{4}\cdot b$. 
So $\mathsf{F}_3^{\alpha}$ does not have the {\sf WEEP}. As we will show shortly, unimodal valuations have the {\sf WEEP}.\medskip

\noindent Finally, note that while for 2-players, 2-values games there always exists a
$\mathsf{CVaR}_{\alpha}$-equilibrium, this is not true for 2-players, 3-values games.
It is shown in~\cite[Theorem 6]{MM17} that the Crawford game  $\left( 
\begin{array}{ll}
(2,2) & (1,3)  \\
(1,3) & (3,1) \\  
\end{array}
\right)$
has no $\mathsf{CVaR}_{\alpha}$-equilibrium.
}

\subsection{{Properties}}

We prove some properties
of unimodal valuations. {First we provide necessary definitions.
Consider a 2-players, 2-values game with a unimodal valuation {\sf V}
and a mixed profile $\langle p_{1}, p_{2}\rangle$. 
Then, we say that player $k\in [2]$ is {\it {\sf V}-constant on $\sigma(p_k)$},
	if $\mathsf{V}_k(\widehat{p}_k,p_{\bar{k}})$ remains constant over
	all strategies $\widehat{p}_k$ with $\sigma (\widehat{p}_{k})
	\subseteq 	\sigma (p_{k})$. The notion of a player being 
	 {\sf E}-constant is defined similarly.} We now show:

\begin{lemma}
	\label{x-constant}
	For a 2-players, 2-values game with a unimodal valuation $\mathsf{V}$,
	consider a mixed profile $\langle p_{1}, p_{2}\rangle$. For each player $k\in [2]$, if 
	$k$ is $\mathsf{E}$-constant on $\sigma(p_k)$ or if $k$ is $\mathsf{V}$-constant
	on $\sigma(p_k)$, then $x_k(\widehat{p}_k,p_{\bar{k}})$ is constant for all
	$\widehat{p}_k$ with $\sigma(\widehat{p}_k) \subseteq \sigma(p_k)$.
\end{lemma}

\begin{proof}
	We consider first the case that player $k\in [2]$ is {\sf E}-constant on $\sigma(p_k)$.
	Since $\mathsf{E}_k(\widehat{p}_k,p_{\bar{k}}) = \mathsf{E}_k(x_k(\widehat{p}_k,p_{\bar{k}}))$
	and the one-parameter function $E(x)$ is strictly monotone decreasing in $x$, this
	implies that $x_k(\widehat{p}_k,p_{\bar{k}})$ is constant for all $\widehat{p}_k$ with
	$\sigma (\widehat{p}_{k}) \subseteq 	\sigma (p_{k})$.
	
	Now consider the case that player $k\in [2]$ is {\sf V}-constant on $\sigma(p_k)$.
	Assume on the contrary that there exist $\widetilde{p}_k, \dbtilde{p}_k$
	with $\sigma(\widetilde{p}_k) \subseteq \sigma(p_k)$ and
	$\sigma(\dbtilde{p}_k) \subseteq \sigma(p_k)$ such that 
	$ y = x_k(\widetilde{p}_k,p_{\bar{k}}) < x_k(\dbtilde{p}_k,p_{\bar{k}}) = z.$ 
	Since $\mathsf{V}_k(\widehat{p}_k,p_{\bar{k}}) = \mathsf{F}_k(x_k(\widehat{p}_k,p_{\bar{k}}))$
	and since {\sf F} is a concave function with a unique maximum, this implies $y < x < z$ where
	$x$ is the position in which the unique maximum of {\sf F} is obtained. The properties of {\sf F}
	imply additionally that $\mathsf{F}(\widehat{x}) > \mathsf{F}(y) = \mathsf{F}(z)$ for all 
	$\widehat{x}$ with $y<\widehat{x}<z$.
	Now consider the mixed strategy $q_k = \frac{1}{2}\widetilde{p}_k + \frac{1}{2}\dbtilde{p}_k$.
	It follows that $q_k \subset \sigma(p_k)$ since $\sigma(p_k)$ is a convex set. Furthermore, 
	$y< x_k(q_k,p_{\bar{k}}) < z$ and therefore $\mathsf{F}(x_k(q_k,p_{\bar{k}})) > \mathsf{F}(y) = \mathsf{F}(z)$.
	This contradicts the fact that 
	$\mathsf{V}_k(\widehat{p}_k,p_{\bar{k}}) = \mathsf{F}_k(x_k(\widehat{p}_k,p_{\bar{k}}))$
	is constant on $\sigma(p_k)$.~\qed
\end{proof}

\noindent
{\bf Remark:} {\em Since $\mathsf{E}$ and $\mathsf{V}$ are one-parameter valuations,  
if $x_k(\widehat{p}_k,p_{\bar{k}})$ is constant on $\sigma(p_k)$, then 
$\mathsf{E}_k(\widehat{p}_k,p_{\bar{k}})$ and $\mathsf{V}_k(\widehat{p}_k,p_{\bar{k}})$
are constant on $\sigma(p_k)$.}

\begin{lemma}
\label{unimodal implies weep}
{A unimodal valuation ${\mathsf{V}}$
has the {\it Optimal-Value} property and
the {\sf WEEP}.}
\end{lemma}

\begin{proof}
{Consider a ${\mathsf{V}}$-equilibrium $\langle p_{1}, p_{2}\rangle$.}
{Since ${\mathsf{V}}$ is concave,
${\mathsf{V}}$ has the {\it Optimal-Value} property
(Proposition~\ref{concavity implies optimal value property}).}
{Hence,
for each player $k \in [2]$,
${\mathsf{V}}_{k}(\widehat{p}_{k}, p_{\overline{k}})$
remains constant
over all strategies
$\widehat{p}_{k}$
with
$\sigma (\widehat{p}_{k})
  \subseteq
  \sigma (p_{k})$.}
According to Lemma~\ref{x-constant} and the remark following
Lemma~\ref{x-constant}, this implies that
${\mathsf{E}}_{k}(\widehat{p}_{k}, p_{\overline{k}})$
is constant on $\sigma(p_{k})$.
The {\sf WEEP} follows.    \qed 
%
\remove{Since ${\mathsf{V}}$ is one-parameter,
this means that
${\mathsf{F}}(x_{i}(\widehat{p}_{i}, p_{\overline{i}}))$
remains constant
over all strategies
$\widehat{p}_{i}$
with
$\sigma (\widehat{p}_{i})
  \subseteq
  \sigma (p_{i})$.
Since
${\mathsf{F}}$ is concave
with a unique maximum,
it follows that
$x_{i}(\widehat{p}_{i}, p_{\overline{i}})$
remains constant
over all strategies
$\widehat{p}_{i}$
with
$\sigma (\widehat{p}_{i})
  \subseteq
  \sigma (p_{i})$.
In particular,
$x_{i}(p_{i}^{s}, p_{\overline{i}})$
remains constant
over all strategies
$s \in \sigma (p_{i})$.
Since ${\mathsf{E}}$
is a one-parameter valuation,
it follows that
${\mathsf{E}}_{i}(p_{i}^{s}, p_{\overline{i}})$
remains constant
over all strategies
$s \in \sigma (p_{i})$.
The {\sf WEEP} follows.
}
\end{proof}

\remove{
\noindent
{We continue to prove:}

\begin{lemma}
\label{useful}
{For a $2$-players, $2$-values game
with a unimodal valuation ${\mathsf{V}}$,
consider a mixed profile
$(p_{1}, p_{2})$.
Then,
each player $i \in [2]$
is \cgr{${\mathsf{E}}$-happy}
on $\sigma (p_{i})$
if and only if
she is \cgr{${\mathsf{F}}$-happy}
on $\sigma (p_{i})$.}\cgr{For constant?}
\end{lemma}

\begin{proof}
{Assume first that ${\bf p}$
is an ${\mathsf{E}}$-equilibrium.}
{It follows,
by the {\sf WEEP}
(Lemma~\ref{unimodal implies weep}),
that for each player $i \in [2]$,
${\mathsf{E}}_{i}(p_{i}^{s}, p_{\overline{i}})$
remains constant
over all strategies
$s \in \sigma (p_{i})$.}
{Since ${\mathsf{E}}$
is a unimodal valuation,
this implies that
$x_{i}(p_{1}^{s}, p_{2})$
remains constant over all strategies
$s \in \sigma (p_{i})$.}
{Since $x_{i}$
is linear in each of its arguments,
it follows that
$x_{i}(\widehat{p}_{i}, p_{\overline{i}})$
remains constant
over all mixed strategies
$\widehat{p}_{i}$
with
$\sigma (\widehat{p}_{i})
  \subseteq
  \sigma (p_{i})$.}
{Since
${\mathsf{V}}$ is a unimodal valuation,
it follows that
for each player $i \in [2]$,
${\mathsf{F}}_{i}(x_{i}(\widehat{p}_{i}, p_{2}))$
remains constant
over all mixed strategies
$\widehat{p}_{i}$
with
$\sigma (\widehat{p}_{i})
  \subseteq
  \sigma (p_{i})$.}
{It follows that
each player $i \in [2]$
is ${\mathsf{F}}$-happy
on $\sigma (p_{i})$.}
{Assume now that
${\bf p}$
is an 
${\mathsf{F}}$-equilibrium.}
{Since ${\mathsf{V}}$ is unimodal,
it has the ${\sf WEEP}$;
so,
${\mathsf{E}}_{i}(p_{i}^{s}, p_{\overline{i}})$
remains constant
over all strategies
$s \in \sigma (p_{i})$.}
{It follows that
each player $i \in [2]$
is ${\mathsf{E}}$-happy
on $\sigma (p_{i})$.}
\end{proof}
}

\noindent
{As a special case,
Lemma~\ref{x-constant}
immediately implies:}

\begin{corollary}
\label{completely trivial}
{For a 2-players, 
2-values game ${\mathsf{G}}$
with a unimodal valuation ${\mathsf{V}}$,
consider  the mixed profile $\langle p_{1}, p_{2}\rangle$.
For each player $k\in [2]$, if $|\sigma(p_k)| = n$,
then $k$ is {\sf E}-happy with  $\langle p_{1}, p_{2}\rangle$
if and only if
$k$ is {\sf V}-happy with  $\langle p_{1}, p_{2}\rangle$.
}
\end{corollary}

\noindent
{Corollary~\ref{completely trivial}
immediately implies:}

\begin{corollary}
	\label{cor:2strategies}
{A 2-players, 2-values, 2-strategies game
with a unimodal valuation ${\mathsf{V}}$
has a ${\mathsf{V}}$-equilibrium.}
\end{corollary}

\noindent
{We now prove a necessary condition
for the existence of an ${\mathsf{F}}$-equilibrium,
which we shall later use repeatedly
in the proofs of Theorems~\ref{hawaii grande} and~\ref{grande resort}:}

\begin{lemma}
\label{frequent lemma}
{Consider a 2-players, 2-values game ${\mathsf{G}}$
with a unimodal valuation ${\mathsf{F}}$,
and a mixed profile
$\langle p_{1}, p_{2}\rangle$
with the following three properties:}
\begin{enumerate}

\item[{\sf (1)}]
{There is a strategy
$\widehat{i} \in \sigma (p_{1})$
such that
${\mathsf{\mu}}_{1}(\widehat{i}, j)
  =
  b$
for all $j \in \sigma (p_{2})$.}

\item[{\sf (2)}]
{For each strategy
$j \in \sigma (p_{2})$,
there is a strategy
$i \in S_{1}$
with ${\mathsf{\mu}}_{1}(i, j) = a$.}

\item[{\sf (3)}]
{It holds that
${\mathsf{F}}(x) < b$
for all $x \geq \frac{\textstyle 1}
                               {\textstyle |\sigma (p_{2})|}$.}

\end{enumerate}
\noindent
{Then,
$p_{1}$ is not an ${\mathsf{F}}$-best response to $p_{2}$.}
\end{lemma}

\noindent
{Note that Property {\sf (2)}
is always fulfilled
when ${\mathsf{G}}$
is a normal game
(Definition~\ref{normal game})
due to Condition {\sf (2)}.}

\begin{proof}
{Assume,
by way of contradiction,
that $p_{1}$ is an ${\mathsf{F}}$-best response to $p_{2}$.
By Lemma~\ref{unimodal implies weep},
${\mathsf{V}}$
has the {\sf WEEP}.
Hence,
by Property {\sf (1)},
Lemma~\ref{domination and weep}
implies that
${\mathsf{\mu}}_{1}(i, j) = 
  {\mathsf{\mu}}_{1}(\widehat{i}, j)
  =
  b$
for all $(i, j) \in \sigma (p_{1}) \times \sigma (p_{2})$.
It follows that
$x_1(p_{1}, p_{2})
  =
  0$.
Hence,
${\mathsf{F}}(x_1(p_{1}, p_{2})) = 
   {\mathsf{F}}(0) = b$.}

{By Property {\sf (2)},
for each strategy $j \in \sigma (p_{2})$,
there is a strategy $i \in S_{1}$
with ${\mathsf{\mu}}_{1}(i, j) = a$;
call it $i(j)$.}
{Thus,
for each strategy $j \in \sigma (p_{2})$,
$i(j) \not\in \sigma (p_{1})$.}
{Set
$y := \max \{ p_{2}(j) \mid j \in \sigma (p_{2}) \}$.
Then,
$y \geq \frac{\textstyle 1}
                     {\textstyle |\sigma (p_{2})|}$.}
{Choose $\widetilde{j} \in \sigma (p_{2})$
with $p_{2}(\widetilde{j}) = y$,
and
set
$\widetilde{i} := i(\widetilde{j})$.}
{Then,
${\mathsf{V}}_{1}\left( p_{1}^{\widetilde{i}}, p_{2}
                              \right)
  =                            
  {\mathsf{F}}\left( x(p_{1}^{\widetilde{i}}, p_{2})
                       \right)
  =
  {\mathsf{F}}(x)$,
for some $x$
with
$x\ge y \geq \frac{\textstyle 1}
                     {\textstyle |\sigma (p_{2})|}$.}
{It follows,
by Property {\sf (3)},
that
${\mathsf{F}}(x) < b$.}                     
{So player $1$ improves
by switching to strategy $\widetilde{i}$.
A contradiction.}  \qed                                          
\end{proof}

\noindent
{We shall sometimes use
Lemma~\ref{frequent lemma}
with the roles of the two players interchanged.}

\subsection{The case {$x_0(\mathsf{F})=0$}}
\label{b-a}
\noindent
{We observe:}

\begin{lemma}
	\label{significant}
	{Assume that 
		the unimodal valuation
		${\mathsf{F}}$ is strictly monotone decreasing,
		and consider a mixed profile
		${\bf p} = \langle p_{1}, p_{2}
		\rangle$.
		Then,
		$p_{1}$
		is an ${\mathsf{E}}$-best-response
		to $p_{2}$
		if and only if
		it is an
		${\mathsf{F}}$-best-response
		to $p_{2}$.}                 
\end{lemma}

\begin{proof}
	{Consider a mixed strategy $\widehat{p}_{1}$
		for player $1$.
		Since both ${\mathsf{E}}$ and ${\mathsf{F}}$
		are monotone decreasing in $x$,
		we have that}
	{
		{
			\begin{eqnarray*}
				{\mathsf{F}}_{1}(x_{1}({\widehat{p}}_{1}, p_{2}))\ \
				<\ \
				{\mathsf{F}}_{1}(x_{1}(p_{1}, p_{2}))
				& \Longleftrightarrow &  x_{1}({\widehat{p}}_{1}, p_{2})\ \
				<\ \
				x_{1}(p_{1}, p_{2})                                                                                                                           \\
				& \Longleftrightarrow &  {\mathsf{E}}_{1}(x_{1}({\widehat{p}}_{1}, p_{2}))\ \
				<\ \
				{\mathsf{E}}_{1}(x_{1}(p_{1}, p_{2}))\, .                                       
			\end{eqnarray*}
		}
	}

\noindent Hence,
		$p_{1}$
		is an ${\mathsf{E}}$-best-response
		to $p_{2}$
		if and only if
		it is an
		${\mathsf{F}}$-best-response
		to $p_{2}$.\qed
\end{proof}

\noindent
{Since ${\mathsf{E}}$-equilibria
	are invariant to translating and scaling the cost values $a$ and $b$,
	Lemma~\ref{significant}
	implies that
	for a strictly monotone decreasing ${\mathsf{F}}$,
	computing an ${\mathsf{F}}$-equilibrium
	for a $2$-values games
	is as hard as computing an ${\mathsf{E}}$-equilibrium
	for a {\it win-lose} game,
	where $a=0$ and $b=1$.
	The latter is
	${\mathcal{PPAD}}$-hard~\cite{AKV05}.
	Hence,
	computing an ${\mathsf{F}}$-equilibrium
	for a strictly monotone decreasing ${\mathsf{F}}$
	is ${\mathcal{PPAD}}$-hard.
	{In the case $x_0(\mathsf{F})=0$, ${\mathsf{F}}$
		is strictly monotone decreasing
		for $x \in [0, 1]$.} 
	{Hence,
		we obtain:} 
	
	\begin{theorem}
		\label{2 players 2 values ppad hard}
		{For a unimodal valuation {\sf F}, if $x_0(\mathsf{F}) = 0$, then an {\sf F}-equilibrium exists for all 
			2-players, 2-values games
			and its computation is ${\mathcal{PPAD}}$-hard.}
	\end{theorem} 
	
	\noindent Recall that for 
		${\mathsf{\gamma}} \cdot (b-a) \leq 1$,
		${\mathsf{EVar}}^{{\mathsf{\gamma}}}$ 
		is strictly monotone decreasing
		for $x \in [0, 1]$ and ${\mathsf{EVar}}^{{\mathsf{\gamma}}}(0)=b$.
		Hence, $x_0(\mathsf{\mathsf{EVar}}^{{\mathsf{\gamma}}})=0$.
		Therefore:
	
	\begin{corollary}
		\label{EVar-PPAD}
		Computing an ${\mathsf{EVar}}^{{\mathsf{\gamma}}}$-equilibrium,
		with ${\mathsf{\gamma}} \cdot (b-a) \leq 1$,
		is ${\mathcal{PPAD}}$-hard
		for 2-players, 2-values games.
	\end{corollary}

\section{{Inexistence of Equilibria}}
\label{sec:non-existence of equilibria}

\remove{Revise: In this section we introduce the $2$-values, $(\n+1)$-strategies bimatrix normal game $\bm{{\mathsf{C}}}_{\n}$, with $\n\ge 2$, and show that 
	it has no $\mathsf{F}$-equilibrium when ${\mathsf{F}}\left( \frac{\textstyle 1}
	{\textstyle n} \right) > b$
	and ${\mathsf{F}}\left( \frac{\textstyle 1}
	{\textstyle n-1} \right) < b$.}


{We begin with a property about uniqueness of equilibria which will have
	the essential role in the proof of the non-existence result. 
	For this purpose, we introduce the 
	$2$-values, $\n$-strategies bimatrix game ${\mathsf{D}}_{\n}$,
	with $\n \geq 2$:}
{
	\begin{eqnarray*}
		{\mathsf{D}}_{\n}
		& = &  \left( \begin{array}{llllll}
			(a, b)  & (b,a)   & (b,b)    & \ldots  & (b,b)   & (b,b)   \\
			(b,b)   & (a,b)   & (b,a)    & \ldots  & (b,b)   & (b,b)   \\
			\vdots & \vdots & \vdots & \ddots & \vdots & \vdots \\
			(b,b)   & (b,b)    & (b,b)   & \ldots  & (a,b)    & (b,a)   \\
			(b,a)   & (b,b)    & (b,b)   & \ldots  & (b,b)    & (a,b)   \\
		\end{array}
		\right)\, .           
	\end{eqnarray*}
}
\noindent Thus,
	${\mathsf{D}}_{\n}
	=
	(\alpha_{ij}, \beta_{ij})_{0 \leq i, j \leq \n-1}$
	is a {\it Toeplitz} bimatrix,  
	with:\medskip
	
		\begin{math}
		\alpha_{ij}
			 =  \left\{ \begin{array}{ll}
				a\, , & \mbox{if $i=j$} \\
				b\, , & \mbox{otherwise} \\
			\end{array}
			\right.\, \text{ and }
			\end{math}           
			 \begin{math}
	  	\beta_{ij}
	 	 =  \left\{ \begin{array}{ll}
	 		a\, , & \mbox{if $j = (i+1) \mod \n$} \\
	 		b\, , & \mbox{otherwise} \\
	 	\end{array}
	 	\right.\, . 
	 	\end{math} \medskip         

\noindent Clearly,
	${\mathsf{D}}_{\n}$ is a normal game.
	Thus:
\begin{itemize}

	\item
	{For each strategy
		$j \in \{ 0, 1, \ldots, \n-1 \}$,
		there is a strategy
		$i \in \mbox{\{0, 1, \ldots, \n-1\}}$
		with
		${\mathsf{\alpha}}_{ij} = a$.
		This implies that
		any mixed profile
		$\langle p_{1}, p_{2}\rangle$
		fulfills Property {\sf (2)}
		from Lemma~\ref{frequent lemma}.}

\end{itemize}
\noindent
{Note that for each player,
	there is exactly one $a$ in every row
	and every column.
	This property is stronger than
	Conditions {\sf (2)} and {\sf (3)} 
	together
	in the definition of a normal game.
	We show:}

\remove{
\begin{theorem}
	\label{hawaii grande}
	{Consider a unimodal valuation
		${\mathsf{F}}$ with
		${\mathsf{F}}(x) < b$
		for
		$x \geq \frac{\textstyle 1}
		{\textstyle n-1}$.
		Then,
		${\mathsf{D}}_{n}$
		has a unique ${\mathsf{F}}$-equilibrium
		$(p_{1}, p_{2})$,
		given by
		$p_{1}(j) = p_{2}(j) = \frac{\textstyle 1}
		{\textstyle n}$,
		for
		$0 \leq j \leq n-1$.}                                             
\end{theorem}
}

\begin{theorem}
	\label{hawaii grande}
	Consider a unimodal valuation ${\mathsf{F}}$. 
	${\mathsf{D}}_{\n}$ 
	has the ${\mathsf{F}}$-equilibrium
	$\langle p_{1}, p_{2}\rangle$,
	given by
	$p_{1}(j) = p_{2}(j) = \frac{\textstyle 1}
	{\textstyle \n}$,
	for
	$0 \leq j \leq \n-1$.
	${\mathsf{D}}_{\n}$ 
	has no other ${\mathsf{F}}$-equilibrium if and only if
	one of the following {three} conditions holds:
	\begin{enumerate}[$(i)$]
		\item {$\n \le 3$},
		\item $\n$ is even, {$\n \ge 4$}, $\mathsf{F}\left(\frac{2}{\n}\right)\neq b$ and $\mathsf{F}(x)<b$ for $x\ge \frac{2}{\n-2}$,
		\item $\n$ is odd, {$\n \ge 5$} and  $\mathsf{F}(x)<b$ for $x\ge \frac{2}{\n-1}$. 
	\end{enumerate}
\end{theorem}


\noindent Before proceeding to the proof, we present a property concerning the distribution of the $b$-values.

\begin{definition}[$b$-blocks and $b$-double-blocks]
	\label{def:bblocks}
	Let $A,B \subseteq \{0,\dots,\n-1\}$ {with $|A|\ge 1$ and $|B|\ge 1$.} Then, 
	\begin{enumerate}[$(i)$]
		\item $(A,B)$ is a {\it $b$-block} for player 1 if $\alpha_{ij}=b$ for all $i\in A, j\in B$.
		\item $(A,B)$ is a {\it $b$--double-block} if $\alpha_{ij}=\beta_{ij}=b$ for all $i\in A, j\in B$. 
	\end{enumerate}
   \noindent In a corresponding way, a $b$-block for player 2 is defined.
\end{definition}

\noindent We show: 

\begin{lemma}
\label{lem:bblocks}
Let $A,B \subseteq \{0,\dots,\n-1\}$. 
	\begin{enumerate}[$(i)$]
		\item If $(A,B)$ is a $b$-block for some player $k\in [2]$, then $|A|+|B|\le \n$.
		\item If $(A,B)$ is a $b$-double-block, then $|A|+|B|\le \n-1$.
	\end{enumerate}	
\end{lemma}

\begin{proof}
	Note that $|B|>1$ implies $|A|<\n$ and $|A|>1$ implies $|B|<\n$. 
    \begin{enumerate}[$(i)$]
    	\item We do the proof only for player 1. Note that $\alpha_{ii}=a$ for all $i\in\mbox{\{0,\ldots,\n-1\}}$.
    	So, if $i\in A$ then $i\notin B$. This implies $|B|\le \n - |A|$, and therefore $|A|+|B|\le \n$.
    	\item Note that for all $i\in \{0,\ldots,\n-1\}$, $\alpha_{ii}=\beta_{i,(i+1)\!~\!\mathrm{mod}\!~\!\n} = a$.
    	So, if $i\in A$, then $i,(i+1)~\mathrm{mod}~\n \notin B$. 
    	Since $|A|< \n$, $|\{i,(i+1)~\mathrm{mod}~\n~:~ i\in A\}| \ge |A|+1$.
    	This implies $|B| \le \n - (|A|+1)$, and therefore  $|A|+|B|\le \n-1$.
    \end{enumerate}  
The claim follows.~\qed
\end{proof}

\noindent We are now ready to prove Theorem~\ref{hawaii grande}.

\begin{proof}
	By the construction of ${\mathsf{D}}_{\n}$, 
	for each player $k \in [2]$  
	and for each strategy $i \in \{ 0, \ldots, \n-1 \}$,
	${\mathsf{E}}_{k}(p_{1}^{i}, p_{2})
	=
	a \cdot \frac{\textstyle 1}
	{\textstyle \n}
	+
	b \cdot \frac{\textstyle \n-1}
	{\textstyle \n}$,
	which is independent of $i$.
	Thus,
	$\langle p_{1}, p_{2}\rangle$
	has the {\sf WEEP}.
	Since 	$\langle p_{1}, p_{2}\rangle$ is fully mixed,
	it follows that
	$\langle p_{1}, p_{2}\rangle$ is an ${\mathsf{E}}$-equilibrium.
	Hence, 
	by Corollary~\ref{completely trivial},                                         
	$\langle p_{1}, p_{2}\rangle$
	is an ${\mathsf{F}}$-equilibrium.

\noindent We now prove that the stated conditions are necessary and sufficient for uniqueness. First we show that they are
sufficient.\smallskip\\
{\bf ``$\bm{\Leftarrow}$'':} Consider an arbitrary ${\mathsf{F}}$-equilibrium
$\langle p_{1}, p_{2}\rangle$. We distinguish three cases.
\begin{enumerate}[{\sf (1)}]
	\item \underline{$(\sigma(p_1),\sigma(p_2))$ is a $b$-double-block:}
	Due to the structure of ${\mathsf{D}}_{\n}$, such a double-block of $b$'s does
	not exist if $\n\le 4$. Now let $\n\ge 5$. Due to Lemma~\ref{lem:bblocks},
	$|\sigma(p_1)|+|\sigma(p_2)|\le \n-1$. 
	Due to the symmetry of ${\mathsf{D}}_{\n}$, we can assume that 
	$|\sigma(p_1)| \le |\sigma(p_2)|$. Then, $|\sigma(p_1)| \le \lfloor\frac{\n-1}{2}\rfloor$.
	We have to distinguish even and odd $\n$.
	\begin{itemize}
		\item [--] If $\n$ is even, then $|\sigma(p_1)| \le \frac{\n-2}{2}$, and by
		assumption $\mathsf{F}(x) < b$ for $x\ge \frac{2}{\n-2}$. 
		Therefore, $\mathsf{F}(x) < b$ for $x\ge \frac{1}{|\sigma(p_1)|}$.
		Lemma~\ref{frequent lemma} implies that $p_2$ is not an {\sf F}-best response to $p_1$. A contradiction.
		\item [--] If $\n$ is odd, then $|\sigma(p_1)| \le \frac{\n-1}{2}$, and by
		assumption $\mathsf{F}(x) < b$ for $x\ge \frac{2}{\n-1}$. 
		Therefore, $\mathsf{F}(x) < b$ for $x\ge \frac{1}{|\sigma(p_1)|}$.
		Lemma~\ref{frequent lemma} implies that $p_2$ is not an {\sf F}-best response to $p_1$. A contradiction.
	\end{itemize} 
	 \item \underline{For some $k\in [2]$, $(\sigma(p_1),\sigma(p_2))$ is a $b$-block for player $k$
	 	but not a $b$-block:} 
 	Due to the symmetry of ${\mathsf{D}}_{\n}$,
 	    we can assume that $k=2$. Then there is a pair $(\hat{i},\hat{j})\in \sigma(p_1)\times\sigma(p_2)$
 	    with $\alpha_{\hat{i}\hat{j}}=a.$
 	    We first prove that for every $i \in \sigma (p_{1}),$
 	    there is some $j \in \sigma (p_{2})$
 	    with $\alpha_{ij} = a$.\\ 
 	    Assume, by way of contradiction, that 
 	    		$\alpha_{\widetilde{i}, j} = b$
 	    		for some $\widetilde{i} \in \sigma (p_{1})$
 	    		and all $j \in \sigma (p_{2}).$
 	    		Then, for player $1$,
 	    		strategy $\widehat{i} \in \sigma (p_{1})$
 	    		dominates strategy $\widetilde{i} \in \sigma (p_{1})$
 	    		with respect to $\sigma (p_{2}).$
 	    		By Lemmas~\ref{domination and weep} and~\ref{unimodal implies weep},
 	    	    it follows that
 	    		no strategy in $\sigma (p_{1})$ dominates
 	    		some other strategy in $\sigma (p_{1})$
 	    		with respect to $\sigma (p_{2})$.
 	    		A contradiction.\\
 	    Note that for all $i,j\in\{0,\ldots,\n-1\}$, $\alpha_{ij}=a$ if and only if $j=i$.
 	    So, $i\in\sigma(p_1)$ implies $i\in\sigma(p_2)$. 
 	    Therefore, 	$|\sigma(p_1)| \le |\sigma(p_2)|$. 
 	    Lemma~\ref{lem:bblocks} implies $|\sigma(p_1)|+|\sigma(p_2)|\le \n$, and therefore
 	    $|\sigma(p_1)|\le \lfloor\frac{\n}{2}\rfloor$. {We distinguish even and odd $\n$.}
 	    \begin{itemize}
 	    	\item [--] {If $\n$ is odd, then $|\sigma(p_1)|\le \frac{\n-1}{2}$,  and by
 	    	assumption $\mathsf{F}(x) < b$ for $x\ge \frac{2}{\n-1}$. 
 	    	Therefore, $\mathsf{F}(x) < b$ for $x\ge \frac{1}{|\sigma(p_1)|}$.
 	    	Lemma~\ref{frequent lemma} implies that $p_2$ is not an {\sf F}-best response to $p_1$. A contradiction.}
 	    	\item [--] {If $\n$ is even, then $|\sigma(p_1)|\le \frac{\n}{2}$. We have two subcases,
 	    	$|\sigma(p_1)|< \frac{\n}{2}$ and $|\sigma(p_1)|= \frac{\n}{2}$.}
 	    	\begin{itemize}
 	    		\item [+] {\underline{$|\sigma(p_1)|\le \frac{\n}{2}-1$:} By assumption, $\mathsf{F}(x) < b$ for $x\ge \frac{2}{\n-2}$. 
 	    		Therefore, $\mathsf{F}(x) < b$ for $x\ge \frac{1}{|\sigma(p_1)|}$.
 	    		Lemma~\ref{frequent lemma} implies that $p_2$ is not an {\sf F}-best response to $p_1$. A contradiction.  }
 	    		\item [+] \underline{$|\sigma(p_1)| = \frac{\n}{2}$:} By assumption, $\mathsf{F}\left(\frac{2}{\n}\right)\ne b$.
 	    		If $\mathsf{F}\left(\frac{2}{\n}\right) < b$, then Lemma~\ref{frequent lemma} implies that $p_2$ is not
 	    		 an {\sf F}-best response to $p_1$.
 	    		A contradiction.  \cg{If $\mathsf{F}\left(\frac{2}{\n}\right) > b$, then the {\sf WEEP} implies
 	    		that $p_2(j)=\frac{2}{\n}$ for all $j\in \sigma(p_2)$. Thus $\mathsf{V}_1(p_1,p_2)=\mathsf{F}\left(\frac{2}{\n}\right)>b$
 	    		and hence player 1 can improve by
 	    		using some strategy $i\notin\sigma(p_1)$. A contradiction.}
 	    	\end{itemize}
 	    \end{itemize} 
 	    
 	    \remove{ 
 	    	\cgr{If  $\n$ is even (resp., odd), then $|\sigma(p_1)|\le \frac{\n}{2} < \frac{\n+2}{2}$
 	    and by assumption $\mathsf{F}(x)<b$ for $x\ge \frac{2}{\n+2}$ 
 	    (resp., $|\sigma(p_1)|\le \frac{\n-1}{2}$
 	    and by assumption $\mathsf{F}(x)<b$ for $x\ge \frac{2}{\n-1}$). 
 	    In both cases  $\mathsf{F}(x)<b$ for $x\ge \frac{1}{|\sigma(p_1)|}$.
 	    Lemma~\ref{frequent lemma} implies that $p_2$ is not an {\sf F}-best response to $p_1$,
 	    a contradiction. }} 
 	    \item \underline{$(\sigma(p_1),\sigma(p_2))$ is neither a $b$-block for player 1
 	    nor a $b$-block for player 2:} Then, there is a pair 
     	$(\hat{i},\hat{j})\in \sigma(p_1)\times\sigma(p_2)$ with $\alpha_{\hat{i}\hat{j}}=a$ and
     	there is a pair $(\widetilde{i},\widetilde{j})\in \sigma(p_1)\times\sigma(p_2)$ with $\beta_{\widetilde{i}\widetilde{j}}=a$.  			
	   Since in an equilibrium no strategy can dominate some other 
	   (Lemmas~\ref{domination and weep} and~\ref{unimodal implies weep}),  
	   for every $i\in \sigma(p_1)$ there is some $j\in \sigma(p_2)$ with $\alpha_{ij}=a$ and for
	   every $j\in \sigma(p_2)$ there is some $i\in \sigma(p_1)$ with $\beta_{ij}=a$. A more
	   elaborate proof was given in case {\sf (2)}. 
	   
	   {Due to the structure of 
	   	$\left( \alpha_{ij} \right)_{0 \leq i, j \leq \n-1}$,
	   	it follows that for every
	   	$i \in \sigma(p_{1})$,
	   	we also have 
	   	$i \in \sigma(p_{2})$.}
	   {Due to the structure of 
	   	$\left( \beta_{ij} \right)_{0 \leq i, j \leq \n-1}$,
	   	it follows that for every
	   	$j \in \sigma (p_{2})$,
	   	we also have 
	   	($j-1 \mod \n) \in \sigma (p_{1})$.}
	   {These two taken together yield  that
	   	$\sigma (p_{1}) = \sigma (p_{2}) = \{ 0, \ldots, \n-1 \}$.}
	   {By the {\sf WEEP}
	   	for player $1$,
	   	the expression 
	   	\begin{eqnarray*}
	   		a \cdot p_{1}(i)
	   		+
	   		b \cdot \sum_{i' \in \{ 0, \ldots, \n-1 \} \setminus \{ i \}} p_{1}(i')
	   	\end{eqnarray*} 
	   	is constant for $i\in\{0,\ldots,\n-1\}$, 
	   	yielding
	   	the unique solution
	   	$p_{1}(i) = \frac{\textstyle 1}
	   	{\textstyle \n}$
	   	for each strategy
	   	$i \in \mbox{\{0, \ldots, \n-1\}}$.}
	   {By the {\sf WEEP} for player $2$,
	   	we get identically the unique solution                
	   	$p_{2}(j) = \frac{\textstyle 1}
	   	{\textstyle \n}$
	   	for each strategy
	   	$j \in \mbox{\{ 0, \ldots, \n-1 \}}$.}
\end{enumerate}	

\noindent {\bf ``$\bm{\Rightarrow}$'':} We now
show that the conditions are necessary for uniqueness.
{The conditions do not hold if and only if
	(there exists even $\n\ge 4$ with $\mathsf{F}\left(\frac{2}{\n}\right)= b$ or $\mathsf{F}\left(\frac{2}{\n-2}\right) \ge b$)
	or (there exists odd $\n\ge 5$ with $\mathsf{F}\left(\frac{2}{\n-1}\right) \ge b$).}
\begin{enumerate}[{\sf(1)}]
	\item If $\n$ is even, {$\n\ge 4$} and $\mathsf{F}\left(\frac{2}{\n}\right)=b$, then
	\begin{itemize}
		\item [] {$\sigma(p_1)=\{i\in \{0,\ldots,\n-1\}~:~i~\text{even}\},~p_1(i)=\frac{2}{\n}$ for all $i\in \sigma(p_1)$}\smallskip\\
		{$\sigma(p_2)=\{j\in \{0,\ldots,\n-1\}~:~j~\text{odd}\},~p_2(j)=\frac{2}{\n}$ for all $j\in \sigma(p_2)$}
	\end{itemize}
	is an {\sf F}-equilibrium. \cg{$|\sigma(p_1)|=|\sigma(p_2)| =\frac{\n}{2}$, and {\bf p} is uniform on $\sigma(p_k)$ for $k\in[2]$.}\\ 
	For example, for $\n=4$ and $\mathsf{F}\left(\frac{1}{2}\right)=b$ we get:
	$ 		\left( \begin{array}{llllll}
	(a, b)  & {\color{green} \fbox{(b,a)}}   & (b,b)   &   {\color{green} \fbox{(b,b)}}    \\
	(b,b)   & (a,b)   & (b,a)   & (b,b)    \\
	(b,b)   &  {\color{green} \fbox{(b,b)}}     & (a,b) &  {\color{green} \fbox{(b,a)}}    \\
	(b,a)   & (b,b)    & (b,b)   & (a,b)   \\
	\end{array}
	\right)\, .$          
	\item {If $\n$ is even, $\n\ge 4$ and $\mathsf{F}\left(\frac{2}{\n-2}\right)\ge b$, then}
	\begin{itemize}
		\item [] {$\sigma(p_1)=\{i~:~\frac{\n}{2}\le i \le \n-2\},~p_1(i)=\frac{2}{\n-2}$ for all $i\in \sigma(p_1)$}\smallskip\\
		{$\sigma(p_2)=\{j~:~0 \le j \le \frac{\n}{2}-2\},~p_2(j)=\frac{2}{\n-2}$ for all $j\in \sigma(p_2)$}
	\end{itemize}
	{is an {\sf F}-equilibrium.} \cg{$|\sigma(p_1)|=|\sigma(p_2)| =\frac{\n-2}{2}$, and {\bf p} is uniform on $\sigma(p_k)$ for $k\in[2]$.}\\ 
	{For example, for $\n=6$ and $\mathsf{F}\left(\frac{1}{2}\right)\ge b$ we get:}
	$ 		\left( \begin{array}{llllll}
	(a,b)  & (b,a)  & (b,b)   &   (b,b)  &   (b,b) &   (b,b)  \\
	(b,b)   & (a,b)   & (b,a)   & (b,b) &   (b,b) &   (b,b)   \\
	(b,b)   & (b,b)   & (a,b)   & (b,a) &   (b,b) &   (b,b)   \\
	{\color{green} \fbox{(b,b)}}   &  {\color{green} \fbox{(b,b)}}   &   (b,b)  & (a,b) &  (b,a) &   (b,b)    \\
	{\color{green} \fbox{(b,b)}}   &  {\color{green} \fbox{(b,b)}}   &   (b,b)  & (b,b) &  (a,b) &   (b,a)    \\
	(b,a)   & (b,b)   & (b,b)   & (b,b) &   (b,b) &   (a,b)   \\
	\end{array}
	\right)\, .$  
	\item If $\n$ is odd, {$\n\ge 5$} and $\mathsf{F}\left(\frac{2}{\n-1}\right)\ge b$, then
	 \begin{itemize} 
		\item [] {$\sigma(p_1)=\left\{i~:~\frac{\n-1}{2}\le i\le \n-2\right\},~p_1(i)=\frac{2}{\n-1}$ for all $i\in \sigma(p_1)$}\smallskip\\
		{$\sigma(p_2)=\left\{j~:~0\le j\le \frac{\n-3}{2} \right\},~p_2(i)=\frac{2}{\n-1}$ for all $j\in \sigma(p_2)$} 
	\end{itemize}
    {is an {\sf F}-equilibrium.} \cg{$|\sigma(p_1)|=|\sigma(p_2)| =\frac{\n-1}{2}$, and {\bf p} is uniform on $\sigma(p_k)$ for $k\in[2]$.}\\
	For example, for $\n=5$ and $\mathsf{F}\left(\frac{1}{2}\right)\ge b$ we get:
	$ 		\left( \begin{array}{llllll}
	(a, b)  & (b,a)   & (b,b)   &   (b,b) & (b,b)    \\
	(b,b)   & (a,b)   & (b,a)   & (b,b)  & (b,b)  \\
	{\color{green} \fbox{(b,b)}}   &  {\color{green} \fbox{(b,b)}}     & (a,b) &  (b,a)  & (b,b)  \\
	{\color{green} \fbox{(b,b)}}   & {\color{green} \fbox{(b,b)}}    & (b,b)   & (a,b)   & (b,a) \\
	(b,a)   & (b,b)    & (b,b)   & (b,b)   & (a,b) \\
	\end{array}
	\right)\, .$\smallskip
\end{enumerate}	
     
\noindent This completes the proof.\qed
\end{proof}

\remove{
\begin{proof}
	{By the construction of ${\mathsf{D}}_{n}$, 
		for each player $i \in [2]$  
		and for each strategy $k \in \{ 0, \ldots, n-1 \}$,
		${\mathsf{E}}_{i}(p_{1}^{k}, p_{2})
		=
		a \cdot \frac{\textstyle 1}
		{\textstyle n}
		+
		b \cdot \frac{\textstyle n-1}
		{\textstyle n}$,
		which is independent of $k$.
		Thus,
		$(p_{1}, p_{2})$
		has the {\sf WEEP}.
		Since $(p_{1}, p_{2})$ is fully mixed,
		it follows that
		$(p_{1}, p_{2})$ is an ${\mathsf{E}}$-equilibrium.
		Hence, 
		by Corollary~\ref{completely trivial},                                         
		$(p_{1}, p_{2})$
		is an ${\mathsf{F}}$-equilibrium.}
	{To prove uniqueness,}
	{consider 
		an arbitrary ${\mathsf{F}}$-equilibrium
		$(p_{1}, p_{2})$.
		We distinguish two cases.}
	\begin{enumerate}
		\item[{\sf (1)}]
		{\underline{$\alpha_{ij} = b$
				for all $(i, j)
				\in 
				\sigma (p_{1}) \times \sigma (p_{2})$:}}           
		 {We shall apply
			Lemma~\ref{frequent lemma}
			to player $1$.} 
		{Property {\sf (1)}
			holds vacuously due to the assumption
			for Case {\sf (1)}.} 
		{For Property {\sf (3)},
			note that
			$|\sigma (p_{2})| \leq n-1$.
			To see this,
			consider
			some $i \in \sigma (p_{1})$.
			{Due to the structure of ${\mathsf{D}}_{n}$,}
			there is a 
			$j \in \{ 0, \ldots, n-1 \}$
			with $\alpha_{ij} = a$.
			From the condition for Case {\sf (1)},
			it follows that $j \not\in \sigma (p_{2})$.}
		{Thus,
			$|\sigma (p_{2})| \leq n-1$.}
		{Since
			${\mathsf{F}}(x)
			<
			b$
			for
			$x \geq \frac{\textstyle 1}
			{\textstyle n-1}$,  
			it follows that
			${\mathsf{F}}(x)
			<
			b$
			for
			$x \geq \frac{\textstyle 1}
			{\textstyle |\sigma (p_{2})|}$,
			which is
			Condition {\sf (3)}.
			It follows,
			by Lemma~\ref{frequent lemma},
			that $p_{1}$
			is not an ${\mathsf{F}}$-best response
			to $p_{2}$.
			A contradiction.}                               
		\item[{\sf (2)}]
		\underline{{There is a pair
				$(\widehat{i},
				\widehat{j})
				\in
				\sigma (p_{1}) \times \sigma (p_{2})$
				with
				$\alpha_{ij} = a$.}}
		{We first prove that
			for every $i \in \sigma (p_{1})$
			(resp., for every $j \in \sigma (p_{2})$),
			there is some $j \in \sigma (p_{2})$
			(resp., some $i \in \sigma (p_{1})$)
			with $\alpha_{ij} = a$
			(resp., with $\beta_{ij} = a$):}
		\begin{quote}
			{Assume,
				by way of contradiction,
				that 
				$\alpha_{\widetilde{i}, j}
				=
				b$
				for some $\widetilde{i} \in \sigma (p_{1})$
				and all
				$j \in \sigma (p_{2})$
				(resp.,
				$\beta_{i, \widetilde{j}}
				=
				b$
				for some $\widetilde{j} \in \sigma (p_{2})$
				and all
				$i \in \sigma (p_{1})$).}
			{Then,
				for player $1$
				(resp., for player $2$),
				strategy $\widehat{i} \in \sigma (p_{1})$
				dominates strategy $\widetilde{i} \in \sigma (p_{1})$
				with respect to $\sigma (p_{2})$
				(resp.,
				strategy $\widehat{j} \in \sigma (p_{2})$
				dominates strategy $\widetilde{j} \in \sigma (p_{2})$
				with respect to $\sigma (p_{1})$).}
			{By Lemmas~\ref{unimodal implies weep}
				and~\ref{domination and weep},
				it follows that
				no strategy in $\sigma (p_{1})$
				(resp., $\sigma (p_{2})$)
				dominates
				some other strategy in $\sigma (p_{1})$
				(resp., $\sigma (p_{2})$)
				with respect to $\sigma (p_{2})$
				(resp., $\sigma (p_{1})$).
				A contradiction.}
		\end{quote}
		{Due to the structure of 
			$\left( \alpha_{ij} \right)_{0 \leq i, j \leq n-1}$,
			it follows that for every
			$i \in \sigma (p_{1})$,
			we also have 
			$i \in \sigma (p_{2})$.}
		{Due to the structure of 
			$\left( \beta_{ij} \right)_{0 \leq i, j \leq n-1}$,
			it follows that for every
			$j \in \sigma (p_{2})$,
			we also have 
			$j-1 \mod n \in \sigma (p_{1})$.}
		{These two taken together,
			it follows that
			$\sigma (p_{1}) = \sigma (p_{2}) = \{ 0, \ldots, n-1 \}$.}
		{By the {\sf WEEP}
			for player $1$,
			for all strategy pairs
			$k, k^{\prime} \in \{ 0, 1, \ldots, n-1 \}$,
			\begin{eqnarray*}
				a \cdot p_{1}(k)
				+
				b \cdot \sum_{k \in \{ 0, 1, \ldots, n-1 \} \setminus \{ i \}} p_{1}(k)
				& = &
				a \cdot p_{1}(i^{\prime})
				+
				b \cdot \sum_{k \in \{ 0, 1, \ldots, n-1 \} \setminus \{ i^{\prime} \}} p_{1}(k)\, ,
			\end{eqnarray*} 
			yielding
			the unique solution
			$p_{1}(k) = \frac{\textstyle 1}
			{\textstyle n}$
			for each strategy
			$k \in \{ 0, 1, \ldots, n-1 \}$.}
		{By the {\sf WEEP} for player $2$,
			we get identically the unique solution                
			$p_{2}(k) = \frac{\textstyle 1}
			{\textstyle n}$
			for each strategy
			$k \in \{ 0, \ldots, n-1 \}$.}                                               
	\end{enumerate}
	\noindent
	{The claim follows.}
\end{proof}
} 


\noindent
We now proceed to prove the inexistence result.    
Game ${\mathsf{C}}_{\n}$ is a  
%
$2$-values,
$(\n+1)$-strategies bimatrix game,
with $\n \geq 2$,
which 
is derived by adding one row and one column to game ${\mathsf{D}}_{\n}$,
as follows:
{
	\begin{eqnarray*}
		{\mathsf{C}}_{\n}
		& = &  \left( \begin{array}{lllll|l}
			&          &                              &           &         & (a,b)   \\
			&          &                              &           &         & (b,b)   \\
			&          & {\mathsf{D}}_{\n}&           &          & \vdots \\
			&          &                              &           &          &            \\
			&          &                              &           &          & (b,b)   \\
			\hline            
			(b,b)   & (b,b)  & \ldots                   & (b,b)  & (b,b) & (b,a)              \\
		\end{array}
		\right)\, .           
	\end{eqnarray*}
}

{\color{black}
\noindent Thus,
	${\mathsf{C}}_{\n}
	=
	(\alpha_{ij}, \beta_{ij})_{0 \leq i, j \leq \n}$,
	where: 
	\begin{itemize}
		\item [] $\alpha_{ij}, \beta_{ij}$ for $0 \leq i, j \leq \n-1$ are
	defined as in the definition of $\mathsf{D}_{\n}$,
	  \item []  $\alpha_{\n j} = \alpha_{i\n}=b$ for $1\le i\le \n,~0\le j\le \n$,
	  \item []   $\alpha_{0\n} =a$,
	  \item []  $\beta_{\n j} = \beta_{i\n}=b$ for $0\le i\le \n-1,~0\le j\le \n-1$, and
	  \item []   $\beta_{\n\n} =a$.
	  \end{itemize}
  }
	
\noindent {Clearly, ${\mathsf{C}}_{\n}$ is a normal game. Due to the structure of ${\mathsf{C}}_{\n}$,
	there hold two properties
	with corresponding consequences
	to Properties {\sf (1)} and {\sf (2)}
	from Lemma~\ref{frequent lemma}:}
\begin{itemize}

	\item
	{${\mathsf{\alpha}}_{\n j} = b$
		for all
		$j \in \{ 0, 1, \ldots, \n \}$.
		{Thus,
			${\mathsf{\mu}}_{1}(\n, j)
			=
			b$
			for all $j \in \{ 0, 1, \ldots, \n \}$.}
		Taking
		$\widehat{i} := \n$,
		this implies Property {\sf (1)}
		from Lemma~\ref{frequent lemma}
		for any mixed profile
		$(p_{1}, p_{2})$
		with $\n \in \sigma (p_{1})$.}

	{${\mathsf{\beta}}_{i\n} = b$
		for all $i \in \{ 0, \ldots, \n-1 \}$.}
	{Thus,
		${\mathsf{\mu}}_{2}(i, \n)
		=
		b$
		for all $i \in \{ 0, \ldots, \n-1 \}$.}  
	{Taking
		$\widehat{j} := \n$,
		this implies
		Property {\sf (1)}
		from Lemma~\ref{frequent lemma}
		for any mixed profile
		$(p_{1}, p_{2})$
		with 
		$\n \in \sigma (p_{2})$
		and
		$\n \not\in \sigma (p_{1})$.}

	\item
	Since $\mathsf{C}_{\n}$ is a normal game, any mixed profile
	$(p_{1}, p_{2})$ fulfills Property {\sf (2)}
	from Lemma~\ref{frequent lemma}.
	\remove{For each strategy 
		$j \in \{ 0, 1, \ldots, \n \}$
		(resp., $i \in \{ 0, 1, \ldots, \n \}$),
		there is a strategy 
		$i \in \{ 0, 1, \ldots, \n \}$
		(resp., $j \in \{ 0, 1, \ldots, \n \}$)
		with
		${\mathsf{\alpha}}_{ij}
		=
		a$
		(resp.,
		${\mathsf{\beta}}_{ij}
		=
		a$).
		(This extends the corresponding property for the game ${\mathsf{D}}_{\n}$
		to row and column $\n$.)
	This implies 
		that any mixed profile
		$(p_{1}, p_{2})$
		fulfills
		Property {\sf (2)}
		from Lemma~\ref{frequent lemma}.}

\end{itemize}
\noindent
{We show:}

\begin{theorem}
	\label{grande resort}
    Let $\n\ge 2$. Consider a unimodal valuation {\sf F}. $\mathsf{C}_{\n}$ has no
    {\sf F}-equilibrium if and only if 
    ${\mathsf{F}}\left(\frac{\textstyle 1}
    {\textstyle \n}\right) > b$ and one of the following conditions holds: 
    \begin{enumerate}[$(i)$]
    	\item $\n$ is even, $\mathsf{F}(x)< b$ for $x\ge \frac{2}{\n}$,
    	\item $\n$ is odd, $\mathsf{F}\left(\frac{2}{\n+1}\right) \ne b$  and
    		$\mathsf{F}(x)< b$ for $x\ge \frac{2}{\n-1}$.
    \end{enumerate}
\end{theorem}

\begin{proof}
		We first prove that the conditions are sufficient for non-existence. {In the proof we will make use
			of the following observation: if $\mathsf{F}(z)<b$ then $\mathsf{F}(x)$ decreases strictly monotone
			for $x\in(z,1)$; in particular, if $\mathsf{F}(z)<b$ and $z<x$, then $\mathsf{F}(x)<b$.}\smallskip\\
{\bf ``$\bm{\Leftarrow}$'':}
	{Assume, by way of contradiction,
		that ${\mathsf{C}}_{\n}$
		has an ${\mathsf{F}}$-equilibrium $\langle p_{1}, p_{2}\rangle$
		with $\n \not\in \sigma (p_{1})$
		and
		$\n \not\in \sigma (p_{2})$.}
	{Due to the structure of ${\mathsf{C}}_{\n}$,
		it follows that $\langle p_{1}, p_{2}\rangle$
		is an ${\mathsf{F}}$-equilibrium
		for ${\mathsf{D}}_{\n}$.}
	{Since
		${\mathsf{F}}(x) < b$
		for
		$x \geq \frac{2}
		{\n}$, with $\n$ even, or 
		$x \geq \frac{2}
		{\n-1}$, with $\n$ odd, by the above observation and 
		Theorem~\ref{hawaii grande}, it follows
		that
		$p_{1}(j) = p_{2}(j) = \frac{1}
		{\n}$
		with $0 \leq j \leq \n-1$.}
		This implies that
		$x_{1}(p_{1}, p_{2})
		=
		\frac{1}
		{\n}$.
	{Thus,
		${\mathsf{V}}_{1}(p_{1}, p_{2})
		=
		{\mathsf{F}}\left( \frac{\textstyle 1}
		{\textstyle \n}
		\right)
		>
		b$.
		Then,
		player $1$ could improve
		by switching to  strategy $\n$
		with 
		${\mathsf{V}}(p_{1}^{\n}, p_{2})
		=
		b$. A contradiction.} \\  
	{Assume now that
		there is an ${\mathsf{F}}$-equilibrium $\langle p_{1}, p_{2}\rangle$
		with $\n \in \sigma (p_{1})$ or $\n \in \sigma (p_{2})$.
		Since game $\mathsf{C}_{\n}$ is normal, as mentioned above, Property {\sf (2)}
		from Lemma~\ref{frequent lemma} is fulfilled.  
		There remain two cases to consider.} 
	\begin{enumerate}[{\sf (1)}]
		\item
		\underline{$\n \in \sigma (p_{1})$:}
		Then $(\sigma(p_1)\setminus\{\n\},\sigma(p_2)\setminus\{\n\})$ is 
	a $b$-block for player 1. This leads to two subcases:
	    \begin{enumerate}
	    	\item [{\sf (1.1)}] \underline{$(\sigma(p_1)\setminus\{\n\},\sigma(p_2)\setminus\{\n\})$ is 
	    		a $b$-double-block:} This implies that $\n\notin\sigma(p_2)$,
    		hence $|\sigma(p_2)\setminus\{\n\}| = |\sigma(p_2)|$. 
    		From Lemma~\ref{lem:bblocks} we have that 
    		$|\sigma(p_1)\setminus\{\n\}| + |\sigma(p_2)\setminus\{\n\}| \le \n-1$.
    		Hence, $|\sigma(p_1)| + |\sigma(p_2)| \le \n$.
    		Let $k\in[2]$ with $|\sigma(p_k)|\le |\sigma(p_{\bar{k}})|$.
    		Then $|\sigma(p_k)|\le \lfloor \frac{\n}{2}\rfloor$.
    		 If  $\n$ is even (resp., odd), then $|\sigma(p_k)|\le \frac{\n}{2}$
    		and by assumption $\mathsf{F}(x)<b$ for $x\ge \frac{2}{\n}$ 
    		(resp., $|\sigma(p_k)|\le \frac{\n-1}{2}$
    		and by assumption $\mathsf{F}(x)<b$ for $x\ge \frac{2}{\n-1}$). 
    		In both cases,  $\mathsf{F}(x)<b$ for $x\ge \frac{1}{|\sigma(p_k)|}$.
    		Lemma~\ref{frequent lemma} implies that $p_{\bar{k}}$ is not an {\sf F}-best response to $p_k$. A contradiction. 
	    	\item [{\sf (1.2)}] 
	    	\underline{$(\sigma(p_1)\setminus\{\n\},\sigma(p_2)\setminus\{\n\})$ is 
	    	not	a $b$-block for player 2:} Then, for every $j\in\sigma(p_2)$ there
    	exists $i\in\sigma(p_1)$ with $\beta_{ij}=a$. For $j\in \sigma(p_2)\setminus \{\n\}$ we get that $(j-1) \mod\n \in \sigma(p_2)\setminus \{\n\}$ and hence
    	$|\sigma(p_2)\setminus \{\n\}|\le |\sigma(p_{\bar{1}})\setminus \{\n\}|$.
    	Lemma~\ref{lem:bblocks} implies that 
    		$|\sigma(p_1)\setminus\{\n\}| + |\sigma(p_2)\setminus\{\n\}| \le \n$.
    		We consider two subcases:
    		\begin{enumerate}
    			\item [{\sf (1.2.1)}] 
    			\underline{$\n\notin\sigma(p_2)$:} Then $|\sigma(p_1)| + |\sigma(p_2)| \le \n+1$ and $|\sigma(p_2)| \le |\sigma(p_1)| - 1.$
    			Thus, $2\cdot|\sigma(p_2)| \le |\sigma(p_1)|+|\sigma(p_2)| - 1 \le \n$,
    	        which yields $|\sigma(p_2)| \leq \lfloor \frac{\n}{2}\rfloor$.
    	        Following the exact same reasoning as in case {\sf (1.1)} and using
    	        Lemma~\ref{frequent lemma}, we conclude that $p_1$ is not an 
    	        {\sf F}-best response to $p_2$. A contradiction. 
    			\item [{\sf (1.2.2)}] \underline{$\n\in\sigma(p_2)$:} Because
    			of the non-symmetry of $\mathsf{C}_\n$ we now need to employ new ideas.
    			Observe that $\n\in\sigma(p_2)$ implies $0\notin\sigma(p_1)$ 
    			(since $\alpha_{0\n}=a$ and $(\sigma(p_1)\setminus\{\n\}, \sigma(p_2)\setminus\{\n\})$ is a  $b$-block for player 1).
    			We distinguish two subcases:
    			\begin{enumerate}
    				\item [{\sf (1.2.2.1)}] \underline{$0\notin\sigma(p_2)$:}
    				For $i\in\{0,\ldots,\n-1\}$, $i\in\sigma(p_1)$ implies
    				$i\notin\sigma(p_2)$. Hence, 
    				$(\sigma(p_{1})\setminus\{\n\})\cap (\sigma(p_{2})\setminus \{\n\})=\emptyset$. 
    				Furthermore, 
    				$\sigma(p_1)\setminus\{\n\}\subseteq\{1,\ldots,\n-1\},~\sigma(p_2\setminus\{\n\}\subseteq\{1,\ldots,\n-1\}$
    				implies	$|\sigma(p_1)\setminus\{\n\}| + |\sigma(p_2)\setminus\{\n\}| \le \n-1$, which
    				implies $|\sigma(p_1)| + |\sigma(p_2)| \le \n+1$, and hence 
    				$|\sigma(p_2)|\leq \lfloor\frac{\n+1}{2}\rfloor$. {We distinguish between $\n$ even and $\n$ odd.}
    				\begin{itemize}
    					\item [--] {If $\n$ is even, then $|\sigma(p_2)|\le \frac{\n}{2}$
    					and by assumption $\mathsf{F}(x)<b$ for $x\ge \frac{2}{\n}$. Hence, 
    					$\mathsf{F}(x)<b$ for $x\ge \frac{1}{|\sigma(p_2)|}$. Lemma~\ref{frequent lemma} implies that 
    					$p_{1}$ is not an {\sf F}-best response to $p_2$. A contradiction. }   
    					\item [--] {If $\n$ is odd, then $|\sigma(p_2)|\le \frac{\n+1}{2}$. We have two subcases,
    						$|\sigma(p_2)|< \frac{\n+1}{2}$ and $|\sigma(p_2)|= \frac{\n+1}{2}$.}
    					    \begin{itemize}
    						\item [+] {\underline{$|\sigma(p_2)|\le \frac{\n+1}{2}-1$:} By assumption, $\mathsf{F}(x) < b$ for $x\ge \frac{2}{\n-1}$. 
    							Therefore, $\mathsf{F}(x) < b$ for $x\ge \frac{1}{|\sigma(p_2)|}$.
    							Lemma~\ref{frequent lemma} implies that $p_1$ is not an {\sf F}-best response to $p_2$.
    							A contradiction.}
    						\item [+] \underline{$|\sigma(p_2)| = \frac{\n+1}{2}$:} By assumption, $\mathsf{F}\left(\frac{2}{\n+1}\right)\ne b$.
    							If $\mathsf{F}\left(\frac{2}{\n+1}\right) < b$, then Lemma~\ref{frequent lemma} implies that $p_1$ is not
    							an {\sf F}-best response to $p_2$.
    							A contradiction.  \cg{For $\mathsf{F}\left(\frac{2}{\n+1}\right) > b$,  the {\sf WEEP} implies
    								that $p_1(i)=\frac{2}{\n}$ for all $i\in \sigma(p_1)$. Thus $\mathsf{V}_2(p_1,p_2)=\mathsf{F}\left(\frac{2}{\n}\right)>b$
    								and hence player 2 can improve by
    							using some strategy $j\notin\sigma(p_2)$. A contradiction.}
    					\end{itemize}
    				\end{itemize} 
    				\item [{\sf (1.2.2.2)}] \underline{$0\in\sigma(p_2)$:} Then,
    				$|\sigma(p_1)\setminus\{\n\}| + |\sigma(p_2)\setminus\{\n\}| \le \n$, which implies,
    				$|\sigma(p_1)| + |\sigma(p_2)| \le \n+2$, and hence
    				$|\sigma(p_2)|\le \lfloor\frac{\n+2}{2}\rfloor$. 
    				{Unlike the previous cases, we do not get the optimal result
    				by just applying Lemma~\ref{frequent lemma}. Instead, we will show,
    				by {using} the method in the proof of Lemma~\ref{frequent lemma}, that $p_1$ is not an {\sf F}-best response to $p_2$
    				if $\mathsf{F}(x)<b$ for $x\ge \frac{1}{|\sigma(p_2)|-1}$, deriving a contradiction.} 
    				
    				{For every $j\in\{0,\ldots,\n\}$ there is exactly
    				one strategy $i\in\{0,\ldots,\n-1\}$ with $\alpha_{ij}=a$,
    				call it $i(j)$.
    				Then, for every strategy $j\in \sigma(p_2),~i(j)\in \sigma(p_1)$.
    				It is $i(0)=i(\n)=0$ and $i(j_1) \ne i(j_2)$ for $j_1\ne j_2\ne 0$. Then, 
    				$V_1(p_1^{0},p_2)=\mathsf{F}(p_2(0)+p_2(\n))$ and 
    				for $j\in \sigma(p_2)\setminus \{0,\n\}$,   
    				$V_1(p_1^{i(j)},p_2)=\mathsf{F}(p_2(j))$. Then,
    				$p_2(0)+p_2(\n) \ge \frac{1}{|\sigma(p_2)|-1}$
    				or there exists some $j\in\sigma(p_2)\setminus \{0,\n\}$
    			with $p_2(j) \ge \frac{1}{|\sigma(p_2)|-1}$.
    			So, when $\mathsf{F}(x)<b$ for $x\ge \frac{1}{|\sigma(p_2)|-1}$,
    			player 1 can improve by switching either to strategy $0$
    			or to a strategy $i(j)$ with  $p_2(j) \ge \frac{1}{|\sigma(p_2)|-1}$. Recall that
    		$|\sigma(p_2)|\le \lfloor\frac{\n+2}{2}\rfloor$. So, player 1
    	can improve for $\mathsf{F}\left(\frac{1}{\lfloor\frac{\n}{2}\rfloor}\right)<b$. A contradiction.}
    				
    				{To better understand this case, we provide an example
    				with $\n=4$.\smallskip\\}
    				$ 		\left( \begin{array}{llllll}
    				(a, b)  & (b,a)   & (b,b)   &   (b,b) & (a,b)    \\
    				{\color{green} \fbox{(b,b)}}   & (a,b)   & {\color{green} \fbox{(b,a)}}   & (b,b)  & {\color{green} \fbox{(b,b)}}  \\
    				(b,b)   &  (b,b)     & (a,b) &  (b,a)  & (b,b)  \\
    				{\color{green} \fbox{(b,a)}}   & (b,b)    & {\color{green} \fbox{(b,b)}}   & (a,b)   & {\color{green} \fbox{(b,b)}} \\
    				{\color{green} \fbox{(b,b)}}   & (b,b)    & {\color{green} \fbox{(b,b)}}   & (b,b)   & {\color{green} \fbox{(b,a)}} \\
    				\end{array}
    				\right)$\smallskip\\
    				{Then, $\sigma(p_1)=\{1,3,4\},~\sigma(p_2)=\{0,2,4\}$, $p_2(0)+p_2(4)\ge \frac{1}{2}$ or $p_2(2)\ge \frac{1}{2}$. For $\mathsf{F}\left(\frac{1}{2}\right)<b$, player 1 can improve by switching to strategy $0$ or to strategy $2$.}			
    			\end{enumerate}
    		\end{enumerate} 
    	
	    \end{enumerate}
      \item
      \underline{$\n\notin\sigma(p_1)$, $\n \in \sigma (p_{2})$:}
      Then $(\sigma(p_1),\sigma(p_2)\setminus\{\n\})$ is 
      a $b$-block for player 2. We distinguish two cases:
      \begin{enumerate}
      \item [{\sf (2.1)}] \underline{$(\sigma(p_1),\sigma(p_2)\setminus\{\n\})$ is 
      	a $b$-double-block:}  
      From Lemma~\ref{lem:bblocks} we have that 
      $|\sigma(p_1)| + |\sigma(p_2)\setminus\{\n\}| \le \n-1$.
      Hence, $|\sigma(p_1)| + |\sigma(p_2)| \le \n$.
      Let $k\in[2]$ with $|\sigma(p_k)|\le |\sigma(p_{\bar{k}})|$.
      Then $|\sigma(p_k)|\le \lfloor \frac{\n}{2}\rfloor$.
      Following exactly the same reasoning as in case {\sf (1.1)}
      and using Lemma~\ref{frequent lemma}, we conclude that 
      $p_{\bar{k}}$ is not an {\sf F}-best response to $p_k$.
      A contradiction. 
      \item [{\sf (2.2)}] 
      \underline{$(\sigma(p_1),\sigma(p_2)\setminus\{\n\})$ is 
      	not	a $b$-block for player 1:} Then, for every $i\in\sigma(p_1)$ there
      exists $j\in\sigma(p_2)$ with $\alpha_{ij}=a$, i.e.,  
      for $i\in \sigma(p_1)\setminus\{0\}$,
     we get that $i\in \sigma(p_2)$. This implies 
      that $\sigma(p_1) \setminus \{0\}\subseteq \sigma(p_2) \setminus \{\n\}$,
      and hence $|\sigma(p_1)|\le |\sigma(p_2)|$.
      Lemma~\ref{lem:bblocks} implies  
      $|\sigma(p_1)| + |\sigma(p_2)\setminus\{\n\}| \le \n$, which
      implies $|\sigma(p_1)| + |\sigma(p_2)| \le \n + 1$, 
      and hence $|\sigma(p_1)| \le \lfloor\frac{\n+1}{2}\rfloor$.
      Following exactly the same reasoning as in case {\sf (1.2.2.1)},
      we reach a contradiction.
      \remove{ 
      If  $\n$ is even (resp., odd), then $|\sigma(p_1)|\le \frac{\n}{2}$
      and by assumption $\mathsf{F}(x)<b$ for $x\ge \frac{2}{\n}$ 
      (resp., $|\sigma(p_1)|\le \frac{\n+1}{2}$
      and by assumption $\mathsf{F}(x)<b$ for $x\ge \frac{2}{\n+1}$). 
      In both cases,  $\mathsf{F}(x)<b$ for $x\ge \frac{1}{|\sigma(p_1)|}$.
      Lemma~\ref{frequent lemma} implies that $p_{2}$ is not an {\sf F}-best response to $p_1$,
      a contradiction.} 
      \end{enumerate}
	\end{enumerate}
{\noindent {\bf ``$\bm{\Rightarrow}$'':} We now
show that the conditions are necessary, that is, if they do not hold, then there is an {\sf F}-equilibrium.}
\cg{We first show the necessity of  ${\mathsf{F}}\left(\frac{1}
	{\n}\right) > b$. If ${\mathsf{F}}\left(\frac{1}
	{\n}\right) \le b$ then $\mathsf{C}_{\n}$ has an {\sf F}-equilibrium $\langle p_1,p_2\rangle$ with}
	\begin{itemize}
		\item [] \cg{$\sigma(p_1)=\{i\in \{0,\ldots,\n-1\}\},~p_1(i)=\frac{1}{m}$ for all $i\in\sigma(p_1)$}\smallskip\\
			\cg{$\sigma(p_2)=\{j\in \{0,\ldots,\n-1\}\},~p_2(j)=\frac{1}{m}$ for all $j\in\sigma(p_2)$}\smallskip\\
			\cg{$|\sigma(p_1)|=|\sigma(p_2)| =\n$, and {\bf p} is uniform on $\sigma(p_k)$ for $k\in[2]$ (this is the case
				$\n\notin \sigma(p_1)$ and $\n\notin \sigma(p_2)$).}
	\end{itemize}		
\cg{Then, given that ${\mathsf{F}}\left(\frac{1}
	{\n}\right) > b$, the remaining conditions do not hold if and only if (there exists $\n$ even with 
	$\mathsf{F}\left(\frac{2}{\n}\right)\ge b$) or (there exists $\n$ odd with $\mathsf{F}\left(\frac{2}{\n+1}\right)= b$ or 
	$\mathsf{F}\left(\frac{2}{\n-1}\right)\ge b$).}
\begin{enumerate}[{\sf (1)}]
	\item For $\n$ even and $\mathsf{F}\left(\frac{2}{\n}\right) \ge b$, $\mathsf{C}_{\n}$ has an {\sf F}-equilibrium $\langle p_1,p_2\rangle$ with
	\begin{itemize}
		\item [] \cg{$\sigma(p_1)=\{\n\}\cup \{i~:~\frac{\n}{2}\le i \le \n-2\},$\smallskip\\
		$\sigma(p_2)=\{j~:~0\le j \le \frac{\n}{2}-1\},$\smallskip\\
		$|\sigma(p_1)|=|\sigma(p_2)| =\frac{\n}{2}$, and {\bf p} is uniform on $\sigma(p_k)$ for $k\in[2]$.}
	\end{itemize}
     	\cg{Observe that we are in case {\sf (1.1)}: $\n\in\sigma(p_1),~\n\notin\sigma(p_2)$, and $(\sigma(p_1)\setminus\{\n\},\sigma(p_2)\setminus\{\n\})$ is a $b$-double-block.}
     
     {\color{black}
     \noindent For example, for $\n=4$ and the above case, $\mathsf{C}_{4}~=$
     $ 		\left( \begin{array}{llllll}
     (a,b)  & (b,a) & (b,b) & (b,b) & (a,b) \\
     (b, b)  & (a,b) & (b,a)   & (b,b)   &   (b,b)   \\
     {\color{green} \fbox{(b,b)}}   &  {\color{green} \fbox{(b,b)}}  & (a,b)   & (b,a) & (b,b)    \\
     (b,a)   & (b,b)    & (b,b)  & (a,b)  & (b,b) \\
     {\color{green} \fbox{(b,b)}}   &  {\color{green} \fbox{(b,b)}} & (b,b)   &   (b,b) & (b,a)    \\
     \end{array}
     \right)\, .$}
	\item For $\n$ odd and $\mathsf{F}\left(\frac{2}{\n+1}\right) =b$, $\mathsf{C}_{\n}$ has an {\sf F}-equilibrium $\langle p_1,p_2\rangle$ with
	\begin{itemize}
		\item [] $\sigma(p_1)=\{i\in \{0,\ldots,\n\}~:~i~\text{odd}\},$\smallskip\\
		$\sigma(p_2)=\{\n\} \cup \{j\in 2,\ldots,\n-1\}~:~j~\text{even}\},$\smallskip\\
		$|\sigma(p_1)|=|\sigma(p_2)| =\frac{\n+1}{2}$, \cg{and {\bf p} is uniform on $\sigma(p_k)$ for $k\in[2]$.}
	\end{itemize}
	{Observe that we are in case {\sf (1.2.2.1)}: $\n\in\sigma(p_1),~\n\in\sigma(p_2),~\text{and}~0\notin\sigma(p_2)$.}
	
	\noindent For example, for $\n=5$ and the above case, $\mathsf{C}_{5}~=$
	$ 		\left( \begin{array}{llllll}
	(a,b)  & (b,a) & (b,b) & (b,b) & (b,b) & (a,b) \\
	(b, b)  & (a,b) & {\color{green} \fbox{(b,a)}}   & (b,b)   &   {\color{green} \fbox{(b,b)}}   &   {\color{green} \fbox{(b,b)}}  \\
	(b,b)   & (b,b)   & (a,b)   & (b,a) & (b,b)     & (b,b)    \\
	(b,b)   &  (b,b)   & {\color{green} \fbox{(b,b)}}     & (a,b) &  {\color{green} \fbox{(b,a)}} &   {\color{green} \fbox{(b,b)}}   \\
	(b,a)   & (b,b)    & (b,b)  & (b,b)  & (a,b) & (b,b)  \\
	(b,b)   &  (b,b)   & {\color{green} \fbox{(b,b)}}     & (b,b) &  {\color{green} \fbox{(b,b)}} &   {\color{green} \fbox{(b,a)}}   \\
	\end{array}
	\right)\, .$
	\item For $\n$ odd and $\mathsf{F}\left(\frac{2}{\n-1}\right) \ge b$, $\mathsf{C}_{\n}$ has an {\sf F}-equilibrium $\langle p_1,p_2\rangle$ with
	\begin{itemize}
		\item [] \cg{$\sigma(p_1)=\{\n\}\cup \{i~:~\frac{\n+1}{2}\le i \le \n-2\},$\smallskip\\
			$\sigma(p_2)=\{j~:~0\le j \le \frac{\n-3}{2}\},$\smallskip\\
			$|\sigma(p_1)|=|\sigma(p_2)| =\frac{\n-1}{2}$, and {\bf p} is uniform on $\sigma(p_k)$ for $k\in[2]$.}
	\end{itemize}
\cg{Observe that we are again in case {\sf (1.1)}: $\n\in\sigma(p_1),~\n\notin\sigma(p_2)$, and $(\sigma(p_1)\setminus\{\n\},\sigma(p_2)\setminus\{\n\})$ is a $b$-double-block.}

{\color{black}
\noindent For example, for $\n=5$ and the above case, $\mathsf{C}_{5}~=$
$ 		\left( \begin{array}{llllll}
(a,b)  & (b,a) & (b,b) & (b,b) & (b,b) & (a,b) \\
(b, b)  & (a,b) &  (b, a)  & (b,b)   &  (b, b)  &  (b, b)   \\
(b,b)   & (b,b)   & (a,b)   & (b,a) & (b,b)     & (b,b)    \\
{\color{green} \fbox{(b,b)}}   &  {\color{green} \fbox{(b,b)}} &  (b,b)      & (a,b) &  (b,a)   &    (b,b)    \\
(b,a)   & (b,b)    & (b,b)  & (b,b)  & (a,b) & (b,b)  \\
{\color{green} \fbox{(b,b)}}   &  {\color{green} \fbox{(b,b)}}  &  (b,b)     & (b,b) &  (b,b)   &   (b,a)     \\
\end{array}
\right)\, .$  }
\end{enumerate}
This completes the proof.\qed
\end{proof}


\noindent {Now, putting together the results of Theorems~\ref{2 players 2 values ppad hard} and \ref{grande resort}, we
	obtain a panorama on the (in)existence of {\sf F}-equilibria, summarized in the following theorem:}

\begin{theorem}
	\label{panorama}
	{For a unimodal valuation {\sf F} the following properties hold:}
	\begin{enumerate}[$(i)$]
		\item {If $x_0(\mathsf{F}) = 0$ then an {\sf F}-equilibrium exists for all 2-players, 2-values games
			and its computation is ${\mathcal{PPAD}}$-hard.}\label{pan-i}
		
		\item {If $x_0(\mathsf{F}) > 0$ and $\mathsf{F}\left(\frac{\textstyle 1}{\textstyle 2}\right)\ne b$,
			then there exists a normal 2-players, 2-values game without {\sf F}-equilibrium. }
		
	\end{enumerate}
\end{theorem}

\begin{proof}
	Property~\ref{pan-i} follows directly from Theorem~\ref{2 players 2 values ppad hard}. 
		{Now consider the case $x_0(\mathsf{F}) > 0$ and $\mathsf{F}\left(\frac{1}{2}\right)\ne b$. 
		Then, $x_1=x_1(\mathsf{F})$ exists with $\mathsf{F}(x_1)=b$.}
		{Choose $\n\in\mathbb{N}$ with 
			$\frac{1}{\n} < x_1(\mathsf{F}) \le \frac{1}{\n-1}$. Then,
			$\mathsf{F}\left(\frac{1}{\n}\right) > b$,
			$\mathsf{F}\left(\frac{1}{\n-1}\right) \le b$
			and $\mathsf{F}(x)$ decreases strictly monotone for 
			$x\in \left[\frac{1}{\n-1},1\right]$. We distinguish between $\n$ even and $\n$ odd.}
			\begin{itemize}
				\item If $\n\ge 4$ is even, then $\frac{2}{\n}>\frac{1}{\n-1}$ and hence 
				$\mathsf{F}\left(\frac{2}{\n}\right) < \mathsf{F}\left(\frac{1}{\n-1}\right) \le b$. \cg{For $\n=2$, we have 
				$\mathsf{F}\left(\frac{2}{\n}\right)=\mathsf{F}(1) = a < b$.} 
				\item If $\n\ge 5$ is odd, then 	$\mathsf{F}\left(\frac{2}{\n-1}\right) < \mathsf{F}\left(\frac{1}{\n-1}\right) \leq b.$
				Furthermore, $\mathsf{F}\left(\frac{2}{\n+1}\right)\ne b$ (otherwise, $x_1(\mathsf{F})=\frac{2}{\n+1}>\frac{1}{\n-1}$, contradicting
				the choice of $\n$ s.t. $x_1(\mathsf{F}) \le \frac{1}{\n-1}$). \cg{For $\n=3$ we have $\mathsf{F}\left(\frac{2}{\n-1}\right)=\mathsf{F}(1) = a< b$; furthermore, $\frac{2}{\n+1} = \frac{1}{2}$, and by assumption, $\mathsf{F}\left(\frac{1}{2}\right)\ne b$.}  
			\end{itemize}
		{So, for all $\n\ge 2$, the conditions of Theorem~\ref{grande resort} hold,  and hence the game 
		$\mathsf{C}_m$ has no {\sf F}-equilibrium.}~\qed 
	\end{proof}

\noindent\cg{Theorem~\ref{panorama} leaves open the case on whether there exists, or not, an {\sf F}-equilibrium when
	$\mathsf{F}\left(\frac{1}{2}\right) = b$. This is the subject of Section~\ref{sec:existence of equilibria}.}

\section{Existence of Equilibria} 
\label{sec:existence of equilibria}

\cg{In this section we show the existence of equilibria for a 2-players,
2-values game {\sf G} when {\sf G} is a normal game with $\mathsf{F}\left(\frac{1}{2}\right) = b$ (Section~\ref{equal b})
and when players have three strategies (Section~\ref{three strategies}), complementing the investigation of the previous section.}\vspace{-1em}
}

\subsection{Normal Games with $\mathbf{\mathsf{F}\left(\frac{1}{2}\right)=b}$} 
	\label{equal b}                                      

We now consider normal games. Specifically, in this section $\mathsf{G}$ is a normal 2-players game with two values $a,b,~a<b$.
Furthermore, we assume that 
${{\mathsf{F}}}\left( \frac{1}{2}\right)=b$.

Recall that in a normal game, for each row $i$ there is exactly one column $j$ with $\beta_{ij}=a$, and for each column $j$ there is exactly one row $i$ with $\alpha_{ij}=a$.
For convenience of presentation, in this section we define the set of strategies over the set 
$[n]=\{1,\ldots,n\}$. Then,  
for $i\in[n]$ 
	 let $col(i)$ be the uniquely determined column with $\beta_{icol(i)}=a$.
For $j\in[n]$ let $row(j)$ be the uniquely determined row with $\alpha_{row(j)j}=a$. 
Set:  \vspace{-.2em}
$$C = \{ col(i)~:~i\in[n]  \} \text{ and } R = \{ row(j)~:~j\in[n]  \}.\vspace{-.2em}$$

Note that in a normal game,  always $|C|\ge 2,~|R|\ge 2$ holds.  
To see this, assume that $|C|=1$; then $\exists j~\forall i: \gamma(i,j)=(b,a)$. But in a normal game $\gamma(row(j),j)=(a,b)$, a contradiction. In the same way $|R|\ge 2$ follows.  

In the remainder of the section we will show that $\mathsf{G}$ 
always has an ${{\mathsf{F}}}$-equilibrium whose support sets both have size 2, and that this equilibrium can be computed in linear time, with respect to the input of the problem (Theorem~\ref{thm:winning-pair} and Corollary~\ref{cor:final-result}).

For the complexity consideration we use the condensed input form where the game is given
by the the two arrays $Col$ and $Row$ with $Col(i)=col(i)$ and $Row(i)=row(i)$ for all 
$i\in [n]$. Using this input convention, the sets $C$ and $R$, represented as Boolean arrays,
can be computed in time $O(n)$. We will show that a {\em winning pair} (c.f. Definition~\ref{def:winpair}) can be found in time
$O(n)$. Note that if the input is given by the two matrices $(\alpha_{ij})_{1\leq i,j\leq n}$,
$(\beta_{ij})_{1\leq i,j\leq n}$, then the input needs $O(n^2)$ space. In this case, the arrays
$Col$ and $Row$ can be computed in time $O(n^2)$ and using afterwards our algorithm leads
to an algorithm for computing a winning pair in time $O(n^2)$. Since the input is of size
$O(n^2)$, this is also a linear time algorithm. We now define the notion of a winning pair.\vspace{-1em}

\begin{center}
	\fbox{
		\begin{minipage}{6in}
			\begin{definition}[Winning Pair]
				\label{def:winpair}
				A pair $(\sigma_1,\sigma_2)$ with $\sigma_1,\sigma_2 \subseteq [n]$, $|\sigma_1|=|\sigma_2|=2$,\break $\sigma_1=\{i_1,i_2\},\sigma_2=\{j_1,j_2\}$ is called
				a {\it winning pair} if the three 
				following conditions are fulfilled:
				
				\begin{enumerate}[(i)]

					\item $\left( 
					\begin{array}{ll}
					\alpha_{i_1j_1} & \alpha_{i_1j_2}  \\
					\alpha_{i_2j_1} & \alpha_{i_2j_2}  
					\end{array}
					\right)$ and 
					$\left( 
					\begin{array}{ll}
					\beta_{i_1j_1} & \beta_{i_1j_2}  \\
					\beta_{i_2j_1} & \beta_{i_2j_2}  
					\end{array}
					\right)$ have the form
					$\left( 
					\begin{array}{ll}
					a & b  \\
					b & a  
					\end{array}
					\right)$ or 
					$\left( 
					\begin{array}{ll}
					b & a  \\
					a & b  
					\end{array}
					\right)$  \\ or $\left( 
					\begin{array}{ll}
					b & b  \\
					b & b  
					\end{array}
					\right)$

					\item $row(j_1) \ne row(j_2)$

					\item $col(i_1) \ne col(i_2)$

				\end{enumerate}
				
			\end{definition}  
		\end{minipage}
	}
\end{center}

\noindent We first show:

\begin{lemma}
	\label{lem:psFequil}
	Let  $(\sigma_1,\sigma_2)$, $\sigma_1=\{i_1,i_2\},\sigma_2=\{j_1,j_2\}$ be a
	winning pair and define a strategy vector $\langle p_1,p_2\rangle$ 
	by $p_{1,i_1} = p_{1,i_2} = 1/2$, $p_{2,j_1} = p_{2,j_2} = 1/2$.
	Then, $\langle p_1,p_2\rangle$ is an  ${{\mathsf{F}}}$-equilibrium.
\end{lemma}

\begin{proof}
	By assumption, ${{\mathsf{F}}}\left( \frac{1}{2}\right)=b$
	and therefore $\mathsf{V}_1(p_1,p_2) = \mathsf{V}_2(p_1,p_2) = b$.
	
	\noindent If player 1 chooses an arbitrary pure strategy $i$, then, due to condition 
	$(ii)$, player 1 has two $b$'s or one $a$ and one $b$.  
	So $\mathsf{V}_1(p^{i}_1,p_2) = b$ for all $i\in[n]$. Hence, player 1 cannot improve by
	choosing strategy $i$. Likewise, due to condition $(iii)$, player 2 cannot improve
	by choosing an arbitrary pure strategy $j$.\qed
\end{proof}

\noindent In the following lemmas we will show results about the distribution of $(a,b)$ entries and
$(b,a)$ entries in the case where there exists no winning pair. 

\begin{lemma}
	\label{lem:s1xs2-winpair}
	Let the two pairs $\sigma_1=\{i_1,i_2\},\sigma_2=\{j_1,j_2\}$ fulfill the following properties:
	
	\begin{enumerate}[(i)]
		
		\item $\{i_1,i_2\} \cap \{row(j_1),row(j_2)\} = \emptyset$\\
		$\{j_1,j_2\} \cap \{col(i_1),col(i_2)\} = \emptyset$
		
		\item $row(j_1) \ne row(j_2)$
		
		\item $col(i_1) \ne col(i_2)$
		
	\end{enumerate}	
	
	\noindent Then, $\sigma_1 \times \sigma_2$ is a winning pair.
\end{lemma}

\begin{proof}
	Observe that properties $(ii)$ and $(iii)$ are equal to 
	the conditions $(ii)$ and $(iii)$ of Definition~\ref{def:winpair}. 
	Therefore, only condition $(i)$ has to be verified.
	
	\begin{itemize}
		
		\item $\alpha_{row(j_v)j_v} = a$ for $v\in\{1,2\}$ and normality
		implies that $\alpha_{ij_v}=b$ for all $i\ne row(j_v)$. So the first 
		line of property $(i)$ implies that $\alpha_{i_\mu j_v} = b$ for $\mu,v \in \{1,2\}$.
		
		\item  $\beta_{i_v col(i_v)} = a$ for $v\in\{1,2\}$ and normality
		implies that $\beta_{i_v j}=b$ for all $j\ne col(i_v)$. So the second 
		line of property $(i)$ implies that $\beta_{i_\mu j_v} = b$ for $\mu,v \in \{1,2\}$.
	\end{itemize}
	
	\noindent The proof is demonstrated in Figure~\ref{fig:proofL2}.\qed
\end{proof}

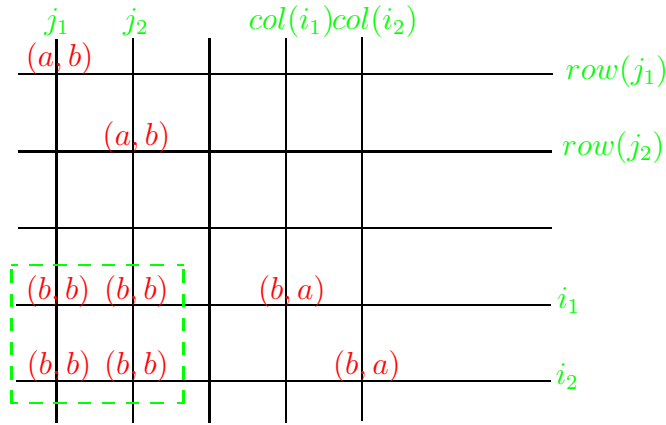
\begin{figure}[h]
	%
	%
	\setlength{\unitlength}{2000sp}%
	\begingroup\makeatletter\ifx\SetFigFont\undefined%
	\gdef\SetFigFont#1#2#3#4#5{%
		\reset@font\fontsize{#1}{#2pt}%
		\fontfamily{#3}\fontseries{#4}\fontshape{#5}%
		\selectfont}%
	\fi\endgroup%
	\begin{picture}(6929,5078)(827,-4913)
	{\color[rgb]{0,0,0}\put(1426,-136){\line( 0,-1){4755}}
	}%
	{\color[rgb]{0,0,0}\put(2378,-129){\line( 0,-1){4755}}
	}%
	{\color[rgb]{0,0,0}\put(3323,-129){\line( 0,-1){4755}}
	}%
	{\color[rgb]{0,0,0}\put(4268,-129){\line( 0,-1){4755}}
	}%
	{\color[rgb]{0,0,0}\put(5213,-114){\line( 0,-1){4755}}
	}%
	{\color[rgb]{0,0,0}\put(961,-571){\line( 1, 0){6615}}
	}%
	{\color[rgb]{0,0,0}\put(954,-1531){\line( 1, 0){6615}}
	}%
	{\color[rgb]{0,0,0}\put(954,-2476){\line( 1, 0){6615}}
	}%
	{\color[rgb]{0,0,0}\put(939,-3436){\line( 1, 0){6615}}
	}%
	{\color[rgb]{0,0,0}\put(939,-4381){\line( 1, 0){6615}}
	}%
	{\color[rgb]{0,1,0}\put(871,-4651){\dashbox{150}(2130,1710){}}
	}%
	\put(7741,-631){\makebox(0,0)[lb]{\smash{{\SetFigFont{12}{14.4}{\rmdefault}{\mddefault}{\updefault}{\color[rgb]{0,1,0}$row(j_1)$}%
	}}}}
	\put(7696,-1569){\makebox(0,0)[lb]{\smash{{\SetFigFont{12}{14.4}{\rmdefault}{\mddefault}{\updefault}{\color[rgb]{0,1,0}$row(j_2)$}%
	}}}}
	\put(7636,-3466){\makebox(0,0)[lb]{\smash{{\SetFigFont{12}{14.4}{\rmdefault}{\mddefault}{\updefault}{\color[rgb]{0,1,0}$i_1$}%
	}}}}
	\put(7614,-4419){\makebox(0,0)[lb]{\smash{{\SetFigFont{12}{14.4}{\rmdefault}{\mddefault}{\updefault}{\color[rgb]{0,1,0}$i_2$}%
	}}}}
	\put(1276, -9){\makebox(0,0)[lb]{\smash{{\SetFigFont{12}{14.4}{\rmdefault}{\mddefault}{\updefault}{\color[rgb]{0,1,0}$j_1$}%
	}}}}
	\put(2254, -9){\makebox(0,0)[lb]{\smash{{\SetFigFont{12}{14.4}{\rmdefault}{\mddefault}{\updefault}{\color[rgb]{0,1,0}$j_2$}%
	}}}}
	\put(3811, -9){\makebox(0,0)[lb]{\smash{{\SetFigFont{12}{14.4}{\rmdefault}{\mddefault}{\updefault}{\color[rgb]{0,1,0}$col(i_1)$}%
	}}}}
	\put(4846,-9){\makebox(0,0)[lb]{\smash{{\SetFigFont{12}{14.4}{\rmdefault}{\mddefault}{\updefault}{\color[rgb]{0,1,0}$col(i_2)$}%
	}}}}
	\put(1050,-466){\makebox(0,0)[lb]{\smash{{\SetFigFont{12}{14.4}{\rmdefault}{\mddefault}{\updefault}{\color[rgb]{1,0,0}$(a,b)$}%
	}}}}
	\put(2000,-1449){\makebox(0,0)[lb]{\smash{{\SetFigFont{12}{14.4}{\rmdefault}{\mddefault}{\updefault}{\color[rgb]{1,0,0}$(a,b)$}%
	}}}}
	\put(3925,-3355){\makebox(0,0)[lb]{\smash{{\SetFigFont{12}{14.4}{\rmdefault}{\mddefault}{\updefault}{\color[rgb]{1,0,0}$(b,a)$}%
	}}}}
	\put(4855,-4284){\makebox(0,0)[lb]{\smash{{\SetFigFont{12}{14.4}{\rmdefault}{\mddefault}{\updefault}{\color[rgb]{1,0,0}$(b,a)$}%
	}}}}
	\put(1060,-4284){\makebox(0,0)[lb]{\smash{{\SetFigFont{12}{14.4}{\rmdefault}{\mddefault}{\updefault}{\color[rgb]{1,0,0}$(b,b)$}%
	}}}}
	\put(2020,-4284){\makebox(0,0)[lb]{\smash{{\SetFigFont{12}{14.4}{\rmdefault}{\mddefault}{\updefault}{\color[rgb]{1,0,0}$(b,b)$}%
	}}}}
	\put(1050,-3355){\makebox(0,0)[lb]{\smash{{\SetFigFont{12}{14.4}{\rmdefault}{\mddefault}{\updefault}{\color[rgb]{1,0,0}$(b,b)$}%
	}}}}
	\put(2020,-3355){\makebox(0,0)[lb]{\smash{{\SetFigFont{12}{14.4}{\rmdefault}{\mddefault}{\updefault}{\color[rgb]{1,0,0}$(b,b)$}%
	}}}}
	\end{picture}%
	\caption{Demonstrating the proof of Lemma~\ref{lem:s1xs2-winpair}.}
	\label{fig:proofL2}
\end{figure}

\begin{lemma}
	\label{lem:alg-inout}
	
	There exists an algorithm running in time $O(n)$ with the following input-output
	relation.
	
	\begin{enumerate}[(i)]
		
		\item input: $j_1,j_2 \in [n]$ with $row(j_1) \ne row(j_2)$\\
		output:
		
		\begin{enumerate}
			
			\item[(i/1)] column $j\in [n]$ such that 
			$\{col(i)~:~ i \ne row(j_1), i \ne row(j_2)\} \subset \{j_1,j_2,j\}$, or
			
			\item[(i/2)] rows $i_1,i_2\in [n]$, $i_1\ne i_2$, such that
			$\{i_1,i_2\} \times \{j_1,j_2\}$ is a winning pair. 
			
		\end{enumerate}		
		
		\item input: $i_1,i_2 \in [n]$ with $col(i_1) \ne col(i_2)$\\
		output:
		
		\begin{enumerate}
			
			\item[(ii/1)] row $i\in [n]$ such that 
			$\{row(j)~:~ j \ne col(i_1), j \ne col(i_2)\} \subset \{i_1,i_2,i\}$, or
			
			\item[(ii/2)] columns $j_1,j_2\in [n]$, $j_1\ne j_2$, such that
			$\{i_1,i_2\} \times \{j_1,j_2\}$ is a winning pair. 
			
		\end{enumerate}		
		
	\end{enumerate}	
	
\end{lemma}

\begin{proof}
	Due to the symmetry of rows and columns, it is sufficient to prove part $(i)$.
	
	\noindent We describe the algorithm, that starts with the input $j_1,j_2$ with 
	$row(j_1)\ne row(j_2)$:
	
	\begin{itemize}
		
		\item It checks $col(i)$, $i\in [n]$.
		
		\item If $\exists j\in [n]$ with $\{col(i)~:~ i \ne row(j_1), i \ne row(j_2)\} \subset \{j_1,j_2,j\}$ then it exists $(i/1)$ and outputs $j$. 
		
		\item Otherwise the algorithm has found some $A\subseteq [n]$ with $|A|=2$, $A\cap\{j_1,j_2\}=\emptyset$ and
		$A\subseteq \{i~:~i\ne~row(j_1),~i\ne~row(j_2)\}$. Take some $i_1,i_2$ with $A = \{col(i_1), col(i_2)\}$. Then $col(i_1)\ne col(i_2)$.
		Due to Lemma~\ref{lem:s1xs2-winpair}, $\{i_1,i_2\} \times \{j_1,j_2\}$ is a winning pair. Therefore, the algorithm 
		exits $(i/2)$ and outputs $i_1,i_2$.
		
	\end{itemize}	
	
	\noindent The algorithm makes its decision following just one scan of $col(i)$, $i\in [n]$. So, the time bound is $O(n$).\qed
\end{proof}

\begin{lemma}
	\label{lem:alg-n5}
	
	There exists an algorithm that computes a winning pair in time $O(n)$ provided that $n\ge 5$ and $(|R|\ge 4$ or $|C|\ge 4)$.
	
\end{lemma}

\begin{proof}
	Due to the symmetry of rows and columns, it is sufficient to describe the algorithm 
	in the case that $|R|\ge 4$. 
	
	\begin{itemize}
		
		\item The algorithm scans $row(j), j\in [n]$ and determines $j_v, v\in[4]$, 
		with $|\{row(j_v)~:~ v\in[n]\}|=4$. 
		\item Then it simulates the algorithm described in the proof of Lemma~\ref{lem:alg-inout} for the six $2$-sets $\{j_v,j_{\mu}\}, v\ne \mu, v,\mu\in [4]$.
		
	\end{itemize}
	Overall, the algorithm runs in time $O(n)$. In the remainder of the proof we will show
	that in the case $n\ge 5$, for at least one of the pairs $\{j_v,j_{\mu}\}$ exit $(i/2)$
	is reached, and thus a winning pair is found. 
	
	We will do so, by proof by contradiction: assume that $n\ge 5$ and for all six $2$-sets $\{j_v,j_{\mu}\}$ exit $(i/1)$ is reached. 
	
	Let $j_v, v\in[4]$, with $|\{row(j_v)~:~ v\in[4]\}|=4$. We apply Lemma~\ref{lem:alg-inout} 
	six times, i.e., for each pair $\{j_v,j_{\mu}\}, v\ne \mu, v,\mu\in [4]$. 
	Due to Lemma~\ref{lem:alg-inout}, there exist constructs $j_{\{v,\mu\}}$ 
	such that $col(i)\in \{j_v,j_{\mu},j_{\{v,\mu\}}\}$ for all $i \not\in \{row(j_{\mu}),row(j_v)\}$.
	
	First choose $i\not\in \{row(j_{\mu})~:~ \mu\in[4]\}$. To be able to do this, we need the
	assumption that $n\ge 5$. For each $2$-set $\{\mu,v\} \subset [4]$ we denote with 
	$\overline{\{\mu,v\}}$ the complement of $\{\mu,v\}$  in $[4]$. 
	Then, due to Lemma~\ref{lem:alg-inout}, 
	
	$$col(i)\in (\{j_k~:~ k\in \{\mu,v\}\} \cup \{j_{\{v,\mu\}}\}) \cap 
	(\{j_k~:~ k\in \overline{\{\mu,v\}}\} \cup \{j_{\overline{\{v,\mu\}}}\}).$$
		Since $\{\mu,v\} \cap \overline{\{\mu,v\}} = \emptyset$, this implies that 
	$col(i)=j_{\{v,\mu\}} = j_{\overline{\{v,\mu\}}}.$ Now, since the $2$-set
	$\{\mu,v\}\subset [4]$ were chosen arbitrarily, this shows that all the 
	$j_{\{v,\mu\}}, \{v,\mu\} \subset [4]$ have the same value. Let this value
	be $\hat{j}$, i.e., $col(i)=\hat{j}$ for $i\not\in \{row(j_{\mu})~:~ \mu\in [4]\}$.
	
	Now consider $i=row(j_{\rho})$ for some $\rho\in[4]$. Then information about $col(i)$ is
	given by Lemma~\ref{lem:alg-inout} using the 3 pairs $\{\mu,v\}\subset [4]$, 
	$\rho\not\in \{\mu,v\}$, i.e., for all $i\in [n]$, 
	$\displaystyle col(i)\in \bigcap_{\{\mu,v\}\subset [4]\setminus\{\rho\}} \{j_{\mu},j_v,\hat{j}\} =\hat{j},$ contradicting the fact that $|C|\ge 2$ (since the game is normal). \qed 
\end{proof}

\begin{lemma}
	\label{lem:alg-winpair}
	There exists an algorithm that computes a winning pair in time $O(n)$, provided that 
	$|\{col(i)~:~ i\not\in R\}| \ge 2$, and $|\{row(j)~:~ j\not\in C\}| \ge 2$.
\end{lemma}

\begin{proof}
	We describe the algorithm: it scans $col(i), row(i), 1\leq i\leq n$, and determines \break
	$i_1,i_2\not\in R$ with $col(i_1)\ne col(i_2)$ and $j_1,j_2\not\in C$ with $row(j_1)\ne row(j_2)$. This takes $O(n)$ time.
	Since $i_1,i_2\not\in R$, $\{i_1,i_2\} \cap \{row(j_1),row(j_2)\} = \emptyset$, and 
	since $j_1,j_2\not\in C$, $\{j_1,j_2\} \cap \{col(i_1),col(i_2)\} = \emptyset$.
	Therefore, due to Lemma~\ref{lem:s1xs2-winpair}, $\{i_1,i_2\} \times \{j_1,j_2\}$
	is a winning pair.\qed 
\end{proof}

{\color{black}
\begin{lemma}
	\label{lem:nis4}
	If $n=4$ and $((|C|=4$ and $|R|\le 3)$ or $(|R|=4$ and $|C|\le 3))$, then there exists a winning pair.
\end{lemma}

\begin{proof}
	Without loss of generality, assume that $|C|=4$ and $|R|\le 3$. Now, we order the rows so that $R=\{1,2,3\}$,
    and we order the columns so that $\gamma_{ii}=(b,a)$ for all $i\in[4]$, i.e., $row(j)=j$ for all $j\in[4]$.
    If $|R|=2$, then $col(1)=2,~col(2)=1$ and $\gamma_{ij}=(b,b)$ for $i\in \{3,4\}$, $j\in\{1,2\}$. So, 
    $\{3,4\}\times\{1,2\}$ is a winning pair.
    
   \noindent  Now let $|R|=3$. Then,\smallskip\\
    $  		{\mathsf{G}} =  
     		\left( \begin{array}{llllll}
    (b,a)  & \cdot   & \cdot   &   \cdot    \\
    \cdot   & (b,a)   & \cdot   & \cdot    \\
    \cdot   &  \cdot     & (b,a) &  \cdot    \\
    (b,b)   & (b,b)    & (b,b)   & (b,a)   \\
    \end{array}
    \right)\, .$  \smallskip\\ 

 \noindent We give special attention to the diagonal block\smallskip\\
    $  		D_i =  
    \left( \begin{array}{llllll}
    \gamma_{ii}  & \gamma_{i(i+1)}   \\
    \gamma_{(i+1)i}   & \gamma_{(i+1)(i+1)} \\
    \end{array}
    \right)\, $ for $i\in[3].$\smallskip\\
    
    \noindent Recall that $\gamma_{ii}=(b,a)$ for all $i\in[4]$. We say that $D_i$ is of type 1 if
    $\gamma_{i(i+1)}=\gamma_{(i+1)i}=(a,b)$ and of type 2 if $\gamma_{i(i+1)}=\gamma_{(i+1)i}=(b,b)$.
    Note that $D_i$ is a winning pair if and only if $D_i$ is of type 1 or ($D_i$ is of type 2 and
    $row(i)\ne row(i+1)$). We will construct all matrices whose $2\times 2$ diagonal blocks are 
    not a winning pair and name the winning pair of these matrices. The proof is divided into two steps: 
    \begin{enumerate}[{\sf (1)}]
    	\item In the first step we construct all matrices without $2\times 2$ diagonal blocks of type 1 or type 2.
    	Note that in this case, $\gamma_{34}=(a,b)$ and $\gamma_{i4}=(b,b)$ for $i\in[2]$ and that
    	the matrix is determined uniquely by $\gamma_{12}$. There are two cases:
    	\begin{itemize}
    		\item [{\sf (1.1)}] $\gamma_{12}=(a,b)$. Then 
    		$\left( \begin{array}{llllll}
    		(b,a)  &  (b,b)  & (a,b)   &   (b,b)    \\
    		{\color{green}\fbox{(a,b)}}   & (b,a)   & (b,b)   & {\color{green}\fbox{(b,b)}}    \\
    		{\color{green}\fbox{(b,b)}}   &  (a,b)    & (b,a) &  {\color{green}\fbox{(a,b)}}    \\
    		(b,b)   & (b,b)    & (b,b)   & (b,a)   \\
    		\end{array}
    		\right)\,$  
    		and $\{2,3\}\times\{1,4\}$ is a winning pair.
    		\item [{\sf (1.2)}] $\gamma_{12}=(b,b)$. Then 
    		$\left( \begin{array}{llllll}
    			(b,a)  &   {\color{green}\fbox{(a,b)}}  & (b,b)   &  {\color{green}\fbox{(b,b)}}    \\
    		(b,b)   & (b,a)   & (a,b)   & (b,b)    \\
    		(a,b)   &  {\color{green}\fbox{(b,b)}}    & (b,a) &  {\color{green}\fbox{(a,b)}}    \\
    		(b,b)   & (b,b)    & (b,b)   & (b,a)   \\
    		\end{array}
    		\right)\,$  
    		and $\{1,3\}\times \{2,4\}$ is a winning pair.
    	\end{itemize}
         \item In this step we allow in the construction of the matrices no blocks $D_i$ of type 1 and 
         a block of type 2 only if $row(i)=row(i+1)$. Since $|R|=3$, only one such $i$ can exist. Furthermore,
         if $row(3)\ne row(4)$, then $row(4)=3$. There are four cases: 
         \begin{itemize}
         	\item [{\sf (2.1)}] $row(1)=row(2)$. Then, $row(1)=row(2)=3$, which implies that $row(4)=3$. But
         	then $|R|\le 2$, a contradiction. 
         	\item [{\sf (2.2)}] $row(2)=row(3)$. Then, $row(2)=row(3)=1$, which implies that $row(4)=3$ and $row(1)=2$. 
         	Hence, $\gamma_{12}=\gamma_{21}=(a,b)$, and thus, $D_1$ is of type 1.
         	\item [{\sf (2.3)}] $row(3)=row(4)=1$. Then $\gamma_{23}=(b,b)$ and this determines the matrix uniquely:\smallskip\\
         	$\left( \begin{array}{llllll}
         	(b,a)   &  (b,b)  & (a,b)   &   (a,b)    \\
         	(a,b)   &  {\color{green}\fbox{(b,a)}}   & (b,b)   & {\color{green}\fbox{(b,b)}}    \\
         	(b,b)   & (a,b)    & (b,a)   & (b,b)  \\
         	(b,b)   &  {\color{green}\fbox{(b,b)}}     & (b,b) &  {\color{green}\fbox{(b,a)}}    \\
         	\end{array}
         	\right)\,$ and $\{1,3\}\times\{2,4\}$ is a winning pair.
         	\item [{\sf (2.4)}] $row(3)=row(4)=2$. Then  $\gamma_{23}=(a,b)$ and this determines the matrix uniquely:\smallskip\\
         	$\left( \begin{array}{llllll}
         	{\color{green}\fbox{(b,a)}}   &  (a,b)  & (b,b)   &   {\color{green}\fbox{(b,b)}}    \\
         	(b,b)   &  (b,a)  & (a,b)   &  (a,b)   \\
         	(a,b)   & (b,b)    & (b,a)   & (b,b)  \\
         	 {\color{green}\fbox{(b,b)}}  &  (b,b)   & (b,b) &  {\color{green}\fbox{(b,a)}}    \\
         	\end{array}
         	\right)\,$ and $\{1,4\}\times\{1,4\}$ is a winning pair.
         \end{itemize}
      \end{enumerate}
         So, up to reordering of rows and columns, there are 4 matrices for which no of the $D_i$'s is a winning
         pair. We have shown that these 4 matrices have a winning pair. \qed 
\end{proof}
}

\noindent We are now ready to prove the key result.

\begin{theorem}
	\label{thm:winning-pair}
	{Let $\mathsf{G}$ be a 2-players, 2-values, $n$-strategies normal game with
	$n\ge 4$ and not $n=|C|=|R|=4$. Then a winning pair exists and can be found in time $O(n)$.}
\end{theorem}

\begin{proof}
	Consider an algorithm that simulates the algorithms described in Lemmas~\ref{lem:alg-n5} and \ref{lem:alg-winpair}.
	So this algorithm finds a winning pair {for $n\ge 5$ and} $|R|\ge 4$ or $|C|\ge 4$, or  $|\{col(i)~:~ i\not\in R\}| \ge 2$ 
	and $|\{row(j)~:~ j\not\in C\}| \ge 2$. \cg{The case $n=4$ and ($|C|=4$ and $|R|\le 3)$) or $(|R|=4$ and $|C|\le 3))$
	is taken care by the construction described in Lemma~\ref{lem:nis4}.} Because of the symmetry of rows and columns we can assume now 
	$|R|\le 3$ and $|C|\le 3$, and $|\{col(i)~:~ i\not\in R\}| = 1$.
	
	In the proof we use some appropriate numbering of rows and columns. We are aware that the
	algorithm would operate with the actual row and column numbers, but we feel that using the renumbering
	allows for an easier understanding of the algorithmic idea. 
	
	By renumbering the columns we can obtain $col(i)=1$ for $i\not\in R$. By renumbering the rows and the columns
	$2,\ldots, n$ we can obtain additionally $R=\{1,2\}$ if $|R|=2$ and $R=\{1,2,3\}$ if $|R|=3$, and
	$row(1)=1$, $row(2)=2$. Note that in this setting $\gamma(i,j)=(b,b)$ for all $i\ge 4,~j\ge 2$. Figure~\ref{fig:initial} illustrates this initial setting.
	
	\begin{figure}[h]
		%
		%
		\setlength{\unitlength}{2000sp}%
		\begingroup\makeatletter\ifx\SetFigFont\undefined%
		\gdef\SetFigFont#1#2#3#4#5{%
			\reset@font\fontsize{#1}{#2pt}%
			\fontfamily{#3}\fontseries{#4}\fontshape{#5}%
			\selectfont}%
		\fi\endgroup%
		\begin{picture}(5992,5340)(489,-5175)
		{\color[rgb]{0,0,0}\put(1426,-136){\line( 0,-1){4755}}
		}%
		{\color[rgb]{0,0,0}\put(2378,-129){\line( 0,-1){4755}}
		}%
		{\color[rgb]{0,0,0}\put(3323,-129){\line( 0,-1){4755}}
		}%
		{\color[rgb]{0,0,0}\put(4268,-129){\line( 0,-1){4755}}
		}%
		{\color[rgb]{0,0,0}\put(5213,-114){\line( 0,-1){4755}}
		}%
		{\color[rgb]{0,0,0}\put(946,-586){\line( 1, 0){4740}}
		}%
		{\color[rgb]{0,0,0}\put(931,-1531){\line( 1, 0){4740}}
		}%
		{\color[rgb]{0,0,0}\put(931,-2491){\line( 1, 0){4740}}
		}%
		{\color[rgb]{0,0,0}\put(931,-3421){\line( 1, 0){4740}}
		}%
		{\color[rgb]{0,0,0}\put(946,-4366){\line( 1, 0){4740}}
		}%
		
		{\color[rgb]{1,0,0}\put(6196,-2941){\line(-1, 0){4020}}
			\put(2176,-2956){\line( 0,-1){2175}}
		}%
		\thicklines
		{\color[rgb]{1,0,0}\put(3826,-3856){\line( 5,-2){2125.862}}
		}%
		\put(1350,-25){\makebox(0,0)[lb]{\smash{{\SetFigFont{12}{14.4}{\rmdefault}{\mddefault}{\updefault}{\color[rgb]{0,1,0}$1$}%
		}}}}
		\put(1050,-466){\makebox(0,0)[lb]{\smash{{\SetFigFont{12}{14.4}{\rmdefault}{\mddefault}{\updefault}{\color[rgb]{1,0,0}$(a,b)$}%
		}}}}
		\put(2000,-1449){\makebox(0,0)[lb]{\smash{{\SetFigFont{12}{14.4}{\rmdefault}{\mddefault}{\updefault}{\color[rgb]{1,0,0}$(a,b)$}%
		}}}}
		\put(6046,-4891){\makebox(0,0)[lb]{\smash{{\SetFigFont{12}{14.4}{\rmdefault}{\mddefault}{\updefault}{\color[rgb]{1,0,0}$\bm{(b,b)}$}%
		}}}}
		\put(1070,-4270){\makebox(0,0)[lb]{\smash{{\SetFigFont{12}{14.4}{\rmdefault}{\mddefault}{\updefault}{\color[rgb]{1,0,0}$(b,a)$}%
		}}}}
		\put(1070,-3346){\makebox(0,0)[lb]{\smash{{\SetFigFont{12}{14.4}{\rmdefault}{\mddefault}{\updefault}{\color[rgb]{1,0,0}$(b,a)$}%
		}}}}
		\put(2300,-25){\makebox(0,0)[lb]{\smash{{\SetFigFont{12}{14.4}{\rmdefault}{\mddefault}{\updefault}{\color[rgb]{0,1,0}$2$}%
		}}}}
		\put(3250,-25){\makebox(0,0)[lb]{\smash{{\SetFigFont{12}{14.4}{\rmdefault}{\mddefault}{\updefault}{\color[rgb]{0,1,0}$3$}%
		}}}}
		\put(4180,-25){\makebox(0,0)[lb]{\smash{{\SetFigFont{12}{14.4}{\rmdefault}{\mddefault}{\updefault}{\color[rgb]{0,1,0}$4$}%
		}}}}
		\put(5150, -25){\makebox(0,0)[lb]{\smash{{\SetFigFont{12}{14.4}{\rmdefault}{\mddefault}{\updefault}{\color[rgb]{0,1,0}$5$}%
		}}}}
		\put(650,-700){\makebox(0,0)[lb]{\smash{{\SetFigFont{12}{14.4}{\rmdefault}{\mddefault}{\updefault}{\color[rgb]{0,1,0}$1$}%
		}}}}
		\put(650,-1620){\makebox(0,0)[lb]{\smash{{\SetFigFont{12}{14.4}{\rmdefault}{\mddefault}{\updefault}{\color[rgb]{0,1,0}$2$}%
		}}}}
		\put(650,-2600){\makebox(0,0)[lb]{\smash{{\SetFigFont{12}{14.4}{\rmdefault}{\mddefault}{\updefault}{\color[rgb]{0,1,0}$3$}%
		}}}}
		\put(650,-3500){\makebox(0,0)[lb]{\smash{{\SetFigFont{12}{14.4}{\rmdefault}{\mddefault}{\updefault}{\color[rgb]{0,1,0}$4$}%
		}}}}
		\put(650,-4500){\makebox(0,0)[lb]{\smash{{\SetFigFont{12}{14.4}{\rmdefault}{\mddefault}{\updefault}{\color[rgb]{0,1,0}$5$}%
		}}}}
		\put(6466,-601){\makebox(0,0)[lb]{\smash{{\SetFigFont{12}{14.4}{\rmdefault}{\mddefault}{\updefault}{\color[rgb]{0,1,0}$R=\{1,2\}$}%
		}}}}
		\put(6466,-950){\makebox(0,0)[lb]{\smash{{\SetFigFont{12}{14.4}{\rmdefault}{\mddefault}{\updefault}{\color[rgb]{0,1,0}if $|R|=2$}%
		}}}}
		\put(6421,-1486){\makebox(0,0)[lb]{\smash{{\SetFigFont{12}{14.4}{\rmdefault}{\mddefault}{\updefault}{\color[rgb]{0,1,0}$R=\{1,2,3\}$}%
		}}}}
		\put(6421,-1850){\makebox(0,0)[lb]{\smash{{\SetFigFont{12}{14.4}{\rmdefault}{\mddefault}{\updefault}{\color[rgb]{0,1,0}if $|R|=3$}%
		}}}}
		\end{picture}%
		\caption{Illustrating the initial setting of the proof of Theorem~\ref{thm:winning-pair}.}
		\label{fig:initial}	
	\end{figure}
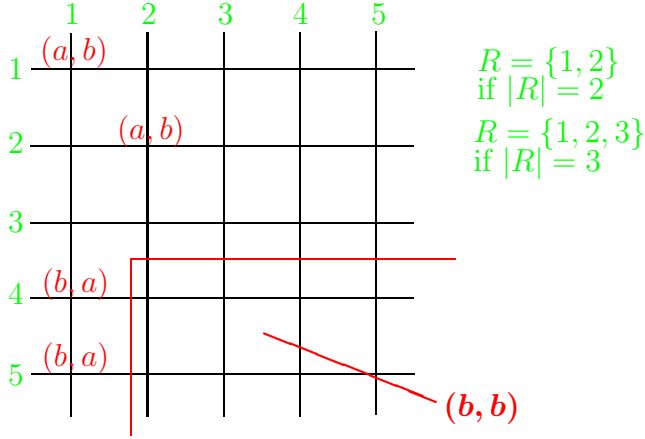	
	
	Since we do not need the information whether $|R|=2$ or $|R|=3$ in this first part of the proof,
	we make no assumptions about the entries in the third row at this stage of the proof. 
	
	Next we discuss the influence of the values $col(1)$ and $col(2)$. If $col(1)=2$ and $col2)=1$ or
	if $col(1)\ne 2$ and $col(2)\ne 1$ and $col(1) \ne col(2)$, then $\{1,2\} \times \{1,2\}$ is a 
	winning pair. In both cases $row(1)\ne row(2)$ and $col(1)\ne col(2)$. So, the following three cases
	remain: 
	\begin{enumerate}[$(i)$]
		\item $col(1)=2,~col(2)\ne 1$
		\item $col(2)=1,~col(1)\ne 2$
		\item $col(1) = col(2)$
	\end{enumerate}	
	
	\noindent Case $(ii)$ is the more elaborate case. We start by discussing cases $(i)$ and $(iii)$.\vspace{.2em}
	
	\noindent{\bf \underline{Cases \bm{$(i)/(iii)$}}}\vspace{.4em}
	
	\noindent Here, in both cases, $\gamma(2,1)=(b,b)$ and $col(2)=j \ge 3$, $\gamma(2,j)=(b,a)$. 
	Then $\{2,4\}\times\{1,j\}$ is a winning pair. 
	This is true, since $\gamma(4,1)=(b,a)$, $\gamma(4,j)=(b,b)$, $col(4)=1\ne j = col(2)$ and
	$row(1)=2 \ne row(j)$, since $\gamma(2,j)=(b,a)$. This situation is illustrated in Figure~\ref{fig:casesi-iii}.
	
	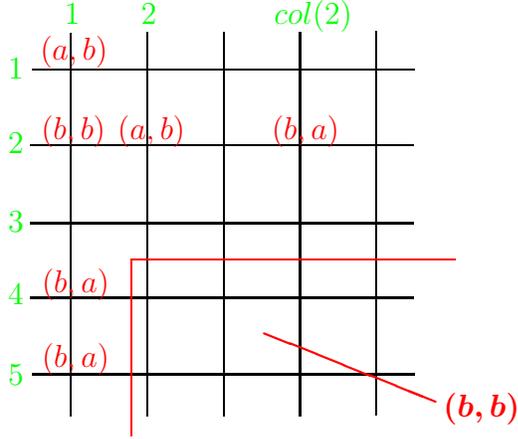
\begin{figure}[h]
		%
		%
		\setlength{\unitlength}{2000sp}%
		\begingroup\makeatletter\ifx\SetFigFont\undefined%
		\gdef\SetFigFont#1#2#3#4#5{%
			\reset@font\fontsize{#1}{#2pt}%
			\fontfamily{#3}\fontseries{#4}\fontshape{#5}%
			\selectfont}%
		\fi\endgroup%
		\begin{picture}(5992,5340)(489,-5175)
		{\color[rgb]{0,0,0}\put(1426,-136){\line( 0,-1){4755}}
		}%
		{\color[rgb]{0,0,0}\put(2378,-129){\line( 0,-1){4755}}
		}%
		{\color[rgb]{0,0,0}\put(3323,-129){\line( 0,-1){4755}}
		}%
		{\color[rgb]{0,0,0}\put(4268,-129){\line( 0,-1){4755}}
		}%
		{\color[rgb]{0,0,0}\put(5213,-114){\line( 0,-1){4755}}
		}%
		{\color[rgb]{0,0,0}\put(946,-586){\line( 1, 0){4740}}
		}%
		{\color[rgb]{0,0,0}\put(931,-1531){\line( 1, 0){4740}}
		}%
		{\color[rgb]{0,0,0}\put(931,-2491){\line( 1, 0){4740}}
		}%
		{\color[rgb]{0,0,0}\put(931,-3421){\line( 1, 0){4740}}
		}%
		{\color[rgb]{0,0,0}\put(946,-4366){\line( 1, 0){4740}}
		}%
		
		{\color[rgb]{1,0,0}\put(6196,-2941){\line(-1, 0){4020}}
			\put(2176,-2956){\line( 0,-1){2175}}
		}%
		\thicklines
		{\color[rgb]{1,0,0}\put(3826,-3856){\line( 5,-2){2125.862}}
		}%
		\put(1350,-25){\makebox(0,0)[lb]{\smash{{\SetFigFont{12}{14.4}{\rmdefault}{\mddefault}{\updefault}{\color[rgb]{0,1,0}$1$}%
		}}}}
		\put(1050,-466){\makebox(0,0)[lb]{\smash{{\SetFigFont{12}{14.4}{\rmdefault}{\mddefault}{\updefault}{\color[rgb]{1,0,0}$(a,b)$}%
		}}}}
		\put(1060,-1449){\makebox(0,0)[lb]{\smash{{\SetFigFont{12}{14.4}{\rmdefault}{\mddefault}{\updefault}{\color[rgb]{1,0,0}$(b,b)$}%
		}}}}
		\put(3920,-1449){\makebox(0,0)[lb]{\smash{{\SetFigFont{12}{14.4}{\rmdefault}{\mddefault}{\updefault}{\color[rgb]{1,0,0}$(b,a)$}%
		}}}}
		\put(2010,-1449){\makebox(0,0)[lb]{\smash{{\SetFigFont{12}{14.4}{\rmdefault}{\mddefault}{\updefault}{\color[rgb]{1,0,0}$(a,b)$}%
		}}}}
		\put(6046,-4891){\makebox(0,0)[lb]{\smash{{\SetFigFont{12}{14.4}{\rmdefault}{\mddefault}{\updefault}{\color[rgb]{1,0,0}$\bm{(b,b)}$}%
		}}}}
		\put(1070,-4270){\makebox(0,0)[lb]{\smash{{\SetFigFont{12}{14.4}{\rmdefault}{\mddefault}{\updefault}{\color[rgb]{1,0,0}$(b,a)$}%
		}}}}
		\put(1070,-3346){\makebox(0,0)[lb]{\smash{{\SetFigFont{12}{14.4}{\rmdefault}{\mddefault}{\updefault}{\color[rgb]{1,0,0}$(b,a)$}%
		}}}}
		\put(2300,-25){\makebox(0,0)[lb]{\smash{{\SetFigFont{12}{14.4}{\rmdefault}{\mddefault}{\updefault}{\color[rgb]{0,1,0}$2$}%
		}}}}
		\put(3250,-25){\makebox(0,0)[lb]{\smash{{\SetFigFont{12}{14.4}{\rmdefault}{\mddefault}{\updefault}{\color[rgb]{0,0,0}}%
		}}}}
		\put(3950,-25){\makebox(0,0)[lb]{\smash{{\SetFigFont{12}{14.4}{\rmdefault}{\mddefault}{\updefault}{\color[rgb]{0,1,0}$col(2)$}%
		}}}}
		\put(5150, -25){\makebox(0,0)[lb]{\smash{{\SetFigFont{12}{14.4}{\rmdefault}{\mddefault}{\updefault}{\color[rgb]{0,0,0}}%
		}}}}
		\put(650,-700){\makebox(0,0)[lb]{\smash{{\SetFigFont{12}{14.4}{\rmdefault}{\mddefault}{\updefault}{\color[rgb]{0,1,0}$1$}%
		}}}}
		\put(650,-1620){\makebox(0,0)[lb]{\smash{{\SetFigFont{12}{14.4}{\rmdefault}{\mddefault}{\updefault}{\color[rgb]{0,1,0}$2$}%
		}}}}
		\put(650,-2600){\makebox(0,0)[lb]{\smash{{\SetFigFont{12}{14.4}{\rmdefault}{\mddefault}{\updefault}{\color[rgb]{0,1,0}$3$}%
		}}}}
		\put(650,-3500){\makebox(0,0)[lb]{\smash{{\SetFigFont{12}{14.4}{\rmdefault}{\mddefault}{\updefault}{\color[rgb]{0,1,0}$4$}%
		}}}}
		\put(650,-4500){\makebox(0,0)[lb]{\smash{{\SetFigFont{12}{14.4}{\rmdefault}{\mddefault}{\updefault}{\color[rgb]{0,1,0}$5$}%
		}}}}
		\end{picture}%
		\caption{Illustrating cases $(i)/(iii)$ of the proof of Theorem~\ref{thm:winning-pair}.}
		\label{fig:casesi-iii}
	\end{figure}
	
	\noindent{\bf \underline{Case \bm{$(ii)$}}}\vspace{.4em}
	
	\noindent In this case $\gamma(1,2)=(b,b)$ and $\gamma(2,1)=(b,a)$. By renumbering the columns 
	that are greater or equal to $3$, we can obtain that $col(1)=3$, i.e., $\gamma(1,3)=(b,a)$. 
	If $row(3)=2$, then $\{1,2\}\times \{1,3\}$ is a winning pair. This true since in this case
	$\gamma(1,1)=\gamma(2,3)=(a,b)$ and $\gamma(1,3)=\gamma(2,1)=(b,a)$, and 
	$row(1)=1\ne 2 = row(2)$ and $col(1)=3\ne 1=col(1)$.
	
	So, there remain the case $row(3)\ne 2$. It is $row(3)\ne 1$ since $\gamma(1,3)=(b,a)$. We are now
	in the case $|R|=3$ and $row(3)=3$. This situation is illustrated in Figure~\ref{fig:caseii}.
	
	\begin{figure}[h]
		%
		%
		\setlength{\unitlength}{2000sp}%
		\begingroup\makeatletter\ifx\SetFigFont\undefined%
		\gdef\SetFigFont#1#2#3#4#5{%
			\reset@font\fontsize{#1}{#2pt}%
			\fontfamily{#3}\fontseries{#4}\fontshape{#5}%
			\selectfont}%
		\fi\endgroup%
		\begin{picture}(5992,5340)(489,-5175)
		{\color[rgb]{0,0,0}\put(1426,-136){\line( 0,-1){4755}}
		}%
		{\color[rgb]{0,0,0}\put(2378,-129){\line( 0,-1){4755}}
		}%
		{\color[rgb]{0,0,0}\put(3323,-129){\line( 0,-1){4755}}
		}%
		{\color[rgb]{0,0,0}\put(4268,-129){\line( 0,-1){4755}}
		}%
		{\color[rgb]{0,0,0}\put(5213,-114){\line( 0,-1){4755}}
		}%
		{\color[rgb]{0,0,0}\put(946,-586){\line( 1, 0){4740}}
		}%
		{\color[rgb]{0,0,0}\put(931,-1531){\line( 1, 0){4740}}
		}%
		{\color[rgb]{0,0,0}\put(931,-2491){\line( 1, 0){4740}}
		}%
		{\color[rgb]{0,0,0}\put(931,-3421){\line( 1, 0){4740}}
		}%
		{\color[rgb]{0,0,0}\put(946,-4366){\line( 1, 0){4740}}
		}%
		
		{\color[rgb]{1,0,0}\put(6196,-2941){\line(-1, 0){4020}}
			\put(2176,-2956){\line( 0,-1){2175}}
		}%
		\thicklines
		{\color[rgb]{1,0,0}\put(3826,-3856){\line( 5,-2){2125.862}}
		}%
		\put(1350,-25){\makebox(0,0)[lb]{\smash{{\SetFigFont{12}{14.4}{\rmdefault}{\mddefault}{\updefault}{\color[rgb]{0,1,0}$1$}%
		}}}}
		\put(1050,-466){\makebox(0,0)[lb]{\smash{{\SetFigFont{12}{14.4}{\rmdefault}{\mddefault}{\updefault}{\color[rgb]{1,0,0}$(a,b)$}%
		}}}}
		\put(2040,-466){\makebox(0,0)[lb]{\smash{{\SetFigFont{12}{14.4}{\rmdefault}{\mddefault}{\updefault}{\color[rgb]{1,0,0}$(b,b)$}%
		}}}}
		\put(2970,-466){\makebox(0,0)[lb]{\smash{{\SetFigFont{12}{14.4}{\rmdefault}{\mddefault}{\updefault}{\color[rgb]{1,0,0}$(b,a)$}%
		}}}}
		\put(1065,-1449){\makebox(0,0)[lb]{\smash{{\SetFigFont{12}{14.4}{\rmdefault}{\mddefault}{\updefault}{\color[rgb]{1,0,0}$(b,a)$}%
		}}}}
		\put(2980,-1449){\makebox(0,0)[lb]{\smash{{\SetFigFont{12}{14.4}{\rmdefault}{\mddefault}{\updefault}{\color[rgb]{1,0,0}$(b,b)$}%
		}}}}
		\put(2010,-1449){\makebox(0,0)[lb]{\smash{{\SetFigFont{12}{14.4}{\rmdefault}{\mddefault}{\updefault}{\color[rgb]{1,0,0}$(a,b)$}%
		}}}}
		\put(6046,-4891){\makebox(0,0)[lb]{\smash{{\SetFigFont{12}{14.4}{\rmdefault}{\mddefault}{\updefault}{\color[rgb]{1,0,0}$\bm{(b,b)}$}%
		}}}}
		\put(1070,-4270){\makebox(0,0)[lb]{\smash{{\SetFigFont{12}{14.4}{\rmdefault}{\mddefault}{\updefault}{\color[rgb]{1,0,0}$(b,a)$}%
		}}}}
		\put(1070,-3346){\makebox(0,0)[lb]{\smash{{\SetFigFont{12}{14.4}{\rmdefault}{\mddefault}{\updefault}{\color[rgb]{1,0,0}$(b,a)$}%
		}}}}
		\put(2950,-2400){\makebox(0,0)[lb]{\smash{{\SetFigFont{12}{14.4}{\rmdefault}{\mddefault}{\updefault}{\color[rgb]{1,0,0}$(a,b)$}%
		}}}}
		\put(2300,-25){\makebox(0,0)[lb]{\smash{{\SetFigFont{12}{14.4}{\rmdefault}{\mddefault}{\updefault}{\color[rgb]{0,1,0}$2$}%
		}}}}
		\put(3250,-25){\makebox(0,0)[lb]{\smash{{\SetFigFont{12}{14.4}{\rmdefault}{\mddefault}{\updefault}{\color[rgb]{0,1,0}$3$}%
		}}}}
		\put(4160,-25){\makebox(0,0)[lb]{\smash{{\SetFigFont{12}{14.4}{\rmdefault}{\mddefault}{\updefault}{\color[rgb]{0,1,0}$4$}%
		}}}}
		\put(5150, -25){\makebox(0,0)[lb]{\smash{{\SetFigFont{12}{14.4}{\rmdefault}{\mddefault}{\updefault}{\color[rgb]{0,1,0}$5$}%
		}}}}
		\put(650,-700){\makebox(0,0)[lb]{\smash{{\SetFigFont{12}{14.4}{\rmdefault}{\mddefault}{\updefault}{\color[rgb]{0,1,0}$1$}%
		}}}}
		\put(650,-1620){\makebox(0,0)[lb]{\smash{{\SetFigFont{12}{14.4}{\rmdefault}{\mddefault}{\updefault}{\color[rgb]{0,1,0}$2$}%
		}}}}
		\put(650,-2600){\makebox(0,0)[lb]{\smash{{\SetFigFont{12}{14.4}{\rmdefault}{\mddefault}{\updefault}{\color[rgb]{0,1,0}$3$}%
		}}}}
		\put(650,-3500){\makebox(0,0)[lb]{\smash{{\SetFigFont{12}{14.4}{\rmdefault}{\mddefault}{\updefault}{\color[rgb]{0,1,0}$4$}%
		}}}}
		\put(650,-4500){\makebox(0,0)[lb]{\smash{{\SetFigFont{12}{14.4}{\rmdefault}{\mddefault}{\updefault}{\color[rgb]{0,1,0}$5$}%
		}}}}
		\end{picture}%
		\caption{Illustrating case $(ii)$ of the proof of Theorem~\ref{thm:winning-pair}.}
		\label{fig:caseii}
	\end{figure}
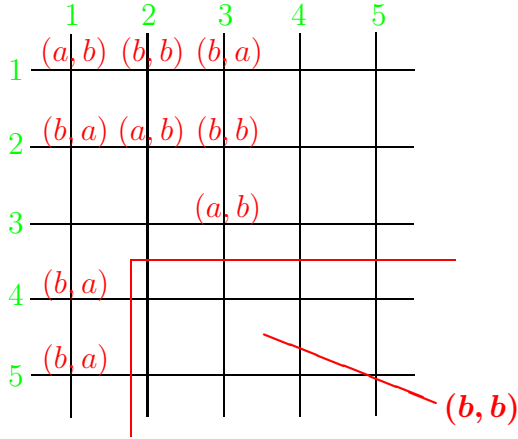

	We now study the influence of $col(3)$. If $col(3)=1$, then $\{1,3\}\times \{1,3\}$ is a winning
	pair, since $\gamma(1,1)=\gamma(3,3)=(a,b)$ and $\gamma(1,3)=\gamma(3,1)=(b,a)$, and
	$row(1)=1\ne 3 = row(3)$ and $col(1)=3\ne 1=col(3)$.
	
	Now let $col(3)=2$, i.e., $\gamma(3,1)=(b,b)$ and $\gamma(3,2)=(b,a)$. 
	Then $\{3,4\}\times \{1,2\}$ is a winning pair, since  
	$\gamma(3,1)=\gamma(4,2)=(b,b)$ and $\gamma(3,3)=\gamma(4,1)=(a,b)$, and
	$col(3)=2\ne 1 = col(4)$ and $row(1)=1\ne 2=row(2)$.
	
	Finally, consider the case $col(3)\ge 3$. Since $\gamma(3,3)=(a,b)$ this implies
	that $col(3)\ge 4$. Then $\gamma(3,2)=(b,b)$ and $\{2,3\}\times \{2,3\}$ is a winning
	pair, since $\gamma(2,2)=\gamma(3,3)=(a,b)$ and $\gamma(2,3)=\gamma(3,2)=(b,b)$, and
	$row(2)2\ne 3 = row(3)$ and $col(2)=1\ne col(3)$, completing proof of the theorem.\qed
\end{proof}

\noindent Theorem~\ref{thm:winning-pair}, Lemma~\ref{lem:psFequil} and the definition of a winning pair
(Definition~\ref{def:winpair}) yield the following result:

\begin{corollary}
	\label{cor:final-result}
    Consider a 2-players, 2-values, $n$-strategies 
    normal game $\mathsf{G}$, 
	and a unimodal valuation ${\mathsf{F}}$
	with
	${{\mathsf{F}}}\left( \frac{1}{2}\right)=b$.
	Then,
	${\mathsf{G}}$ has an ${\mathsf{F}}$-equilibrium
	that can be computed in
	$O(n)$ time.
\end{corollary}

\noindent{\bf Remark:} {\em Observe that Theorem~\ref{thm:winning-pair} does not cover the cases $n\leq 3$ and 
$n=|R|=|C|=4$. However, we know from Theorem~\ref{2 values 3 strategies existence} (next section)  that for $n=3$ there always exists 
an ${{\mathsf{F}}}$-equilibrium with ${{\mathsf{F}}}\left( \frac{\textstyle 1}{\textstyle 2}\right)=b$.
From Corollary~\ref{cor:2strategies} we know that an {\sf F}-equilibrium exists also for $n=2$. Regarding the case
$n=|R|=|C|=4$, 
the strategy vector $p_{1i}=p_{2i}=1/4$ for $i\in [n]$ is a fully mixed equilibrium.
}


\subsection{Three Strategies} 
\label{three strategies}
{Here we focus on two-players,
	two-values games with $\mathsf{F}\left(\frac{1}{2}\right)\le b$ and $3$ strategies (i.e., $n=3$).
	In this case we have a complete picture on the existence and inexistence
	of {\sf F}-equilibria (c.f. Corollary~\ref{3strategies}).}
We first show:

\begin{theorem}
	\label{2 values 3 strategies existence}
	{Consider
		a unimodal valuation ${\mathsf{F}}$
		with
		${\mathsf{F}}\left( \frac{1}
		{2} 
		\right) \leq b$.
		Then,
		every 2-players, 2-values, 3-strategies game
		has an ${\mathsf{F}}$-equilibrium.}
\end{theorem}

\begin{proof}
	{We start
		with an ${\mathsf{E}}$-equilibrium
		$(p_{1}, p_{2})$ for ${\mathsf{G}}$.}
	{If $|\sigma (p_{1})| = 1$
		or
		$|\sigma (p_{2})| = 1$,
		then,
		by Lemma~\ref{lemma A},
		there is also a pure equilibrium for ${\mathsf{G}}$.}
	{So assume that
		$|\sigma (p_{i})| \geq 2$
		for each player $i \in [2]$.}
	{By Lemma~\ref{lemma B hat},
		assume,
		without loss of generality,
		that if for a player $i \in [2]$,
		$|\sigma (p_{i})| = 2$,
		then
		$p_{i}(j) = \frac{\textstyle 1}
		{\textstyle 2}$
		for each strategy $j \in \sigma(p_{i})$.} 
	{If $|\sigma (p_{1})| = |\sigma (p_{2})| = 3$,
		then $(p_{1}, p_{2})$ is a fully mixed ${\mathsf{E}}$-equilibrium;
		hence,
		by Corollary~\ref{completely trivial},
		it is also an ${\mathsf{F}}$-equilibrium.}
	{So assume, 
		without loss of generality, that
		$|\sigma (p_{2})| = 2$,
		with $\sigma (p_{2}) = \{ 1, 2 \}$.
		We distinguish two cases
		with respect to $|\sigma (p_{1})|$:}
	
	\begin{enumerate}

		\item[{\sf (A)}]
		\underline{{$|\sigma (p_{1})| = 2$, with $\sigma (p_{1}) = \{ 1, 2 \}$:}} 
		{The idea of the proof is to show that $(p_{1}, p_{2})$ is also an {\sf F}-equilibrium.}
		{Since player $1$ is $\mathsf{E}$-constant
			on $\sigma (p_{1})$,
			Lemma~\ref{x-constant} implies that
			she is also {${\mathsf{F}}$-constant} on $\sigma (p_{1})$.}                                                   
		{So it remains to prove
			that player $1$ cannot ${\mathsf{F}}$-improve
			by switching to strategy $3$.}
		{By assumption,
			$p_{2}(1) = 
			p_{2}(2) = 
			\frac{\textstyle 1}
			{\textstyle 2}$.}                                                                 
		{We distinguish again the three cases from
			the proof of Lemma~\ref{lemma B hat}:}
		\begin{enumerate}

			\item[{\sf (A/1)}]
			\underline{{${\mathsf{\mu}}_{1}(k, 1)
					=
					{\mathsf{\mu}}_{1}(k, 2)
					=
					a$
					for all $k \in \sigma (p_{1})$:}} 
			{Then,
				$x_{1}(p_{1}, p_{2}) = 1$,
				so that 
				${\mathsf{V}}_{1}(p_{1}, p_{2})
				=
				{\mathsf{F}}(1)
				=
				a$.
				Thus,
				player $1$ cannot ${\mathsf{F}}$-improve.}

			\item[{\sf (A/2)}]
			\underline{{${\mathsf{\mu}}_{1}(k, 1)
					=
					{\mathsf{\mu}}_{1}(k, 2)
					=
					b$
					for all $k \in \sigma (p_{1})$:}}  
			{Then,
				$x_{1}(p_{1}, p_{2})
				=
				0$,
				so that
				${\mathsf{V}}_{1}(p_{1}, p_{2})
				=
				{\mathsf{F}}(0)
				=
				b$.}  
			{Since $(p_{1}, p_{2})$ is an ${\mathsf{E}}$-equilibrium,  
				it must also hold that
				${\mathsf{\mu}}_{1}(3,1)
				=
				{\mathsf{\mu}}_{1}(3, 2)
				=
				b$.
				Thus,
				$x_{1}(p_{1}^{3}, p_{2})
				=
				0$,
				so that
				${\mathsf{V}}_{1}(p_{1}^{3}, p_{2})
				=
				{\mathsf{F}}(0)
				=
				b$,
				and player $1$ cannot ${\mathsf{F}}$-improve.}

			\item[{\sf (A/3)}]
			\underline{{${\mathsf{\mu}}_{1}(k, 1)
					\neq
					{\mathsf{\mu}}_{1}(k, 2)$
					for all $k \in \sigma (p_{1})$:}}                    
			{Then,}
			{$x_{1}(p_{1}, p_{2}) = \frac{\textstyle 1}
				{\textstyle 2}$,
				so that                                                     
				${\mathsf{V}}_{1}(p_{1}, p_{2})
				=
				{\mathsf{F}}\left( \frac{\textstyle 1}
				{\textstyle 2}
				\right)$.}
			Since player $1$ is {${\mathsf{E}}$-happy},
				there are only two cases 
				(the case $\mu_1(3,1)=\mu_1(3,2)=a$ is excluded since in this
			    case player 1 is not $\mathsf{E}$-happy):
			\begin{enumerate}

				\item[{\sf (A/3/i)}]
				\underline{{${\mathsf{\mu}}_{1}(3, 1)
						\neq
						{\mathsf{\mu}}_{1}(3, 2)$:}}
				{Then,
					$x_{1}(p_{1}^{3}, p_{2})
					=
					\frac{\textstyle 1}
					{\textstyle 2}$,
					so that          
					${\mathsf{V}}_{1}(p_{1}^{3}, p_{2})
					=
					{\mathsf{F}}\left( \frac{\textstyle 1}
					{\textstyle 2}
					\right)$.
					So, player $1$ cannot ${\mathsf{F}}$-improve.}

				\item[{\sf (A/3/ii)}]
				\underline{{${\mathsf{\mu}}_{1}(3,1)
						=
						{\mathsf{\mu}}_{1}(3,2)
						=
						b$:}}                                                                             
				{Then,
					$x_{1}(p_{1}^{3}, p_{2}) = 0$,
					so that
					${\mathsf{V}}_{1}(p_{1}^{3}, p_{2})
					=
					{\mathsf{F}}(0)
					=
					b$.
					Since 
					${\mathsf{F}}\left( \frac{\textstyle 1}
					{\textstyle 2}
					\right)
					\leq
					b$,
					player $1$ cannot ${\mathsf{F}}$-improve.}

			\end{enumerate}

		\end{enumerate}
		\noindent So player $1$
			cannot ${\mathsf{F}}$-improve by switching to
			strategy $3$.
		{Due to the symmetry
			between the two players,}
		{player $2$ cannot improve either
			by switching to strategy $3$.}

		\item[{\sf (B)}]                                                                                                                     
		\underline{{$|\sigma (p_{1})| = 3$:}}  
		Then, 
		Corollary~\ref{completely trivial} implies that 
		player $1$ is ${\mathsf{F}}$-happy.
		{The idea of the proof is to show that either $(p_{1}, p_{2})$ is also an {\sf F}-equilibrium,
			or define a new probability distribution for player~$1$ so that the resulting mixed profile is an {\sf F}-equilibrium. The reason we can do this is the following: Because of the Optimal-Value Property (c.f. Definition~\ref{optimal value definition}) player $1$ is {\sf V}-constant on $\sigma(p_1)$; this
			together with Lemma~\ref{x-constant} imply that 
			$\mathsf{V}(p_1,p_2)=\mathsf{F}(x_1(p_1,p_2)) = \mathsf{F}(x_1(p^{s_{i}}_1,p_2))$ for all $s_i\in \sigma(p_1)$.
			Now define a new mixed strategy $\widehat{p}_1$ with 
			$\sigma(\widehat{p}_1)  \subseteq \sigma(p_1)$. Then observe that player $1$ remains {\sf F}-happy.} 
		
		{So, what it remains is to consider the
			${\mathsf{F}}$-happiness of player $2$.
			By the {\sf WEEP} for player $2$,
			${\mathsf{E}}_{2}(p_{1}, p_{2}^{1})
			=
			{\mathsf{E}}_{2}(p_{1}, p_{2}^{2})$.}
		We use the following observation: 
		\begin{equation}
		\text{Since } {\mathsf{F}}\left( \frac{\textstyle 1}
		{\textstyle 2} 
		\right) \leq b,~\mathsf{F}(y)\ge\mathsf{F}(z),~\text{for all}~y<z,~z\ge 1/2.
		\label{Fgreater}
		\end{equation}
		
		{Now, assume,
			without loss of generality,
			that
			${\mathsf{\mu}}_{2}(1,1)
			\leq
			{\mathsf{\mu}}_{2}(1,2)
			\leq
			{\mathsf{\mu}}_{2}(1,3)$.}
		{We distinguish two cases:}                                                                                                                     
		\begin{enumerate}

			\item[{\sf (B/1)}]
			\underline{{${\mathsf{\mu}}_{2}(k, 1)
					=
					{\mathsf{\mu}}_{2}(k, 2)$
					for all strategies $k \in [3]$:}}                                                                                                                                                                                                                                                                                                                                                                                                                                                                                                                                                                                                                                                                                                                                 
			{There are three subcases:}                                                                                                                                                                                                                                                                                                                                                                                                                                                                                                                                                                                                                                                                                                                                                                                                                                                                                                                                                                                                                                                                                                                                                                                                                                                                                                                                                                                                                                                            
			\begin{enumerate}

				\item[{\sf (B/1/i)}]
				\underline{{${\mathsf{\mu}}_{2}(k, 1) = a$
						for all strategies
						$k \in [3]$:}}                                                                                                                                                                                                                                                                                                                                                                                                                                                                                                                                                                                                                                                                                                                                                                                                                                                                                                                                                                                                                                                                                                                                                                                                                                                                                                                                                                                                                                                                                                                                                                                                                                                                                                                                                                                                                                                                                                                                                                                                                                                                                                                                                                                                                                                                                                                                                                                                                                                                                                                                                                                                                                                                                                                                                                                                                                                                                                                                                                                                                                                                                                                                                                                                                                                                                                                                                                                                                                                                                                                                                                                                                                                                                                                                                                                                                                                                                                                                                                                                                                                                                                                                                                                                                                                  
				{Then,
					$x_{2}(p_{1}, p_{2}) = 1$,
					so that 
					${\mathsf{V}}_{2}(p_{1}, p_{2})
					=
					{\mathsf{F}}(1)
					=
					a$.
					So,
					player $2$ cannot ${\mathsf{F}}$-improve.}

				\item[{\sf (B/1/ii)}]                                                                                                                                                                                                                                                                                                                                                                                                                                                                                                                                                                                                                                                                                                                                                                                                                                                                                                                                                                                                                                                                                                                                                                                                                                                                                                                                                                                                                                                                                                                                                                                                                                                                                                                                                                                                                                                                                                                                                                                                                                                                                                                                                                                                                                                                                                                                                                                                                                                                                                                                                                                                                                                                                                                                                                                                                                                                                                                                                                                                                                                                                                                                                                                                                                                                                                                                                                                                                                                                                                                                                                                                                                                                                                                                                                                                                                                                                                                                                                                                                                                                                                                                                                                                                                           \underline{{${\mathsf{\mu}}_{2}(k, 1) = b$
						for all strategies
						$k \in [3]$:}}                                                                                                                                                                                                                                                                                                                                                                                                                                                                                                                                                                                                                                                                                                                                                                                                                                                                                                                                                                                                                                                                                                                                                                                                                                                                                                                                                                                                                                                                                                                                                                                                                                                                                                                                                                                                                                                                                                                                                                                                                                                                                                                                                                                                                                                                                                                                                                                                                                                                                                                                                                                                                                                                                                                                                                                                                                                                                                                                                                                                                                                                                                                                                                                                                                                                                                                                                                                                                                                                                                                                                                                                                                                                                                                                                                                                                                                                                                                                                                                                                                                                                                                                                                                                                                                  
				\remove{removing old proof
					\textcolor{magenta}{Then,
						since $\sigma (p_{2}) = \{ 1, 2 \}$,
						${\mathsf{E}}_{2}(p_{1}, p_{2}) = b$.}
					\textcolor{magenta}{The assumption implies
						that ${\mathsf{\mu}}_{2}(k, 3)
						=
						a$
						for all strategies $k \in [3]$.}  
					\textcolor{magenta}{So,
						${\mathsf{E}}_{2}(p_{1}, p_{2}^{3})
						=
						a$,
						and player $2$ is not \cgr{${\mathsf{E}}$-happy}
						in $(p_{1}, p_{2})$.
						A contradiction.}                   
				}
				Then,
				since $\sigma (p_{2}) = \{ 1, 2 \}$,
				${\mathsf{E}}_{2}(p_{1}, p_{2}) = b$.
				By assumption, $(p_1,p_2)$ is an $\mathsf{E}$-equilibrium. 
				This implies that $\mu_2(k,3)=b$ for all strategies $k\in [3]$
				and $\mathsf{V}_2(p_1,p_2)=\mathsf{V}_2(p_1,p^3_2) = \mathsf{F}(0)=b$. 
				So, player 2 cannot $\mathsf{F}$-improve.

				\item[{\sf (B/1/iii)}]                                                                                                                                                                                                                                                                                                                                                                                                                                                                                                                                                                                                                                                                                                                                                                                                                                                                                                                                                                                                                                                                                                                                                                                                                                                                                                                                                                                                                                                                                                                                                                                                                                                                                                                                                                                                                                                                                                                                                                                                                                                                                                                                                                                                                                                                                                                                                                                                                                                                                                                                                                                                                                                                                                                                                                                                                                                                                                                                                                                                                                                                                                                                                                                                                                                                                                                                                                                                                                                                                                                                                                                                                                                                                                                                                                                                                                                                                                                                                                                                                                                                                                                                                                                                                                         \underline{{${\mathsf{\mu}}_{2}(1, 1) = a$
						and ${\mathsf{\mu}}_{2}(3, 1) = b$:}}                                                                                                                                                                                                                                                                                                                                                                                                                                                                                                                                                                                                                                                                                                                                                                                                                                                                                                                                                                                                                                                                                                                                                                                                                                                                                                                                                                                                                                                                                                                                                                                                                                                                                                                                                                                                                                                                                                                                                                                                                                                                                                                                                                                                                                                                                                                                                                                                                                                                                                                                                                                                                                                                                                                                                                                                                                                                                                                                                                                                                                                                                                                                                                                                                                                                                                                                                                                                                                                                                                                                                                                                                                                                                                                                                                                                                                                                                                                                                                                                                                                                                                                                         
				{Since player $2$
					is {${\mathsf{E}}$-happy} with $(p_{1}, p_{2})$,
					${\mathsf{E}}_{2}(p_{1}, p_{2}^{1})
					\leq
					{\mathsf{E}}_{2}(p_{1}, p_{2}^{3})$.}                                                                                                                                                                                                                                                                                                                                                                                                                                                                                                                                                                                                                                                                                                                                                                                                                                                                                                                                                                                                                                                                                                                                                                                                                                                                                                                                                                                                                                                                                                                                                                                                                                                                                                                                                                                                                                                                                                                                                                                                                                                                                                                                                                                                                                                                                                                                                                                                                                                                                                                                                                                                                                                                                                                                                                                                                                                                                                                                                                                                                                                                                                                                                                                                                                                                                                                                                                                                                                                                                                                                                                                                                                                                                                                                                                                                                                                                                                                                                                                                                                                                                                                                                                                                                  
				There are two subcases to consider:
				\begin{itemize}

					\item
					\underline{{${\mathsf{\mu}}_{2}(2,1) = a$:}}                                                                                                                                                                                                                                                                                                                                                                                                                                                                                                                                                                                                                                                                                                                                                                                                                                                                                                                                                                                                                                                                                                                                                                                                                                                                                                                                                                                                                                                                                                                                                                                                                                                                                                                                                                                                                                                                                                                                                                                                                                                                                                                                                                                                                                                                                                                                                                                                                                                                                                                                                                                                                                                                                                                                                                                                                                                                                                                                                                                                                                                                                                                                                                                                                                                                                                                                                                                                                                                                                                                                                                                                                                                                                                                                                                                                                                                                                                                                                                                                                                                                                                                                                                                               
					Then, 
					${\mathsf{\mu}}_{2}
					=
					\left( \begin{array}{lll}
					a & a & \ast \\
					a & a & \ast \\
					b & b & \ast
					\end{array}
					\right)$, where an $\ast$ is an arbitrary value from $\{ a, b \}$. \break          
					\remove{
						\textcolor{BrickRed}{Hence,                                                                                                                                            \begin{eqnarray*}
								a \cdot (p_{1}(1) + p_{1}(2))
								+
								b \cdot p_{1}(3)
								&  \leq  &
								b \cdot (p_{1}(1) + p_{1}(2))
								+
								a \cdot p_{1}(3)\, ,
							\end{eqnarray*}
							which implies 
							$p_{1}(3)
							\leq
							p_{1}(1) + p_{1}(2)$.
							So
							$x_{2}(p_{1}, p_{2}^{3})
							\leq
							x_{2}(p_{1}, p_{2})$,
							with
							$x_{2}(p_{1}, p_{2})
							\geq
							\frac{\textstyle 1}
							{\textstyle 2}$.} 
					}
					By assumption, $(p_1,p_2)$ is an $\mathsf{E}$-Equilibrium.
					This implies that there is some $i\in[3]$ with $\mu(i,3)=b$. Define a new
					probability distribution $\widehat{p}_1$ for player $1$ by 
					$\widehat{p}_1(k)=1/3$, for all $k\in[3].$
					Then, $V_2(\widehat{p}_1,p_2)=\mathsf{F}\left(\frac{\textstyle 2}
					{\textstyle 3}\right)$ and  $V_2(\widehat{p}_1,p^{3}_2)=\mathsf{F}(x)$ with some 
					$x\le 2/3$.
					{By observation~(\ref{Fgreater}), $\mathsf{F}(x) \ge \mathsf{F}\left(\frac{\textstyle 2}
						{\textstyle 3}\right)$}.
					Hence, $(\widehat{p}_1,p_2)$ is an $\mathsf{F}$-equilibrium, since player 2 cannot $\mathsf{F}$-improve
					by switching to strategy $3$.

					\item
					\underline{{${\mathsf{\mu}}_{2}(2,1) = b$:}}
					{Then,
						${\mathsf{\mu}}_{2}
						=
						\left( \begin{array}{lll}
						a & a & \ast \\
						b & b & \ast \\
						b & b & \ast
						\end{array}
						\right)$.} 
					\remove{
						\textcolor{BrickRed}{Hence,
							\begin{eqnarray*}
								a \cdot p_{1}(1) 
								+
								b \cdot (p_{1}(2) + p_{1}(3))
								&  \leq  &
								b \cdot p_{1}(1) 
								+
								a \cdot (p_{1}(2) + p_{1}(3))\, ,
							\end{eqnarray*}
							which implies 
							$p_{1}(2) + p_{1}(3)
							\leq
							p_{1}(1)$.
							So
							$x_{2}(p_{1}, p_{2}^{3})
							\leq
							x_{2}(p_{1}, p_{2})$,
							with
							$x_{2}(p_{1}, p_{2})
							\geq
							\frac{\textstyle 1}
							{\textstyle 2}$.} 
					}
					If $\mu_2(1,3)=a$, then $\mu_2(2,3)=\mu_2(3,3)=b$,
					since otherwise player $2$ could $\mathsf{E}$-improve
					by choosing strategy $3$, contradicting the assumption
					that $(p_1,p_2)$ is an $\mathsf{E}$-Equilibrium. 
					
					So, consider now the case that $\mu_2(1,3)=b$. Define a 
					new  probability distribution $\widehat{p}_1$ for player $1$ by 
					$\widehat{p}_1(1)=1/2$, and $\widehat{p}_1(2)=\widehat{p}_1(3)=1/4$.
					Then, $V_2(\widehat{p}_1,p_2)=\mathsf{F}\left(\frac{\textstyle 1}
					{\textstyle 2}\right)$ and  $V_2(\widehat{p}_1,p^{3}_2)=\mathsf{F}(x)$ with some
					$x\le 1/2$. {By observation~(\ref{Fgreater}), $\mathsf{F}(x)\ge \mathsf{F}\left(\frac{\textstyle 1}
						{\textstyle 2}\right)$}.
					Hence, $(\widehat{p}_1,p_2)$ is an $\mathsf{F}$-equilibrium, since player 2 cannot $\mathsf{F}$-improve
					by switching to strategy~$3$.
					
				\end{itemize}                                                                                                                                                           \remove{
					\leq
					x_{2}(p_{1}, p_{2})$,
					with
					$x_{2}(p_{1}, p_{2})
					\geq
					\frac{\textstyle 1}
					{\textstyle 2}$. 
					Since ${\mathsf{F}}$ is a unimodal valuation 
					with ${\mathsf{F}}\left( \frac{1}
					{2}
					\right)
					\leq
					b$,
					it follows that
					${\mathsf{F}}(p_{1}, p_{2})
					\leq
					{\mathsf{F}}(p_{1}, p_{2}^{3})$,
					and player $2$ cannot ${\mathsf{F}}$-improve.
				}
				\end{enumerate}

				\item[{\sf (B/2)}]
				\underline{{There is a strategy $\widehat{k} \in [3]$
						with
						${\mathsf{\mu}}_{2}(\widehat{k}, 1)
						\neq
						{\mathsf{\mu}}_{2}(\widehat{k}, 2)$:}}                                                                                                                                                                                                                                                                                                                                                                                                                                                                                                                                                                                                                                                                                                                                                                                                                                                                                                                            
				{By Lemmas~\ref{domination and weep} and \ref{unimodal implies weep},
					it follows that
					no strategy in $\sigma (p_{2})$
					dominates
					some other strategy in $\sigma (p_{2})$
					with respect to $\sigma (p_{1})$.}
				{Hence,
					there is at least one other $\widetilde{k} \in [3]$
					with
					${\mathsf{\mu}}_{2}(\widetilde{k}, 1)
					\neq
					{\mathsf{\mu}}_{2}(\widetilde{k}, 2)$.}
				{By the {\sf WEEP}
					for player $2$,
					${\mathsf{E}}_{2}(p_{1}, p_{2}^{1})
					=
					{\mathsf{E}}_{2}(p_{1}, p_{2}^{2})$.}
				{We distinguish four cases,
					each represented by the matrix
					$\left( {\mathsf{\mu}}_{2}(k, j)
					\right)_{1 \leq k, j \leq 3}$,
					where, as before, 
					$\ast$ is an arbitrary value from $\{ a, b \}$:}
				\begin{enumerate}

				\item[{\sf (B/2/i)}]
				{${\mathsf{\mu}}_{2}
					=
					\left( \begin{array}{lll}
						a & a & \ast     \\
						a & b & \ast \\                                                                                                                                                                                                                                                                                                                                                                                                                                                                                                                                                                                                                                                                                                                                                                                                                                                                                                                                                                                                                                                                                                                                                                                                                                                                                                                                                                                                                                                                                                                                                                                                                                                                                                                                                                                                                                                                                                                                                                                                                                                                                                                                                                                                                                                                                                                                                                                                                                                                                                                                                                                                                                                                                                                                                                                                                                                                                                                                                                                                                                                                                                                                                                                                                                                                                                                                                                                                            
						b & a & \ast
					\end{array}
					\right)$:} 
				\remove{ 
					\textcolor{magenta}{Then,
						\begin{eqnarray*}
						a \cdot (p_{1}(1) + p_{1}(2)) + b \cdot p_{1}(3) 
						& = & a \cdot (p_{1}(1) + p_{1}(3)) + b \cdot p_{1}(2)\, ,              
						\end{eqnarray*}
						which implies that
						$p_{1}(2) = p_{1}(3)$.}                                                                                                                                                                                                                                                                                                                                                                                                                                                                                                                                                                                                                                                                                                                                                                                                                                                                                                                                                                                                                                                                                                                                                                                                                                                                                                                                                                                                                                                                                                                                                                                                                                                                                                                                                                                                                                                                                                                                                                                                                                                                                                                                                                                                                                                                                                                                                                                                                                                                                                                                                                                                                                                                                                                                                                                                                                                                                                                                                                                                                                                                                                                                                                                                                                                                                                                                                                                                          
					\textcolor{magenta}{It follows that
						\begin{eqnarray*}
						x_{2}(p_{1}, p_{2})
						&  = & 
						\frac{\textstyle 1}
						{\textstyle 2}
						\cdot
						(p_{1}(1) + p_{1}(2))
						+
						\frac{\textstyle 1}
						{\textstyle 2}
						\cdot
						(p_{1}(1) + p_{1}(3))\ \
						=\ \
						p_{1}(1)
						+
						\frac{\textstyle 1}
						{\textstyle 2}
						(p_{1}(2) + p_{1}(3))\ \
						>\ \
						\frac{1}
						{2}\, .        
						\end{eqnarray*}
					}
					\noindent
					\textcolor{magenta}{Since ${\mathsf{V}}$
						is a unimodal valuation,
						it follows that
						${\mathsf{V}}_{2}(p_{1}, p_{2})
						=
						{\mathsf{F}}(x)$,
						with $x > \frac{\textstyle 1}
						{\textstyle 2}$.}

					\textcolor{magenta}{Now,
						since player $2$ is \cgr{${\mathsf{E}}$-happy}
						in $(p_{1}, p_{2})$,
						${\mathsf{E}}_{2}(p_{1}, p_{2})
						\leq
						{\mathsf{E}}_{2}(p_{1}, p_{2}^{3})$.}
					\textcolor{magenta}{Since ${\mathsf{E}}(x)$
						is strictly monotone decreasing in $x$,
						it follows that
						$x_{2}(p_{1}, p_{2})
						\geq
						x_{2}(p_{1}, p_{2}^{3})$.}
					\textcolor{magenta}{Since ${\mathsf{V}}$
						is a unimodal valuation,
						it follows that
						${\mathsf{V}}_{2}(p_{1}, p_{2}^{3})
						=
						{\mathsf{F}}(x(p_{1}, p_{2}^{3}))$
						with     
						$x_{2}(p_{1}, p_{2})
						\geq
						x_{2}(p_{1}, p_{2}^{3})$.}
					\textcolor{magenta}{Since ${\mathsf{F}}$ is a unimodal valuation 
						with ${\mathsf{F}}\left( \frac{\textstyle 1}
						{\textstyle 2}
						\right)
						\leq
						b$,
						it follows that
						${\mathsf{F}}(p_{1}, p_{2})
						\leq
						{\mathsf{F}}(p_{1}, p_{2}^{3})$,
						and player $2$ cannot ${\mathsf{F}}$-improve.}                                                                                                                                                                                                                                                                                                                                                                                                                                                                                                                                                                                                                                                                                                                                                                                                                                                                                                                                                                                                                                                                                                                                                                                                                                                                                                                                                                                                                                                                                                                                                                                                                                                                                                                                                                                                                                                                                                                                                                                                                                                                                                                                                                                                                                                                                                                                                                                                                                                                                                                                                                                                                                                                                                                                                                                                                                                                                                                                                                                                                                                                                                                                                                                                                                                                                                                                                                                                                                                                                                                                                                                                                                                                                                                                                                                                                                                                                                                                                                                                                                                                                                                                                                                                                 
				}
				By assumption, $(p_1,p_2)$ is an $\mathsf{E}$-Equilibrium.
				This implies that there is some $i\in[3]$ with $\mu(i,3)=b$. Define a new
				probability distribution $\widehat{p}_1$ for player $1$ by 
				$\widehat{p}_1(k)=1/3$, for all $k\in[3].$
				Then, $V_2(\widehat{p}_1,p_2)=\mathsf{F}\left(\frac{\textstyle 2}
				{\textstyle 3}\right)$ and  $V_2(\widehat{p}_1,p^{3}_2)=\mathsf{F}(x)$ with some 
				$x\le 2/3$. As in case~{\sf (B/1/iii)}, {and by observation~(\ref{Fgreater})}, we conclude
				that player $2$ cannot $\mathsf{F}$-improve by switching to strategy $3$.

				\item[{\sf (B/2/ii,iii)}]                                                                                                                                                                                                                                                                                                                                                                                                                                                                                                                                                                                                                                                                                                                                                                                                                              
				We consider two subcases:
				\begin{itemize}
				\item                                                                                                                                                                                                                                                                                                                                                                                                                                                                                                                                                                                                                                                                                                                                                                                                                                                                                                                                                                                                                                                                                                                                                                                                                                                                                                                                                                                                                                                                                                                                                                                                                                                                                                                                                                                                                                                                                                                                                                                                                                                                                                                                                                                                                                                                                                                                                                                                                                                                                                                                                                                                                                                                                                                                                                                                                                                                                                                                                                                                                                                                                                                                                                                                                                                                                                                                                                                                                                                                                                                                                                                                                                                                                                                                                                                                                                                                                                                                                                                                                                                                         ${\mathsf{\mu}}_{2}
				=
				\left( \begin{array}{lll}
					a & b & \ast \\
					a & b & \ast \\                                                                                                                                                                                                                                                                                                                                                                                                                                                                                                                                                                                                                                                                                                                                                                                                                                                                                                                                                                                                                                                                                                                                                                                                                                                                                                                                                                                                                                                                                                                                                                                                                                                                                                                                                                                                                                                                                                                                                                                                                                                                                                                                                                                                                                                                                                                                                                                                                                                                                                                                                                                                                                                                                                                                                                                                                                                                                                                                                                                                                                                                                                                                                                                                                                                                                                                                                                                                            
					b & a & \ast
				\end{array}
				\right)$:
				\remove{ 
					\textcolor{magenta}{Then,}
					\textcolor{magenta}{
						\begin{eqnarray*}
						a \cdot (p_{1}(1) + p_{1}(2))
						+ 
						b \cdot  p_{1}(3)                                                                                                                                                                                                                                                                                                                                                                                                                                                                                                                                                                                                                                                                                                                                                                                                                                                                                                                                                                                                                                                                                                                                                                                                                                                                                                                                                                                                                                                                                                                                                                                                                                                                                                                                                                                                                                                                                                                                                                                                                                                                                                                                                                                                                                                                                                                                                                                                                                                                                                                                                                                                                                                                                                                                                                                                                                                                                                                                                                                                                                                                                                                                                                                                                                                                                                                                                                                                            
						& = & b \cdot (p_{1} (1) + p_{1}(2))
						+
						a \cdot p_{1}(3)\, ,                                                                                                                                                                                                                                                                                                                                                                                                                                                                                                                                                                                                                                                                                                                                                                                                                                                                                                                                                                                                                                                                                                                                                                                                                                                                                                                                                                                                                                                                                                                                                                                                                                                                                                                                                                                                                                                                                                                                                                                                                                                                                                                                                                                                                                                                                                                                                                                                                                                                                                                                                                                                                                                                                                                                                                                                                                                                                                                                                                                                                                                                                                                                                                                                                                                                                                                                                                                                             
						\end{eqnarray*}
					}
					\textcolor{magenta}{which implies that
						$p_{1}(1) + p_{1}(2) = p_{1}(3) = \frac{\textstyle 1}
						{\textstyle 2}$. Thus,
						$x_{2}(p_{1}, p_{2}) 
						=
						\frac{\textstyle 1}
						{\textstyle 2}$.}
				}
				The {\sf WEEP} for player~$2$ implies $p_1(1)+p_1(2)=p_1(3)=1/2$. 
				Thus, $\mathsf{V}_2(p_1,p_2)=\mathsf{F}\left(\frac{\textstyle 1}{\textstyle 2}\right)$.
				If $\mu_2(3,3)=a$, then $(p_1,p_2)$ being an {\sf E}-equilibrium implies that
				$\mu_2(1,3)=\mu_2(2,3)=b$. Thus, $(p_1,p_2)$ is also an {\sf F}-equilibrium.
				
				If $\mu_2(3,3)=b$, then $\mathsf{V}_2(p_1,p^{3}_2) = \mathsf{F}(x)$ with some $x\leq 1/2$
				and hence, {by observation~(\ref{Fgreater}),} player~$2$ cannot {\sf F}-improve by switching to strategy $3$.

				\item                                                                                                                                                                                                                                                                                                                                                                                                                                              
				{${\mathsf{\mu}}_{2}
					=
					\left( \begin{array}{lll}
						a & b & \ast     \\
						b & a & \ast \\                                                                                                                                                                                                                                                                                                                                                                                                                                                                                                                                                                                                                                                                                                                                                                                                                                                                                                                                                                                                                                                                                                                                                                                                                                                                                                                                                                                                                                                                                                                                                                                                                                                                                                                                                                                                                                                                                                                                                                                                                                                                                                                                                                                                                                                                                                                                                                                                                                                                                                                                                                                                                                                                                                                                                                                                                                                                                                                                                                                                                                                                                                                                                                                                                                                                                                                                                                                                            
						b & a & \ast
					\end{array}
					\right)$:} This case is equivalent to the above case, by interchanging the first and second columns and the first and third rows. 
				\remove{ 
					\textcolor{magenta}{Then,}
					\textcolor{magenta}{
						\begin{eqnarray*}
						a \cdot p_{1}(1) 
						+ 
						b \cdot  (p_{1}(2) + p_{1}(3))                                                                                                                                                                                                                                                                                                                                                                                                                                                                                                                                                                                                                                                                                                                                                                                                                                                                                                                                                                                                                                                                                                                                                                                                                                                                                                                                                                                                                                                                                                                                                                                                                                                                                                                                                                                                                                                                                                                                                                                                                                                                                                                                                                                                                                                                                                                                                                                                                                                                                                                                                                                                                                                                                                                                                                                                                                                                                                                                                                                                                                                                                                                                                                                                                                                                                                                                                                                                            
						& = & b \cdot p_{1} (1)
						+
						a \cdot (p_{1}(2) + p_{1}(3))\, ,                                                                                                                                                                                                                                                                                                                                                                                                                                                                                                                                                                                                                                                                                                                                                                                                                                                                                                                                                                                                                                                                                                                                                                                                                                                                                                                                                                                                                                                                                                                                                                                                                                                                                                                                                                                                                                                                                                                                                                                                                                                                                                                                                                                                                                                                                                                                                                                                                                                                                                                                                                                                                                                                                                                                                                                                                                                                                                                                                                                                                                                                                                                                                                                                                                                                                                                                                                                                             
						\end{eqnarray*}
					}
					\textcolor{magenta}{which implies that
						$p_{1}(1) = p_{1}(2) + p_{1}(3) = \frac{\textstyle 1}
						{\textstyle 2}$. Thus,
						$x_{2}(p_{1},p_{2})
						=
						\frac{\textstyle 1}
						{\textstyle 2}$.}
				}

				\end{itemize}                                                                                                                                                           

				\item[{\sf (B/2/iv)}]
				{${\mathsf{\mu}}_{2}
					=
					\left( \begin{array}{lll}
					a & b & \ast     \\
					b & a & \ast \\                                                                                                                                                                                                                                                                                                                                                                                                                                                                                                                                                                                                                                                                                                                                                                                                                                                                                                                                                                                                                                                                                                                                                                                                                                                                                                                                                                                                                                                                                                                                                                                                                                                                                                                                                                                                                                                                                                                                                                                                                                                                                                                                                                                                                                                                                                                                                                                                                                                                                                                                                                                                                                                                                                                                                                                                                                                                                                                                                                                                                                                                                                                                                                                                                                                                                                                                                                                                            
					b & b & \ast
					\end{array}
					\right)$:} 
				The {\sf WEEP} for player~$2$ implies $p_1(1)=p_1(2)$. Since $(p_1,p_2)$ is an {\sf E}-equilibrium, 
				then $\mu_2(i,3)=a$ for at most one value of $i\in [3]$. Define a new probability distribution $\widehat{p}_1$ for 
				player~$1$. We distinguish two cases:                                                                                                                                                           
				
				\begin{itemize}
					\item If $\mu_2(i,3)=a$ for exactly one value $i\in [3]$ or if $\mathsf{F}\left(\frac{\textstyle 1}{\textstyle 3}\right)
					\le b$, then $\widehat{p}_1=1/3$ for all $k\in [3]$. 
					Then $\mathsf{V}_2(\widehat{p}_1,p_2)=\mathsf{F}\left(\frac{\textstyle 1}{\textstyle 3}\right)$
					and player~$2$ cannot {\sf F}-improve by choosing strategy~$3$.

					\item If $\mu_2(i,3)=b$ for all $i\in [3]$ and $\mathsf{F}\left(\frac{\textstyle 1}{\textstyle 3}\right)
					> b$, then set $\sigma(\widehat{p}_1)=\{1,2\}$ and $\widehat{p}_1(1)=\widehat{p}_1(2)=1/2$.
					Player~$1$ cannot {\sf F}-improve since $\mathsf{V}_1(p^{i}_1,p_2)$ has the same value for all $i\in [3]$,
					and player~$2$ cannot {\sf F}-improve since 
					$\mathsf{V}_2(\widehat{p}_1,p_2)= \mathsf{F}\left(\frac{\textstyle 1}{\textstyle 2}\right) \le b$
					and $\mathsf{V}_2(\widehat{p}_1,p^{3}_2)= \mathsf{F}(0)=b$.
				\end{itemize}

		\end{enumerate}                                                                                                                                                                                                                                                                                                                                                                                                                                                        \end{enumerate}                                                                                                                                                                                                                                         
		
	\end{enumerate}
	The claim follows.   \qed                                    
\end{proof}

\noindent
{Now observe that by Theorem~\ref{grande resort} we conclude that  
	the $2$-players, $2$-values and $3$-strategies normal game
	with bimatrix}
{
	\small
	{
		\begin{eqnarray*}
			{\mathsf{C}}_{2}
			& = & \left( \begin{array}{lll}
				(a,b) & (b,a) & (a,b) \\
				(b,a) & (a,b) & (b,b) \\
				(b,b) & (b,b) & (b,a)\\ 
			\end{array}
			\right)\, ,            
		\end{eqnarray*}
	}
}

\noindent has no ${\mathsf{F}}$-equilibrium
	when
	${\mathsf{F}}\left( \frac{1}
	{2}
	\right)
	>
	b$.       
{This complements the result of Theorem~\ref{2 values 3 strategies existence}. In conclusion,
	for 2-players, 2-values games with 3 strategies we have:}
\begin{corollary}
	\label{3strategies}
	{For a unimodal valuation {\sf F} the following properties hold:}
	\begin{enumerate}[$(i)$]
		\item {If $\mathsf{F}\left(\frac{1}{2}\right) \le b$, then every
			2-player, 2 values, 3-strategies game has an {\sf F}-equilibrium.}
		
		\item {If $\mathsf{F}\left(\frac{1}{2}\right) > b$, then there exists
			a normal 2-player, 2 values, 3-strategies game without {\sf F}-equilibrium.}
		
	\end{enumerate}
\end{corollary}

\section{{Conclusion}}
\label{sec:conclusion}

\cg{In this work we have investigated the (in)existence of equilibria for $2$-players, $2$-values games under unimodal valuations. Our work is the first to adopt the combination of abstract settings of unimodal valuations and $2$-values sparse bimatrix games, such as normal games. Normal games are a class of counterexample games, which provide a canonical way of the sparsest bimatrix games. Unimodal valuations provide the simplest possible way of specifying a non-monotone valuation function. Combining these two critical parameters enables the simplest modeling of sparsity and payoff concavity.}

The most striking open problem arising from the context of this paper
is the special role of the condition $\mathsf{F}(\frac{1}{2})=b$,
documented for normal games by Theorem~\ref{panorama} and 
Corollary~\ref{cor:final-result}. Will
$\mathsf{F}(\frac{1}{2})=b$ guarantee the existence of {\sf F}-equilibria
also under relaxed normality, e.g., by allowing two $a$'s per row
and column, respectively? Winning pairs are not necessarily 
{\sf F}-equilibria in this case.

Another open problem is the complexity of computing {\sf F}-equilibria
for $2$-values games. Very likely the problem is ${\mathcal NP}$-hard.
The reductions given in~\cite{MM16,MM17} rely on the Crawford game and
do not care about sparsity or the number of values. The Crawford game
has $3$ values. Reduction for $2$-values games could rely on the games
$\mathsf{C}_{\n},~\n\ge 2$.

%
%


\begin{thebibliography}{}


\bibitem{AKV05}
T.~G.~Abbott, D.~M.~Kane and P.~Valiant,
``On the Complexity of
  Two-Player Win-Lose Games,''
{\it Proceedings of the
     46th Annual IEEE Symposium
     on Foundations of Computer Science,}
pp.\ 113--122,
October 2005.




\bibitem{BM12}
V.~Bil\`{o} and M.~Mavronicolas,
``The Complexity of Decision Problems
about Nash Equilibria
in Win-Lose Games,''
{\it Proceedings of the 5th International Symposium
      on Algorithmic Game Theory,}
pp.\ 37--48,
Vol.~7615,
Lecture Notes in Computer Science,
Springer-Verlag,
October 2012.

\remove{
\bibitem{BM14}
V.~Bil\`{o} and M.~Mavronicolas,
``Complexity of Rational and Irrational Nash Equilibria,''
{\it Theory of Computing Systems,}
Vol.~54,
No.~3,
pp.\ 491--527,
April 2014.

      
\bibitem{BM16}
V.~Bil\`{o} and M.~Mavronicolas,
``A Catalog of $\exists\mathbb{R}$-Complete Decision Problems About Nash Equilibria
in Multi-Player Games,''
{\it Proceedings of the 33rd Symposium on Theoretical Aspects of Computer Science,}
pp.\ 17:1--17:13,
LIPIcs,
February 2016.


\bibitem{BM17}
V.~Bil\`{o} and M.~Mavronicolas,
``$\exists\mathbb{R}$-Complete Decision Problems about Symmetric Nash Equilibria in Symmetric Multi-Player Games,''
{\it Proceedings of the 34th Symposium on Theoretical Aspects of Computer Science,}
pp.\ 13:1--13:14,
LIPIcs, March 2017.
}


\bibitem{CSS16}
R.~L.~G.~Cavalcante, Y.~Shen and S.~Sta\'{n}czak,
``Elementary Properties of Positive Concave Mappings With Applications to Network Planning and Optimization,''
{\it IEEE Transactions on Signal Processing,} Vol.~64, Issue~7, 
pp.\ 1774--1783, April 2016.

\bibitem{CDT06}
{X.~Chen, X.~Deng and S.-H.~Teng,
``Sparse Games are Hard,''
{\it Proceedings of the
      2nd International Workshop
      on Internet and Network Economics,}
pp.\ 262--273,
Vol.~4286,
Lecture Notes in Computer Science,
Springer-Verlag,
December 2006.}
      



\bibitem{CLR06}
{B.~Codenotti, M.~Leoncini and G.~Resta,
``Efficient Computation of Nash Equilibria 
    for Very Sparse Win-Lose Bimatrix Games,''
{\it Proceedings of the 14th Annual European Symposium
     on Algorithms,}
pp.\ 232-243,
Vol.~4168,
Lecture Notes in Computer Science,
Springer-Verlag,
September 2006.}




\bibitem{C90}
V.~P.~Crawford,
``Equilibrium Without Independence,''
{\it Journal of Economic Theory,}
Vol.~50,
No.~1,
pp.\ 127--154,
February 1990.





\bibitem{FP10}
A.~Fiat and C.~H.~Papadimitriou,
``When the Players Are Not
    Expectation-Maximizers,''
{\it Proceedings of the 3rd International Symposium
      on Algorithmic Game Theory,}
pp.\ 1--14,
Vol.~6386,
Lecture Notes in Computer Science,
Springer-Verlag,
October 2010.
          




\bibitem{HL69}
G.~Hanoch and H.~Levy,
``The Efficiency Analysis of Choices Involving Risk,''
{\it  The Review of Economic Studies},
Vol.~36, No.~3,
pp. 335--346,
July 1969.

\bibitem{LS18}
Z.~Lin and Y.~Sheng,
``On the Approximation of Nash Equilibria in Sparse Win-Lose Games,''
{\it Proceedings of the 32nd AAAI Conference
	on Artificial Intelligence,}
pp.\ 1154--1160, 2018.


\bibitem{M52}
H.~Markowitz,
``Portfolio Selection,''
{\it Journal of Finance,}
Vol.~7,
No.~1,
pp.\ 77--91,
March 1952.





\bibitem{MM15}
M.~Mavronicolas and B.~Monien,
``Minimizing Expectation Plus Variance,''
{\it Theory of Computing Systems,}
Vol.~57,
No.~3,
pp.\ 617--654,
October 2015.





\bibitem{MM16}
M.~Mavronicolas and B.~Monien,
``The Complexity of Equilibria
  for Risk-Modeling Valuations,''
{\it Theoretical Computer Science,}
Vol.~634,
pp.\ 67--96,
June 2016.




\bibitem{MM17}
M.~Mavronicolas and B.~Monien,
``Conditional Value-at-Risk:
    Structure and Complexity of Equilibria,''
{\it Proceedings of the
      10th International Symposium
      on Algorithmic Game Theory,}
pp.\ 131--143,
Vol.~10504,
Lecture Notes in Computer Science,
Springer-Verlag,
September 2017. 

\bibitem{PR15}
M.~Pemberton and N.~Rau,
{\it Mathematics for Economists: An Introductory Textbook},
Manchester University Press, 4th edition,  October 2015.


\bibitem{RU02}
R.~T.~Rockafellar and S.~Uryasev,
``Conditional Value-at-Risk for General Loss Distributions,''
{\it Journal of Banking and Finance,}
Vol.~26, 
pp.\ 1443--1471,
2002.


\bibitem{SZ07}
Z.~Shao and H.~Zhou,
``Optimal Transportation Network with Concave Cost Functions: Loop Analysis and Algorithms,''
{\it Physical Review E}, Vol.~75, Issue~6, 5 pages, June 2007. 

\bibitem{S63}
W.~F.~Sharpe,
``A Simplified Model for Portfolio Analysis,''
{\it Management Science,}
Vol.~9, No.~2,
pp.\ 277--293,
January 1963.



\end{thebibliography}
\end{document}